\documentclass[11pt]{article}

\usepackage{times}
\usepackage{amsmath,amssymb,amsthm,amsfonts,latexsym,bbm,xspace,graphicx,float,mathtools,dsfont}
\usepackage[backref,colorlinks,citecolor=blue,bookmarks=true]{hyperref}
\usepackage[dvips, paper=letterpaper, top=1in, bottom=.75in, left=0.9in, right=0.95in, nohead, includefoot, footskip=.25in]{geometry}
\usepackage{color}
\usepackage{comment}
\usepackage{algorithm}
\usepackage{algorithmic}
\newcommand{\ind}{\mathds{1}}
\usepackage{tweaklist}
\usepackage{appendix}
\usepackage{cleveref}

\usepackage{accents}
\newcommand{\ubar}[1]{\underaccent{\bar}{#1}}

\newtheorem{theorem}{Theorem}
\newtheorem{corollary}{Corollary}
\newtheorem{lemma}{Lemma}

\newtheorem{definition}{Definition}

\newtheorem{example}{Example}

\newcommand{\single}{\textsc{Single}}

\newcommand{\core}{\textsc{Core}}
\newcommand{\tail}{\textsc{Tail}}
\newcommand{\rev}{\textsc{Rev}}

\newcommand{\copies}{\textsc{OPT}^{\textsc{Copies-UD}}}

\newcommand{\srev}{\textsc{SRev}}
\newcommand{\brev}{\textsc{BRev}}

\newcommand{\prev}{\textsc{PostRev}}

\newcommand{\bbeta}{\boldsymbol{\beta}}
\newenvironment{prevproof}[2]{\noindent {\em {Proof of {#1}~\ref{#2}:}}}{$\Box$\vskip \belowdisplayskip}

\newcommand{\R}{\ensuremath{\mathbb{R}}} 

\newcommand{\poly}{{\rm poly}}
\newcommand{\opt}{\text{OPT}}

\newcommand{\yangnote}[1]{{\color{blue}{#1}}}

\newcommand{\notshow}[1]{{}}

\DeclareMathOperator{\E}{E}

\def \DD  {{\mathcal{D}}}
\def \KK  {{\mathcal{K}}}
\def \AA  {{\mathcal{A}}}
\def \CC {{\mathcal{C}}}
\def \HH {{\mathcal{H}}}
\def \XX {{\mathcal{X}}}

\def \LL  {{\mathcal{L}}}
\def \BB  {{\cal B}}

\def \EE  {{\mathcal{E}}}

\def \FF  {{\mathcal{F}}}
\def \SS  {{\mathcal{S}}}

\def \E  {{\mathbb{E}}}

\DeclareMathOperator{\argmax}{argmax}

\definecolor{MyGray}{rgb}{0.8,0.8,0.8}

\title{Learning Multi-item Auctions with (or without) Samples}

\author{Yang Cai\\
McGill University, Canada\\
cai@cs.mcgill.ca \and Constantinos Daskalakis\\ EECS and CSAIL, MIT, USA\\ costis@csail.mit.edu}

\begin{document}
\maketitle

\begin{abstract}
We provide algorithms that learn simple auctions whose revenue is approximately optimal in multi-item multi-bidder settings, for a wide range of bidder valuations including unit-demand, additive, constrained additive, XOS, and subadditive. We obtain our learning results in two settings. The first is the commonly studied setting where sample access to the bidders' distributions  over valuations is given, for both regular distributions and arbitrary distributions with bounded support. Here, our algorithms require polynomially many samples in the number of items and bidders. The second is a more general max-min learning setting that we introduce, where we are given ``approximate distributions,'' and we seek to compute a mechanism whose revenue is approximately optimal simultaneously for all ``true distributions'' that are close to the ones we were given. These results are more general in that they imply the sample-based results, and are also applicable in settings where we have no sample access to the underlying distributions but  have estimated them indirectly via market research or by observation of bidder behavior in previously run, potentially non-truthful auctions. 

All our results hold for valuation distributions satisfying the standard (and necessary) independence-across-items property. They also generalize and improve upon recent works of Goldner and Karlin~\cite{GoldnerK16} and Morgenstern and Roughgarden~\cite{MorgensternR16}, which have provided algorithms that learn approximately optimal multi-item mechanisms in more restricted settings with additive, subadditive and unit-demand valuations using sample access to distributions. We generalize these results to the complete unit-demand, additive, and XOS setting, to i.i.d. subadditive bidders, and to the max-min setting. 

Our results are enabled by new uniform convergence bounds for hypotheses classes under product measures. Our bounds result in exponential savings in sample complexity compared to bounds derived by bounding the VC dimension, and are of independent interest.
\end{abstract}
\thispagestyle{empty}
\addtocounter{page}{-1}
\newpage

\section{Introduction} \label{sec:intro}

The design of revenue-optimal auctions is a central problem in Economics and Computer Science, which has found myriad applications in online and offline settings, ranging from sponsored search and online advertising to selling artwork by auction houses, and public goods such as drilling rights and radio spectrum by governments. The problem involves a seller who wants to sell one or several items to one or multiple strategic bidders with private valuation functions, mapping each bundle of items they may receive to how much value they derive from the bundle. As no meaningful revenue guarantee can possibly be achieved without any information about the valuations of the bidders, the problem has been classically studied under {\em Bayesian assumptions,} where a joint distribution from which all bidders' valuations are drawn is common knowledge, and the goal is to maximize  revenue in expectation with respect to this distribution. 

In the {\em single-item setting}, Bayesian assumptions have enabled beautiful and influential developments in auction theory. Already 36 years ago, a breakthrough result by Myerson identified the optimal single-item auction when bidder values are independent~\cite{Myerson81}, and the ensuing decades saw a great deal of further understanding and practical applications of single-item auctions, importantly in online settings. 

However, the quest for optimal {\em multi-item auctions} has been quite more challenging. It has been recognized that revenue-optimal multi-item auctions can be really complex, may exhibit counter-intuitive properties, and be fragile to changes in the underlying distributions; for a discussion and examples see survey~\cite{Daskalakis15}. As such, it is doubtful that there is a crisp characterization of the structure of optimal multi-item auctions, at least not beyond single-bidder settings~\cite{DaskalakisDT17}. On the other hand, there has been significant recent progress in efficient computation of revenue-optimal auctions~\cite{ChawlaHK07,ChawlaHMS10,Alaei11,CaiD11b,AlaeiFHHM12,CaiDW12a,CaiDW12b,CaiH13,CaiDW13b,AlaeiFHHM13,BhalgatGM13,DaskalakisDW15}. Importantly, this progress has enabled identifying {\em simple auctions} (mostly variations of sequential posted pricing mechanisms) that achieve constant factor approximations to the revenue of the optimum~\cite{BabaioffILW14,Yao15,CaiDW16,ChawlaM16, CaiZ17}, under the {\em item-independence} assumption of Definition~\ref{def:subadditive independent} and Example~\ref{eg:valuation}. These auctions are {\em way simpler} than the optimum, and have {\em strong incentive properties}: they are dominant strategy truthful, while still competing against the optimal Bayesian truthful mechanism. The current state-of-the-art is given as Theorem~\ref{thm:simple XOS}, which applies to bidders with valuation functions from the broad class of fractionally subbaditive (a.k.a.~XOS) valuations, which contains submodular.

\medskip As our discussion illustrates, studying auctions assuming Bayesian priors has been quite fruitful, enabling us to identify guiding principles for how to structure auctions to achieve optimal (in single-item settings) or approximately optimal (in multi-item settings) revenue. To apply this theory to practice, however, one needs knowledge of the underlying distributions. Typically, one would estimate these distributions via market research or by observations of bidder behavior in prior auctions, then use the estimated distributions to design a good auction. However, estimation  involves approximation, and the performance of mechanisms can be quite fragile to errors in the distributions. This motivates studying whether optimal or approximately optimal auctions can be identified when one has imperfect knowledge of the true distributions.

With this motivation, recent work in Computer Science has studied whether approximately optimal mechanisms can be ``learned'' given sample access to the underlying distributions. This work has lead to an almost complete picture for the single-item (and the more general single-parameter) setting where Myerson's theory applies, showing how near-optimal mechanisms can be learned from polynomially many (in the approximation and the number of bidders) samples~\cite{Elkind07,ColeR14,MohriM14,HuangMR15,MorgensternR15,DevanurHP16,RoughgardenS16,GonczarowskiN16}. 

On the multi-item front, however, where the analogue of Myerson's theory is elusive, and unlikely, our understanding is much sparser. Recent work of Morgenstern and Roughgarden~\cite{MorgensternR16} has taken a computational learning theory approach to identify the sample complexity required to optimize over classes of simple auctions. Combined with the afore-described results on the revenue guarantees of simple auctions, their work leads to algorithms that learn approximately optimal auctions in multi-item settings with multiple unit-demand bidders, or a single subadditive bidder, from polynomially many samples in the number of items and bidders. These results apply to distributions satisfying the {\em item-independence} assumption of Definition~\ref{def:subadditive independent} and Example~\ref{eg:valuation}, under which the approximate optimality of simple auctions has been established.
 
While well-suited for identifying the sample complexity required to optimize over a class of simple mechanisms, which is a perfectly reasonable goal to have but not the one in this paper, the approach taken in~\cite{MorgensternR16} is arguably imperfect towards proving polynomial sample bounds for learning approximately optimal auctions in the settings where simple mechanisms are known to perform well in the first place. This is due to the following discordance: (i) On the one hand, simple and approximately optimal mechanisms in multi-item settings are mostly only known under item-independence. (ii) On the other hand, the computational learning techniques employed in~\cite{MorgensternR16}, and in particular bounding the {\em pseudo-dimension} of a class of auctions, are not fine enough to discern the difference in sample complexity required to optimize under item-independence and without item-independence. As such, this technique can only obtain polynomial sample bounds for approximate revenue optimization if it so happens that a class of mechanisms is both learnable from polynomially-many samples under arbitrary distributions, and it guarantees approximately optimal revenue under item-independence, or for some other interesting class of distributions.\footnote{It is known that some restriction {\em needs to be made} on the distribution to gain polynomial sample complexity, as otherwise exponential lower bounds are known for learning approximately optimal auctions even for a single unit-demand bidder~\cite{DughmiHN14}.} 

In particular, bounding the pseudo-dimension of classes of auctions as a means to prove polynomial-sample bounds for approximate revenue optimization hits a barrier even for multiple additive bidders with independent values for items. In this setting, the approximately optimal auctions that are known are the best of selling the items separately or running a VCG mechanism with entry fees~\cite{Yao15,CaiDW16}, as described in Section~\ref{sec:additive}. Unfortunately, the latter can easily be seen to have pseudo-dimension that is exponential in the number of bidders, thus only implying a sufficient exponentially large sample size to optimize over these mechanisms. Is this exponential sample size really necessary or an artifact of the approach? Recent work of Goldner and Karlin~\cite{GoldnerK16} gives us hope that it is the latter. They show how to learn approximately optimal auctions in the multi-item multi-bidder setting with additive bidders using only one sample from each bidder's distribution, assuming that it is {\em regular} and independent across items. 

\vspace{-10pt}\paragraph{Our results.} We show that simple and approximately optimal mechanisms are learnable from polynomially-many samples for multi-item multi-bidder settings, whenever:
\begin{itemize} 
\item the bidder valuations are fractionally subadditive (XOS), i.e. we can accommodate additive, unit-demand, constrained additive, and submodular valuations;
\item the distributions over valuations satisfy the standard item-independence assumption of Definition~\ref{def:subadditive independent} and Example~\ref{eg:valuation}, and their single-item marginals are arbitrary and bounded, or (have arbitrary supports but are) regular.\footnote{We note again that without the standard item-independence (or some other) restriction on the distributions, we cannot hope to learn approximately optimal auctions from sub-exponentially many samples, even for a single unit-demand bidder~\cite{DughmiHN14}.}
\end{itemize}
In particular, our results constitute vast extensions of known results on the polynomial learnability of approximately optimal auctions in multi-item settings~\cite{MorgensternR16,GoldnerK16}. Additionally we show that:
\begin{itemize}
\item whenever the valuations are additive and unit-demand, or whenever the bidders are symmetric and have XOS valuations, our approximately optimal mechanisms can be identified from polynomially many samples and in polynomial time;
\item whenever the bidders are symmetric (i.e.~their valuations are independent and identically distributed)  and have {\em subadditive valuations}, we can compute from polynomially many samples and in polynomial-time a simple mechanism whose revenue is a $\Omega\left({n \over \max\{m,n\}}\right)$-fraction of the optimum, where $m$ and $n$ are respectively the number of items and bidders. In particular, if the number of bidders is at least a constant fraction of the number of items, the mechanism is a constant factor approximation; and

\item in the setting of the previous bullet, if the item marginals are regular, our mechanism is {\em prior-independent}, i.e. there is a single mechanism, identifiable without any samples from the distributions, providing the afore-described revenue guarantee.
\end{itemize}
Finally, the mechanisms learned by our algorithms for XOS bidders are either {\em rationed sequential posted price mechanisms} (RSPMs) or {\em anonymous sequential posted price mechanisms with entry fees} (ASPEs) as defined in Section~\ref{sec:constrained additive}. The mechanisms learned for symmetric subadditive bidders are RSPMs. RSPMs maintain a price $p_{ij}$ for every bidder and item pair and, in some order over bidders $i=1,\ldots,n$, they give one opportunity to bidder $i$ to purchase {\em one} item $j$ that has not been purchased yet at price $p_{ij}$. ASPEs maintain one price $p_j$ for every item and, in some order over bidders $i=1,\ldots,n$, they give one opportunity to bidder $i$ to purchase {\em any subset} $S'$ of the items $S$ that have not been purchased yet as long as he also pays an ``entry fee'' that depends on $S$ and the identity of the bidder. See Algorithm~\ref{alg:aspe-mech}. 

\vspace{-10pt}\paragraph{Learning without Samples.} Thus far, our algorithms used {\em samples} from the valuation distributions to identify an approximately optimal and simple mechanism under item-independence. However, having sample access to the distributions may be impractical. Often we can observe the actions used by bidders in non-truthful auctions that were previously run, and use these observations to estimate the distributions over valuations using econometric methods~\cite{GuerrePV00,PaarschH06,AtheyH2007nonparametric}. In fact, it may likely be the case we have never sold all the items together in the past, and only have observations of bidder behavior in non-truthful auctions selling each item separately. Econometric methods would achieve better approximations in this case, but only for the item marginals. Finally, we may want to combine multiple sources of information about the distributions, combining past bidder behavior in several different auctions and with market research data. 

With this motivation in mind, we would like to extend our learnability results beyond the setting where sample access to the valuation distributions is provided. We propose ``learning'' approximately optimal multi-item auctions given distributions that are close to the true distributions under some distribution distance $d(\cdot,\cdot)$. In particular, given approximate distributions $\hat{D}_1,\ldots,\hat{D}_n$ over bidder valuations, we are looking to identify a mechanism $\cal M$ satisfying the following {\em max-min style objective}:
\begin{align}\forall {D_1,\ldots,D_n~\text{s.t.}~d(D_i,\hat{D}_i) \le \epsilon, \forall i}: {\rm Rev}_{\cal M}(D_1,\ldots,D_n) \ge \Omega({\rm OPT}(D_1,\ldots,D_n)) -{\rm poly}({\epsilon},m,n). \label{eq:max min goal}
\end{align}
That is, we want to find a mechanism $\cal M$ whose revenue is within a constant multiplicative and a ${\rm poly}({\epsilon},m,n)$ additive error from optimum, simultaneously in all possible worlds $D_1,\ldots,D_n$, where $d(D_i,\hat{D}_i) \le \epsilon, \forall i$. It is not a priori clear that such a ``one-fits-all'' mechanism actually exists. 

There are several notions of distance $d(\cdot,\cdot)$ between distributions that we could study in the formulation of Goal~\eqref{eq:max min goal}, but we opt for an easy one to satisfy. We only require that we know every bidder's marginal distributions over single-item values to within~$\epsilon$ in Kolmogorov distance;\footnote{{Indeed, Goal~\eqref{eq:max min goal} is achievable only for bounded distributions even in the single-item single-bidder setting. Given any bounded distribution $\hat{D}$, create $D$ by moving $\epsilon$ probability mass in $\hat{D}$ to $+\infty$. It is not hard to see that $D$ and $\hat{D}$ are within $\epsilon$ in Kolmogorov distance, but no single mechanism can satisfy the approximation guarantee for both $D$ and $\hat{D}$ simultaneously. Using a similar argument, we can argue that the additive error has to depend on $H$ which is the upper bound on any bidder's value for a single item. See Section~\ref{sec:prelim} for our formal model.}} see Definition~\ref{def:Kolm and TV}. All that this requires is that the cumulative density functions of the approximating distributions over single-item values is within $\epsilon$ in infinity norm from the corresponding cumulative density functions of the corresponding true distributions. As such, it is an easy property to satisfy. For example, given sample access to any single-item marginal, the DKW inequality~\cite{DvoretzkyKW56} implies that $O(\log(1/\delta)/\epsilon^2)$ samples suffice to learn it to within $\epsilon$ in Kolmogorov distance, with probability at least $1-\delta$. So achieving Goal~\eqref{eq:max min goal} directly also implies polynomial sample learnability of approximately optimal auctions. But a Kolmogorov approximation can also be arrived at by combining different sources of information about the single-item marginals such as the ones described above. Regardless of how the approximations were obtained, the max-min goal outlined above guarantees robustness of the revenue of the identified mechanism $\cal M$ with respect to all sources of error that came into the estimation of the single-item marginal distributions. 

While Goal~\eqref{eq:max min goal} is not a priori feasible, we show how to achieve it in multi-item multi-bidder settings with constrained additive bidders, or symmetric bidders with subadditive valuations, under the standard assumption of item-independence. Our results are polynomial-time in the same cases as our sample-based results discussed above.


\vspace{-10pt}\paragraph{Roadmap and Technical Ideas.} In Section~\ref{sec:uniform convergence under product measure}, we present a new approach for obtaining uniform convergence bounds for hypotheses classes under product distributions; see Theorem~\ref{thm:uniform convergence for product measure PARTITION} and Corollary~\ref{cor:VC for product measure}. We show that our approach can significantly improve the sample complexity bound obtained via traditional methods such as VC theory. In particular, Table~\ref{tab:productVC} compares the sample complexity bounds obtained via our approach to those obtained by VC theory for different classes of hypotheses.

Our results for mechanisms make use of recent work on the revenue guarantees of simple mechanisms, which are mainly variants of sequential posted pricing mechanisms~\cite{CaiDW16,CaiZ17}. Using our results from Section~\ref{sec:uniform convergence under product measure}, in Section~\ref{sec:uniform convergence of SPEM}, we derive uniform convergence bounds for the revenue of a class of mechanisms shown to achieve a constant fraction of optimal revenue when all bidders have valuations that are constrained additive over independent items. These mechanisms are called Sequential Posted Price with Entry Fee Mechanisms, a.k.a.~SPEMs,\footnote{Note that any RSPM or ASPE is an SPEM.}. As a corollary of the uniform convergence of SPEMs, we obtain our sample based results for constrained additive bidders. In fact, we obtain a slightly stronger statement than uniform convergence of the revenue of SPEMs, which also implies our max-min results for constrained-additive bidders; see Theorems~\ref{thm:revenue stability under K-distance} and~\ref{thm:constrained additive Kolmogorov}. In particular, Theorem~\ref{thm:constrained additive Kolmogorov} and the DKW inequality imply the polynomial-sample learnability of approximately revenue-optimal auctions for constrained additive bidders.

Technically speaking, our sample based and max-min approximation results for constrained additive bidders provide a crisp illustration of how we leverage item-independence and our new uniform convergence bounds for product measures to sidestep the exponential pseudo-dimension of the class of mechanisms that we are optimizing over. Let us discuss our max-min results which are stronger. Suppose $D_i=\times_j D_{ij}$ is the true distribution over bidder $i$'s valuation and $\hat{D}_i=\times_j \hat{D}_{ij}$ is the approximating distribution, where $D_{ij}$ and $\hat{D}_{ij}$ are respectively the item $j$ marginals. To argue that the revenue of some anonymous sequential posted price with entry fees (ASPE) mechanism is similar under $D=\times_i D_i$ and $\hat{D}=\times_i \hat{D}_i$, we need to couple in total variation distance the decisions of what sets all bidders buy in the execution of the mechanism under $D$ and $\hat{D}$. The issue that we encounter is that there are exponentially many subsets each bidder may buy, hence the naive use of the Kolmogorov bound $||D_{ij}-\hat{D}_{ij}||_K \le \epsilon$, on each single-item marginal results in an exponential blow-up in the total variation distance of what subset of items bidder $i$ buys, invalidating our desired coupling. To circumvent this challenge, we argue in Lemma~\ref{lem:stable demand set} that the events corresponding to which subset of items each buyer will buy are {\em single-intersecting}, according to Definition~\ref{def:single-intersecting}, when seen as events on the buyer's single-item values. Single-intersecting events may be non-convex and have infinite VC dimension. Nevertheless, because single-item values are independent, our new uniform convergence bounds for product measures (Lemma~\ref{lem:Kolmogorov stable for sc}) imply that the difference in probabilities of any such event under $D$ and $\hat{D}$ is only a factor of $m$, the number of items, larger than the bound $\epsilon$ on the Kolmogorov distance between single-item marginals.

We specialize our results to unit-demand bidders in Section~\ref{sec:unit-demand} to obtain computationally efficient solutions for both max-min and sample-based models. Similarly, Section~\ref{sec:additive} contains our results for additive bidders. 
 We also generalize our sample-based results for constrained additive bidders to XOS bidders in Section~\ref{sec:constrained additive}.  
 Finally, we provide computationally efficient solutions for symmetric XOS and even symmetric subadditive bidders in Section~\ref{sec:symmetric bidders}.
 These results are based on showing that (i) the right parameters of SPEMs can be efficiently and approximately identified with sample  or max-min access to the distributions; and (ii) that the revenue guarantees of simple mechanisms can be robustified to accommodate error in the setting of the parameters. In particular, our sample-based result for unit-demand bidders robustifies the ex-ante relaxation of the revenue maximization problem from~\cite{Alaei11} and its conversion to a sequential posted pricing mechanism from~\cite{ChawlaHMS10}, and makes use of the extreme-value theorem for regular distributions from~\cite{CaiD11b}. Our sample-based result for additive bidders shows how to use samples to design mechanisms that approximate the revenue of Yao's VCG with entry fees mechanism~\cite{Yao15}. Our sample-based results for  XOS bidders show how to use samples to approximate the parameters of the RSPM and ASPE mechanisms of~\cite{CaiZ17}, and argue, by re-doing their duality proofs, that their revenue guarantees are robust to errors in the approximation. Finally, our sample based result for symmetric subadditive bidders is based on a new, duality based, approximation, showing how to eliminate the use of ASPEs from the result of~\cite{CaiZ17}. This even allows us to obtain prior-independent mechanisms when the item marginals are regular.

\section{Preliminaries}\label{sec:prelim}

We focus on revenue maximization in the combinatorial auction with $n$ independent bidders and $m$ heterogenous items. Each bidder has a valuation that is \textbf{subadditive over independent items} (see Definition~\ref{def:subadditive independent}). We denote bidder $i$'s type $t_i$ as $\langle t_{ij}\rangle_{j=1}^m$, where $t_{ij}$ is bidder $i$'s private information about item $j$. For each $i$, $j$, we assume $t_{ij}$ is drawn independently from the distribution $D_{ij}$. Let $D_i=\times_{j=1}^m D_{ij}$ be the distribution of bidder $i$'s type and $D=\times_{i=1}^n D_i$ be the distribution of the type profile. We use $T_{ij}$ (or $T_i, T$) and $f_{ij}$ (or $f_i, f$) to denote the support and density function of $D_{ij}$ (or $D_i, D$). For notational convenience, we let  $t_{-i}$ to be the types of all bidders except $i$. 
Similarly, we define $D_{-i}$, $T_{-i}$ 
  and $f_{-i}$ for the corresponding  distributions, support sets and density functions. When bidder $i$'s type is $t_i$, her valuation for a set of items $S$ is denoted by $v_i(t_i,S)$. Throughout the paper we use $\opt$ to denote the optimal revenue obtainable by any randomized and Bayesian truthful mechanism.
  
    \begin{definition}~\cite{RubinsteinW15}\label{def:subadditive independent}
For every bidder $i$, whose type $t_i$ is drawn from a product distribution $F_i=\times_j F_{ij}$, her distribution, $\mathcal{V}_i$, over valuation functions $v_i(t_i,\cdot)$ is \textbf{subadditive over independent items} if:

\vspace{.05in}
\noindent \textbf{- $v_i(\cdot,\cdot)$ has no externalities}, i.e., for each $t_i\in T_i$ and $S\subseteq [m]$, $v_i(t_i,S)$ only depends on $\langle t_{ij}\rangle_{j\in S}$, formally, for any $t_i'\in T_i$ such that $t_{ij}'=t_{ij}$ for all $j\in S$, $v_i(t_i',S)=v_i(t_i,S)$.

\vspace{.05in}
\noindent \textbf{- $v_i(\cdot,\cdot)$ is monotone}, i.e., for all $t_i\in T_i$ and $U\subseteq V\subseteq [m]$, $v_i(t_i,U)\leq v_i(t_i,V)$.

\vspace{.05in}
\noindent \textbf{- $v_i(\cdot,\cdot)$ is subadditive}, i.e., for all $t_i\in T_i$ and $U$, $V\subseteq [m]$, $v_i(t_i,U\cup V)\leq v_i(t_i,U)+ v_i(t_i,V)$.

\vspace{.05in}
 We use $V_i(t_{ij})$ to denote $v_i(t_i,\{j\})$, as it only depends on $t_{ij}$. When $v_i(t_i,\cdot)$ is XOS (or constrained additive) for all $i$ and $t_i\in T_i$, we say $\mathcal{V}_i$ is XOS (or constrained additive) over independent items.
\end{definition}

  \begin{example}\label{eg:valuation}\cite{RubinsteinW15}
We may instantiate Definition~\ref{def:subadditive independent} to define restricted families of subadditive valuations as follows. In all cases, suppose $t=\{t_j\}_{j\in[m]}$ is drawn from $\times_j D_j$. To define a valuation function that is:
\vspace{.05in}

\noindent\textbf{- unit-demand}, we can take $t_j$ to be the value of item $j$, and set $v(t,S)=\max_{j\in S}t_j$.

\vspace{.05in}

\noindent \textbf{- additive}, we can take $t_j$ to be the value of item $j$, and set $v(t,S)=\sum_{j\in S}t_j$.

\vspace{.05in}

\noindent \textbf{- constrained additive}, we can take $t_j$ to be the value of item $j$, and set $v(t,S)=\max_{R\subseteq S, R\in \mathcal{I}}\sum_{j\in R} t_j$, for some downward closed set system ${\cal I} \subseteq 2^{[m]}$.

\vspace{.05in}

\noindent \textbf{- XOS (a.k.a. fractionally subadditive),} we can take $t_j=\{t_{j}^{(k)}\}_{k\in[K]}$ to encode all possible values associated with item $j$, and take $v(t,S)=\max_{k\in[K]}\sum_{j\in S}t_{j}^{(k)}$.

\vspace{.05in}
Note that constrained additive valuations contain additive and unit-demand valuations as special cases, and are contained in XOS valuations.
\end{example}

\noindent\textbf{Distribution Access Models}
\vspace{.05in}

\noindent We consider the following three different models to access the distributions.

\notshow{\vspace{.05in}

\noindent\textbf{Sample access to bounded distributions.} We assume that for any buyer $i$ and any type $t_i\in T_i$, her value $V_i(t_{ij})$ for any single item $j$ lies in $[0,H]$. 

\vspace{.05in}

	\noindent \textbf{Sample access to regular distributions.} We assume that for any buyer $i$ and any type $t_i\in T_i$, the distribution of her value $V_i(t_{ij})$ for any item $j$ is regular.
	
	\vspace{.05in}

	\noindent\textbf{Direct access to approximate distributions.} We assume that we have direct access to a distribution $\hat{D}=\times_{i\in[n], j\in[m]} \hat{D}_{ij}$, that is, we can query the pdf, cdf of $\hat{D}$ and take samples from $\hat{D}$. Moreover, for any buyer $i$ and any type $t_i\in T_i$, the distributions of the random variable $V_i(t_{ij})$ when $t_{ij}$ is sampled from $\hat{D}_{ij}$ or  $D_{ij}$ are within $\epsilon$ in Kolmogorov distance, and both distributions  are supported on $[0,H]$.}
{\begin{itemize}
	\item \textbf{Sample access to bounded distributions.} We assume that for any buyer $i$ and any type $t_i\in T_i$, her value $V_i(t_{ij})$ for any single item $j$ lies in $[0,H]$. 
	\item \textbf{Sample access to regular distributions.} We assume that for any buyer $i$ and any type $t_i\in T_i$, the distribution of her value $V_i(t_{ij})$ for any item $j$ is regular.
	\item \textbf{Direct access to approximate distributions.} We assume that we have direct access to a distribution $\hat{D}=\times_{i\in[n], j\in[m]} \hat{D}_{ij}$, for example we can query the pdf, cdf of $\hat{D}$ and take samples from $\hat{D}$. Moreover, for any buyer $i$ and any type $t_i\in T_i$, the distributions of the random variable $V_i(t_{ij})$ when $t_{ij}$ is sampled from $\hat{D}_{ij}$ or  $D_{ij}$ are within $\epsilon$ in Kolmogorov distance, and both distributions  are supported on $[0,H]$. 
\end{itemize}}

\begin{definition} \label{def:Kolm and TV}
The {\em Kolmogorov distance} between two distributions $P$ and $Q$ over $\mathbb{R}$, denoted $|| P-Q||_K$, is defined as $\sup_{x \in \mathbb{R}}|\Pr_{X \sim P}[X \le x]-\Pr_{X \sim Q}[X \le x]|$. The {\em total variation distance} between two probability measures $P$ and $Q$ on a sigma-algebra $\cal F$ of subsets of some sample space $\Omega$, denoted $|| P-Q||_{TV}$, is defined as $\sup_{E \in {\cal F}}|P({E})-Q({E})|$.
\end{definition}

\section{Summary of Our Results}

We summarize our results in the following two tables. Table~\ref{tab:sample-based results} contains all sample-based results and Table~\ref{tab:max-min learning results} contains all results under the max-min learning model. 
\begin{table}[h]
	\centering
		\begin{tabular}{c || c | c| c| c}
			
			\hline\hline
			Valuations & \# bidders &   Distributions  & Approximation  & Sample Complexity \\
						\hline
												&&&&\\

									 additive~\cite{GoldnerK16} &  $n$  & regular & $\Omega(\opt)$ & $1$ \\	
			  additive &  $n$  & arbitrary $[0,H]$ & $\Omega(\opt)-\epsilon\cdot H$ & $\poly(n,m,1/\epsilon)$	\\					\hline
						&&&&\\

			unit-demand~\cite{MorgensternR16} &  $n$ & arbitrary $[0,H]$ & $\Omega(\opt)-\epsilon\cdot H$ & $\poly(n,m,1/\epsilon)$\\
			unit-demand &  $n$ & regular & $\Omega(\opt)$ & $\poly(n,m)$ \\
				\hline
										&&&&\\

			constrained additive &  $n$ &	 arbitrary $[0,H]$ & $\Omega(\opt)-\epsilon\cdot H$ & $\poly(n,m,1/\epsilon)$ \\

			constrained additive &  $n$ &	regular & $\Omega(\opt)$ & $\poly(n,m)$\\
			\hline
										&&&&\\
										XOS &  $n$ &	 arbitrary $[0,H]$ & $\Omega(\opt)-\epsilon\cdot H$ & $\poly(n,m,1/\epsilon)$\\

			XOS &  $n$ &	regular & $\Omega(\opt)$ & $\poly(n,m)$\\
			\hline
										&&&&\\

			subadditive~\cite{MorgensternR16} & $1$ & arbitrary $[0,H]$  &  $\Omega(\opt)-\epsilon\cdot H$ &$\poly(m,1/\epsilon)$\\
						subadditive & $n$ i.i.d. & arbitrary $[0,H]$  &  $\Omega\left(\frac{n}{\max\{n,m\}}\right)\cdot \opt-\epsilon\cdot H$ &$\poly(n,m,1/\epsilon)$\\
						subadditive & $n$ i.i.d. & regular & $\Omega\left(\frac{n}{\max\{n,m\}}\right)\cdot \opt$ & \emph{prior-independent}\\
\hline
		\end{tabular}
		
	\caption{Summary of Our Sample-based Results.}
	\label{tab:sample-based results}
\end{table}

\begin{table}[h]
	\centering
		\begin{tabular}{c || c | c| c}
			
			\hline\hline
			Valuations & \# bidders &   Distributions  & Approximation \\
						\hline
												&&&\\
			  additive &  $n$  & arbitrary $[0,H]$ & $\Omega(\opt)-O(\epsilon\cdot n\cdot m\cdot H)$ 	\\					\hline
						&&&\\

			unit-demand&  $n$ & arbitrary $[0,H]$ & $\Omega(\opt)-O(\epsilon\cdot n\cdot m\cdot H)$  \\
				\hline
										&&&\\

			constrained additive &  $n$ &	 arbitrary $[0,H]$ & $\Omega(\opt)-O(\epsilon\cdot n\cdot m^2\cdot H)$\\

			\hline
										&&&\\

						subadditive & $n$ i.i.d. & arbitrary $[0,H]$  & $\Omega\left(\frac{n}{\max\{n,m\}}\right)\cdot \opt-O(\epsilon\cdot n \cdot m\cdot H)$\\
\hline
		\end{tabular}

	\caption{Summary of Our Max-min Learning Results.}
	\label{tab:max-min learning results}
\end{table}
\section{Uniform Convergence under Product Measures}\label{sec:uniform convergence under product measure}

In this section, we develop machinery for obtaining uniform convergence bounds for hypotheses over product measures. Our goal is to save on the sample complexity implied by VC dimension bounds, as summarized in Table~\ref{tab:productVC}. Indeed, we obtain low sample complexity  bounds for indicators over \emph{single-intersecting} sets (see Definition~\ref{def:single-intersecting}), which play a key role in proving our results for learning approximately revenue-optimal auctions. Our main results of this section are Theorem~\ref{thm:uniform convergence for product measure PARTITION} for general functions, and Corollary~\ref{cor:VC for product measure} for sets.

We first define what type of uniform convergence bounds we seek to prove.

\begin{definition}[$(\epsilon,\delta)$-uniform convergence with respect to proxy measure]
	A hypothesis class $\HH$ of functions mapping domain set $\XX$ to $\mathbb{R}$ has {\em $(\epsilon,\delta)$-uniform convergence with sample complexity $s(\epsilon,\delta)$} iff, for all $\epsilon,\delta > 0$, there exists a {\em processing} ${\cal P}:{\cal X}^{s(\epsilon,\delta)} \rightarrow \Delta(\XX)$ such that for any distribution $\DD \in \Delta(\XX)$ when $k=s(\epsilon,\delta)$: 
	\begin{align}
	\Pr_{z_1,\cdots, z_k \sim \DD}\left[\sup_{g\in \HH}\left|\E_{z\sim {\cal P}(z_1,\cdots, z_k)}[g(z)] - \E_{z\sim \DD}[g(z)]\right|\leq \epsilon\right]\geq 1-\delta. \notag
	\end{align}

When ${\cal X}$ is the Cartesian product of a collection of sets ${\cal X}_1,\ldots,{\cal X}_k$, i.e.~${\cal X} =\times_i {\cal X}_i$, we say that a hypothesis class $\HH$ as above has {\em $(\epsilon,\delta)$-p.m.~uniform convergence with sample complexity $s(\epsilon,\delta)$} if the above holds for all ${\cal D}$ that are product measures over ${\cal X}$.	
\end{definition}

Next we provide a simple lemma, which leads to a simple version of our main result stated as Theorem~\ref{thm:uniform convergence for product measure}. Our main result, stated as Theorem~\ref{thm:uniform convergence for product measure PARTITION}, follows.

%

\begin{lemma} \label{lem:approx product dist}
	Let $\XX_1,\ldots, \XX_d$ be $d$ domain sets and $\HH$ be a hypothesis class with functions mapping from the product space $\times_{i=1}^{d} \XX_i$ to $\mathbb{R}$. For all $i\in[d]$, let $\HH_i$ be the projected hypothesis class of $\HH$ on $\XX_i$, that is, $\HH_i=\left\{g\ |\ \exists f\in \HH\ \exists\ a_{-i}\in \times_{j\neq i} \XX_j \ \forall\ x_i\in \XX_i, \ g(x_i) = f(x_i,a_{-i})\right\}$. For every $i\in[d]$, let $\DD_i$ and $\hat{\DD}_i$ be two distributions supported on $\XX_i$. Suppose for all $i\in[d]$, $$\sup_{g\in\HH_i}\left|\E_{x\sim \DD_i}[g(x)]-\E_{x\sim \hat{\DD}_i}[g(x)]\right|\leq \epsilon,$$ then $$\sup_{f\in\HH} \left|\E_{\boldsymbol{x}\sim \times_{i=1}^d \DD_i}\left[f(\boldsymbol{x})\right]-\E_{\boldsymbol{x}\sim \times_{i=1}^d \hat{\DD}_i}\left[f(\boldsymbol{x})\right]\right|\leq d\cdot \epsilon.$$
\end{lemma} 

\begin{proof}
	Let $\FF_i$ and $\hat{\FF}_i$ be the probability measure function for $\DD_i$ and $\hat{\DD}_i$ respectively. We will prove the statement using a hybrid argument. We create a sequence of product distributions $\{\DD^{(j)}\}_{j\leq d}$, where $\DD^{(j)}=\hat{\DD}_1\times\cdots\times\hat{\DD}_j\times \DD_{j+1}\times\cdots\times \DD_{d},$ and $\DD^{(0)}=\DD$, $\DD^{(d)}=\hat{\DD}$. To prove our claim, it suffices to show that for any integer $j\in[d]$, $$\left|\E_{\boldsymbol{x}\sim \DD^{(j-1)}}\left[f(\boldsymbol{x})\right]-\E_{\boldsymbol{x}\sim \DD^{(j)}}\left[f(\boldsymbol{x})\right]\right|\leq \epsilon.$$
	
	Next, we show how to derive this inequality.
	\begin{align*}
		&\left|\E_{\boldsymbol{x}\sim \DD^{(j-1)}}\left[f(\boldsymbol{x})\right]-\E_{\boldsymbol{x}\sim \DD^{(j)}}\left[f(\boldsymbol{x})\right]\right|\\
		=&\Bigg{|}\int_{\times_{i\neq j} \XX_i}\left(\int_{\XX_j}f(x_{j},x_{-j}) d \FF_j(x_j)\right) d\hat{\FF}_1(x_1)\cdots d\hat{\FF}_{j-1}(x_{j-1}) d\FF_{j+1}(x_{j+1})\cdots d\FF_{d}(x_{d})\\
		 & ~~~~~~~~~~~~~~~~~~~~ - \int_{\times_{i\neq j} \XX_i}\left(\int_{\XX_j}f(x_{j},x_{-j}) d \hat{\FF}_j(x_j)\right) d\hat{\FF}_1(x_1)\cdots d\hat{\FF}_{j-1}(x_{j-1}) d\FF_{j+1}(x_{j+1})\cdots d\FF_{d}(x_{d})\Bigg{|}\\
		 = & \left|\int_{\times_{i\neq j} \XX_i} \left(\E_{x_j\sim \DD_j}\left[f(x_j,x_{-j})\right]-\E_{x_j\sim \hat{\DD}_j}[f(x_j,x_{-j})]\right)
		  d\hat{\FF}_1(x_1)\cdots d\hat{\FF}_{j-1}(x_{j-1}) d\FF_{j+1}(x_{j+1})\cdots d\FF_{d}(x_{d})\right|\\
		  \leq & \epsilon \cdot\int_{\times_{i\neq j} \XX_i} d\hat{\FF}_1(x_1)\cdots d\hat{\FF}_{j-1}(x_{j-1}) d\FF_{j+1}(x_{j+1})\cdots d\FF_{d}(x_{d})\\
		   = & \epsilon
	\end{align*}
\end{proof}

\begin{theorem}\label{thm:uniform convergence for product measure}
Let $\XX_1,\ldots, \XX_d$ be $d$ domain sets and $\HH$  a hypothesis class of functions mapping from the product space $\times_{i=1}^{d} \XX_i$ to $\mathbb{R}$. For all $i\in[d]$, let $\HH_i$ be the projected hypothesis class of $\HH$ on $\XX_i$, that is $\HH_i=\left\{g\ |\ \exists f\in \HH\ \exists\ a_{-i}\in \times_{j\neq i} \XX_j \ \forall\ x_i\in \XX_i, \ g(x_i) = f(x_i,a_{-i})\right\}$. 

Suppose that, for all $i\in [d]$, $\HH_i$ has $(\epsilon,\delta)$-uniform convergence with sample complexity $s_i(\epsilon,\delta)$.  Then $\HH$ has $(\epsilon,\delta)$-p.m.~uniform convergence with sample complexity $s(\epsilon,\delta) = \max_{i\in[d]} s_i(\epsilon/d,\delta/d)$.

In particular, let $\boldsymbol{z}^{(1)},\ldots, \boldsymbol{z}^{(\ell)}$ be a sample of size $\ell = s(\epsilon,\delta)$ from a product measure $\times_{i\in[d]} \DD_i$. Define $\hat{\DD}_i = {\cal P}_i({z}^{(1)}_i,\ldots,z^{(\ell)}_i)$, for all $i\in[d]$, where $z^{(j)}_i$ is the $i$-th entry of sample $\boldsymbol{z}^{(j)}$ and ${\cal P}_i$ is the processing corresponding to $\HH_i$'s uniform convergence. Then 
$$\Pr_{\boldsymbol{z}^{(1)},\ldots, \boldsymbol{z}^{(\ell)}}\left[\sup_{f\in \HH}\left|\E_{\boldsymbol{z}\sim \times_{i\in[d]}\hat{\DD}_i}\left[f(\boldsymbol{z})\right] - \E_{\boldsymbol{z}\sim \times_{i\in[d]}{\DD}_i}\left[f(\boldsymbol{z})\right]\right|\leq \epsilon\right]\geq 1-\delta.$$
\end{theorem}

\begin{prevproof}{Theorem}{thm:uniform convergence for product measure}
	Since $\ell\geq s_i(\epsilon/d,\delta/d)$, $\Pr\left[\sup_{g\in \HH_i}\left|\E_{z\sim \hat{\DD}_i}[g(z)]- \E_{z\sim \DD_i}[g(z)]\right|\leq \epsilon/d\right]\geq 1-\delta/d$ for all $i\in[d]$. By the union bound, with probability at least $1- \delta$, $\sup_{g\in \HH_i}\left|\E_{z\sim \hat{\DD}_i}[g(z)]- \E_{z\sim \DD_i}[g(z)]\right|\leq \epsilon/d$ for all $i\in[d]$. According to Lemma~\ref{lem:approx product dist}, $\sup_{f\in \HH}\left|\E_{\boldsymbol{z}\sim \times_{i\in[d]}\hat{\DD}_i}\left[f(\boldsymbol{z})\right] - \E_{\boldsymbol{z}\sim \times_{i\in[d]}{\DD}_i}\left[f(\boldsymbol{z})\right]\right|\leq \epsilon$ with probability at least $1- \delta$.
\end{prevproof}

\begin{theorem}\label{thm:uniform convergence for product measure PARTITION}
Let $\XX_1,\ldots, \XX_d$ be $d$ domain sets and $\HH$ a hypothesis class of functions mapping from the product space $\times_{i=1}^{d} \XX_i$ to $\mathbb{R}$. For all $T\subseteq [d]$, let $\HH_T$ be the projected hypothesis class of $\HH$ on $\XX_T \equiv \times_{i \in T}\XX_i$, that is, $\HH_T=\left\{g\ |\ \exists f\in \HH\ \exists\ a_{-T}\in \times_{j\notin T} \XX_j \ \forall\ x_T\in \XX_T, \ g(x_T) = f(x_T,a_{-T})\right\}$. Suppose that, for all $T \subseteq [d]$, $\HH_T$ has $(\epsilon,\delta)$-p.m.~uniform convergence with sample complexity $s_T(\epsilon,\delta)$, and define 
\begin{align}s(\epsilon,\delta) = \min_{\begin{minipage}[h]{3.5cm}\centering $k$,~partitions\\$T_1 \sqcup T_2 \sqcup \ldots \sqcup T_k = [d]$\end{minipage}} \max_{i=1,\ldots,k} s_{T_i}(\epsilon/k,\delta/k). \label{eq:yang}
\end{align}
Then $\HH$ has $(\epsilon,\delta)$-p.m.~uniform convergence with sample complexity $s(\epsilon,\delta)$.

In particular, let $\boldsymbol{z}^{(1)},\ldots, \boldsymbol{z}^{(\ell)}$ be a sample of size $\ell = s(\epsilon,\delta)$ from a product measure $\times_{i\in[d]} \DD_i$. Suppose that the optimum of~\eqref{eq:yang} is attained at $k=\tilde{k}$ for partition $\tilde{T}_1 \sqcup \tilde{T}_2 \sqcup \ldots \sqcup \tilde{T}_{\tilde{k}} = [d]$. Define $\hat{\DD}_{\tilde{T}_i} = {\cal P}_{\tilde{T}_i}({\boldsymbol{z}}^{(1)}_{\tilde{T}_i},\ldots,{\boldsymbol{z}}^{(\ell)}_{\tilde{T}_i})$, for all $i\in[d]$,  where ${\boldsymbol{z}}^{(j)}_{\tilde{T}_i}$ contains the entries of sample $\boldsymbol{z}^{(j)}$ in coordinates $\tilde{T}_i$ and ${\cal P}_{\tilde{T}_i}$ is the processing corresponding to $\HH_{\tilde{T}_i}$'s uniform convergence. Then 
$$\Pr_{\boldsymbol{z}^{(1)},\ldots, \boldsymbol{z}^{(\ell)}}\left[\sup_{f\in \HH}\left|\E_{\boldsymbol{z}\sim \times_{i\in[\tilde{k}]}\hat{\DD}_{\tilde{T}_i}}\left[f(\boldsymbol{z})\right] - \E_{\boldsymbol{z}\sim \times_{i\in[d]}{\DD}_i}\left[f(\boldsymbol{z})\right]\right|\leq \epsilon\right]\geq 1-\delta.$$

\end{theorem}
\begin{prevproof}{Theorem}{thm:uniform convergence for product measure PARTITION}
For every possible partition use Theorem~\ref{thm:uniform convergence for product measure}.
\end{prevproof}

Nest, we specialize Theorem~\ref{thm:uniform convergence for product measure PARTITION} to indicator functions over sets.

\begin{corollary}\label{cor:VC for product measure}
	We use the same notation as in Theorem~\ref{thm:uniform convergence for product measure PARTITION}. Suppose that all functions in $\HH$ map $\times_{i=1}^{d} \XX_i$ to $\{0,1\}$, i.e.~they are indicators over sets. Suppose also that the VC dimension of $\HH_T$ (viewed as a collection of sets) is $V_T$. Define 
	\begin{align}V_{\max}=\min_{\begin{minipage}[h]{3.5cm}\centering $k$,~partitions\\$T_1 \sqcup T_2 \sqcup \ldots \sqcup T_k = [d]$\end{minipage}} \left\{ k^2 \cdot\max_{i=1,\ldots,k} V_{T_i}\right\}. \label{eq:costas}
	\end{align}
	Assume that the optimum of~\eqref{eq:costas} is attained at $k=\tilde{k}$ for partition $\tilde{T}_1 \sqcup \tilde{T}_2 \sqcup \ldots \sqcup \tilde{T}_{\tilde{k}} = [d]$. 
	
	Then
	%
	$\ell = O\left(\frac{V_{\max}}{\epsilon^2}\cdot \ln \frac{\tilde{k}}{\epsilon}+\frac{\tilde{k}^2}{\epsilon^2}\cdot \ln \frac{\tilde{k}}{\delta} \right)$ samples from $\times_{i\in[d]}\DD_i$ suffice to obtain $(\epsilon,\delta)$-p.m. uniform convergence for $\HH$. 
	Formally,
	$$\Pr_{\boldsymbol{z}^{(1)},\ldots, \boldsymbol{z}^{(\ell)}}\left[\sup_{f\in \HH}\left|\E_{\boldsymbol{z}\sim \times_{i\in[\tilde{k}]}\hat{\DD}_{\tilde{T}_i}}\left[f(\boldsymbol{z})\right] - \E_{\boldsymbol{z}\sim \times_{i\in[d]}{\DD}_i}\left[f(\boldsymbol{z})\right]\right|\leq \epsilon\right]\geq 1-\delta,$$
	where for a given sample $\boldsymbol{z}^{(1)},\ldots, \boldsymbol{z}^{(\ell)}$ from a product distribution $\times_{i\in[d]} \DD_i$ the distributions $\hat{\DD}_{\tilde{T}_i}$ are defined to be uniform over ${\boldsymbol{z}}^{(1)}_{\tilde{T}_i},\ldots,{\boldsymbol{z}}^{(\ell)}_{\tilde{T}_i}$, where ${\boldsymbol{z}}^{(j)}_{\tilde{T}_i}$ contains the entries of sample $\boldsymbol{z}^{(j)}$ in coordinates $\tilde{T}_i$.
\end{corollary}

Table~\ref{tab:productVC} compares the sample complexity for uniform convergence implied by Theorem~\ref{thm:uniform convergence for product measure PARTITION} and Corollary~\ref{cor:VC for product measure} to that implied by VC theory, when the underlying measures are product. Suppose $\HH$ contains the indicator functions of all convex sets in $\mathbb{R}^d$. VC theory does not provide any finite sample  bound for  uniform convergence, as the VC dimension of $\HH$ is $\infty$. Do our results provide a finite bound? Notice that, for all $i$, $\HH_i$  simply contains all intervals in $\mathbb{R}$. Hence, $V_i=2$ and Corollary~\ref{cor:VC for product measure} implies that $\ell = O(\frac{d^2}{\epsilon^2}\cdot \left(\log \frac{d}{\delta}+\log \frac{d}{\epsilon}\right) )$ samples  suffice to obtain $(\epsilon,\delta)$-p.m.~uniform convergence for $\HH$ . In fact, our sample complexity bound can be improved to $O\left(\frac{d^2}{\epsilon^2}\cdot \log \frac{d}{\delta}\right)$, as $O\left(\frac{\log \frac{1}{\delta}}{\epsilon^2}\right)$ samples suffice to guarantee $(\epsilon,\delta)$-uniform convergence for all intervals in $\mathbb{R}$ due to the DKW inequality~\cite{DvoretzkyKW56}. 


In the next a few sections, we apply our uniform convergence results to learn a mechanism with approximately optimal revenue. A type of events called \emph{single-intersecting} (see Definition~\ref{def:single-intersecting})   plays a key role in our analysis. These events are defined based on the geometric shape of the corresponding sets. For example, balls, rectangles and all convex sets are single-intersecting, but this definition includes some non-convex sets as well, for example, ``cross-shaped'' sets. It turns out that being able to handle these non-convex sets is crucial for our results, as many events we care about are not convex but nonetheless are single-intersecting. 

\begin{definition}[Single-intersecting Events]\label{def:single-intersecting}
For any event $\EE$ in $\mathbb{R}^{\ell}$, $\EE$ is \textbf{single-intersecting} if the intersection of $\EE$ and any line that is parallel to one of the axes is an interval. More formally, for any $i\in[\ell]$ and any line $L_i=\left\{x \in \mathbb{R}^{\ell}\ |\ x_{-i}=a_{-i} \right \}$, where $a_{-i} \in \mathbb{R}^{\ell-1}$, the intersection of $L_i$ and $\EE$ is of the form $\left\{x \in \mathbb{R}^{\ell}\ |\ x_{-i}=a_{-i}, x_i\in[\ubar{a},\bar{a}] \right\}$ where $\ubar{a}\leq \bar{a}$. In particular, we allow $\ubar{a}$ to be $-\infty$ and $\bar{a}$ to be $+\infty$.
\end{definition}

\noindent We establish a uniform convergence bound for single-intersecting events by combing the DKW inequality and Theorem~\ref{thm:uniform convergence for product measure}.

\begin{lemma}\label{lem:uniform convergence for single-intersecting}
	For any integer $\ell$, let $\HH$ be the hypothesis class that contains all indicator functions for single-intersecting events in $\mathbb{R}^\ell$. Then $\HH$ has $(\epsilon,\delta)$-p.m. uniform convergence with sample complexity $O\left(\frac{\ell^2}{\epsilon^2}\cdot \log \frac{\ell}{\delta}\right)$.
\end{lemma}
\begin{proof}
	As the projected hypothesis class for the $i$-th coordinate simply contains all intervals in $\R$, the sample complexity for $(\epsilon,\delta)$-uniform convergence is $O({1 \over \epsilon^2}\cdot \log \frac{1}{\delta})$ due to the DKW inequality. The claim follows from Theorem~\ref{thm:uniform convergence for product measure}.
\end{proof}

Next, we show a slightly stronger statement, which is a type of uniform convergence bound when access to approximate distributions is given. More specifically, we argue that for any single-intersecting event, the difference in the probability of this event under two product distributions $\DD=\times_{i\in[\ell]}\DD_i$ and $\hat{\DD}=\times_{i\in[\ell]}\hat{\DD}_{i}$ is at most $2\xi\cdot \ell$, if $||\DD_i-\hat{\DD}_i||_K\leq \xi$ for all $i$. It is not hard to see that Lemma~\ref{lem:Kolmogorov stable for sc} and the DKW inequality imply Lemma~\ref{lem:uniform convergence for single-intersecting}.

\begin{lemma}\label{lem:Kolmogorov stable for sc}
For any integer $\ell$, let $\DD=\times_{i=1}^\ell \DD_i$ and $\hat{\DD}=\times_{i=1}^\ell \hat{\DD}_i$, where $\DD_i$ and $\hat{\DD}_i$ are both supported on $\mathbb{R}$ for any $i\in[\ell]$. If $||\DD_i-\hat{\DD}_i||_K\leq \xi$, $\left|\Pr_\DD[\EE]-\Pr_{\hat{\DD}}[\EE]\right|\leq 2\xi\cdot \ell$ for any single-intersecting event $\EE$.
\end{lemma}
\begin{proof}
	Let $\HH=\left\{\ind_{x\in \EE} :\EE\text{ is \emph{single-intersecting}}\right\}$. By the definition of single-intersecting events, $\HH_i$ is the set of the indicator functions of all intervals in $\mathbb{R}$ for any $i\in [\ell]$. Since $||\DD_i-\hat{\DD}_i||_K\leq \xi$, $$\sup_{g\in\HH_i}\left|\E_{x\sim \DD_i}[g(x)]-\E_{x\sim \hat{\DD}_i}[g(x)]\right|\leq 2\xi.$$ By Lemma~\ref{lem:approx product dist}, $$\sup_{f\in\HH} \left|\E_{\boldsymbol{x}\sim \DD}\left[f(\boldsymbol{x})\right]-\E_{\boldsymbol{x}\sim \hat{\DD}}\left[f(\boldsymbol{x})\right]\right|\leq 2\xi\cdot \ell.$$
\end{proof}

The following table (Table~\ref{tab:productVC}) summarizes some uniform convergence bounds implied by our results in this section.

\begin{table}[h]
	\centering
		\begin{tabular}{c || c | c}
			
			\hline\hline
			Hypotheses Class & \begin{minipage}[h]{4cm}\centering VC Bound\end{minipage} & \begin{minipage}[h]{7cm}\centering  Bounds from Theorem~\ref{thm:uniform convergence for product measure PARTITION} and Corollary~\ref{cor:VC for product measure}\end{minipage}\\
			\hline
			& &\\
			 axis-aligned rectangles in $\mathbb{R}^d$ &  $\tilde{O}(d /\epsilon^2)$ & $\tilde{O}(d /\epsilon^2)$	\\		
			polytopes with $k$ facets in $\mathbb{R}^d$ &  $\tilde{O}(d k /\epsilon^2)$ & $\tilde{O}(d \cdot \min\{d,k\} /\epsilon^2)$\\
			arbitrary convex sets in $\mathbb{R}^d$ & $\infty$ & $\tilde{O}(d^2 /\epsilon^2)$\\
			single-intersecting sets in $\mathbb{R}^d$ & $\infty$ & $\tilde{O}(d^2 /\epsilon^2)$
		\end{tabular}
	\caption{Number of samples required for $(\epsilon,\Theta(1))$-p.m. uniform convergence for different ${\cal H}$'s.}
	\label{tab:productVC}
\end{table}
\section{Constrained Additive Bidders: Uniform Convergence of the Revenue of Sequential Posted Price with Entry Fee Mechanisms}\label{sec:uniform convergence of SPEM}

We consider a specific class of mechanisms, namely Sequential Posted Price with Entry fee Mechanisms, a.k.a. \textbf{SPEM}s; see Algorithm~\ref{alg:spem-mech} for details. Cai and Zhao~\cite{CaiZ17} recently showed that if the bidders' valuations are XOS over independent items, the best SPEM achieves a constant fraction of the optimal revenue. \footnote{Cai and Zhao~\cite{CaiZ17} showed that the best ASPE or RSPM achieves a constant fraction of the optimal revenue. Clearly, any ASPE is also a SPEM, and any RSPM is simply a SPEM if we force the bidders to be unit-demand by only allowing each of them to purchase at most one item.} This section has two goals. The first is to show that, when bidders have constrained additive valuations over independent items, polynomially many samples suffice to guarantee uniform convergence for the revenue of all SPEMs, and hence our ability to select a near-optimal SPEM from polynomially many samples. This can be proven by applying our uniform convergence result for single-intersecting events (Lemma~\ref{lem:uniform convergence for single-intersecting}). The second (and stronger goal) is to show that we can learn a near-optimal SPEM under the max-min learning model (Theorem~\ref{thm:constrained additive Kolmogorov}). We show that the revenue of any SPEM changes no more than $O(\epsilon \cdot m^2\cdot n \cdot H)$ under the true and approximate valuation distributions (Theorem~\ref{thm:revenue stability under K-distance}), where $\epsilon$ is an upper bound of the Kolmogorov distance between the true and approximate distributions for every item marginal of every bidder. It is, of course, not hard to see that Theorem~\ref{thm:revenue stability under K-distance} and the DKW inequality imply uniform convergence of the revenue of all SPEMS. To establish Theorem~\ref{thm:revenue stability under K-distance}, we need to apply Lemma~\ref{lem:Kolmogorov stable for sc} instead of Lemma~\ref{lem:uniform convergence for single-intersecting}.

\begin{algorithm}[ht]
\begin{algorithmic}[1]
\REQUIRE A collection of prices $\{p_{ij}\}_{i\in[n], j\in[m]}$ and a collection of entry fee functions $\{\delta_i(\cdot)\}_{i\in[n]}$ where $\delta_i: 2^{[m]}\mapsto \mathbb{R}$ is bidder $i$'s entry fee function.
\STATE $S\gets [m]$
\FOR{$i \in [n]$}
	\STATE Show bidder $i$ {the} set of available items $S$ and set the entry fee for bidder $i$ to be ${\delta_i}(S)$.
    \IF{Bidder $i$ pays the entry fee ${\delta_i}(S)$}
        \STATE $i$ receives her favorite bundle $S_i^{*}$ and pays $\sum_{j\in S_i^{*}}p_{ij}$.
        \STATE $S\gets S\backslash S_i^{*}$.
    \ELSE
        \STATE $i$ gets nothing and pays $0$.
    \ENDIF
\ENDFOR
\end{algorithmic}
\caption{{\sf Sequential Posted Price with Entry Fee Mechanism (SPEM)}}
\label{alg:spem-mech}
\end{algorithm}

 We first establish a technical lemma, which states that, for any set of items $S$, any set of prices $\{p_j\}_{j\in [m]}$ and entry fee $\delta$, the distribution over the {\em set of items} purchased by a constrained additive bidder whose valuation is drawn from  $\DD=\times_{j\in[m]} \DD_j$ and $\hat{\DD}=\times_{j\in[m]} \hat{\DD}_j$ has total variation distance at most $2m\xi$, if $||\DD_j-\hat{\DD}_j||_K\leq \xi$ for every item $j\in [m]$. This is quite surprising. Given that, for each set of items $S' \subseteq S$, the difference in the probability that the buyer will purchase this particular set $S'$ under $\DD$ and $\hat{\DD}$ could already be as large as {$\Theta(m\xi)$}, and the distribution has an exponentially large support size,  a trivial argument would give a bound of {$2^m\cdot \Theta(m\xi)$}. To overcome this analytical difficulty, we argue instead that for any collection of sets of items, the event that the buyer's favorite set lies in this collection is single-intersecting. Then our result follows from Lemma~\ref{lem:Kolmogorov stable for sc}. Notice that it is crucial that Lemma~\ref{lem:Kolmogorov stable for sc} holds for all events that are single-intersecting, as the event we consider here is clearly non-convex in general.

\begin{lemma}\label{lem:stable demand set}
	For any set $S\subseteq [m]$, any prices $\{p_j\}_{j\in[m]}$ and entry fee $\delta(S)$, let $\LL$ and $\hat{\LL}$ be the distributions over the set of items purchased from $S$ by a constrained additive bidder under prices $\{p_j\}_{j\in[m]}$ and entry fee $\delta$ when her type is drawn from $\DD=\times_{j\in[m]} \DD_j$ and $\hat{\DD}=\times_{j\in[m]} \hat{\DD}_j$ respectively. If $||\DD_{j}-\hat{\DD}_{j}||_K\leq \xi$ for all item $j$, $||\LL-\hat{\LL}||_{TV}\leq 2m \xi.$
\end{lemma}
\begin{proof}

For any set $R\subseteq S$, let $\EE_R$ be the event that the bidder purchases set $R$. Proving that the total variation distance between $\LL$ and $\hat{\LL}$ is no more than $2m\cdot \xi$ is the same as proving that for any $K\leq 2^{|S|}$, $\left|\ \Pr_{\DD}\left[t\in \bigcup_{\ell=1}^K \EE_{R_\ell}\right]-\Pr_{\hat{\DD}}\left[t\in \bigcup_{\ell=1}^K \EE_{R_\ell}\right]\right|\leq 2m\cdot \xi$ where $R_1,\cdots R_K$ are arbitrary distinct subsets of $S$. Since the dimension of the bidder's type space is $m$, if we can prove that $\bigcup_{\ell=1}^K \EE_{R_\ell}$ is always single-intersecting, our claim follows from Lemma~\ref{lem:Kolmogorov stable for sc}. 

For any $j\in[m]$ and $a_{-j}\in \mathbb{R}_{\geq 0}^{{m}-1}$, let $L_j(a_{-j})=\left\{ (t_{j},a_{-j}) \ | t_{j}\in \mathbb{R}_{\geq 0} \right\}$. We claim that $L_j(a_{-j})$ intersects with at most two different $\EE_{U}$ and $\EE_{V}$ where $U$ and $V$ are subsets of $S$. 
	WLOG, we assume that  $(0,a_{-j})\in \EE_U$. 
	\begin{itemize}
		\item If $U=\emptyset$, that means the utility of the favorite set for type $(0,a_{-j})$ is smaller than the entry fee $\delta(S)$. If we increase the value of $t_{j}$, two cases could happen: (1) the utility of the favorite set is still lower than the entry fee; (2) the utility of the favorite set is higher than the entry fee. In case (1), $(t_{j},a_{-j})\in \EE_{\emptyset}$. In case (2), the bidder pays the entry fee and purchases her favorite set $V$. Then item $j$ must be in $V$, because otherwise the utility for set $V$ does not change from type $(0,a_{-j})$ to type $(t_{j},a_{-j})$. If we keep increasing $t_{j}$, bidder $i$'s favorite set remains to be $V$ and she keeps accepting the entry fee and purchasing $V$. Hence, $L_j(a_{-j})$ can intersect with at most one event $\EE_R$ where $R$ is non-empty.
		\item If $U\neq \emptyset$, that means $U$ is the favorite set of type $(0,a_{-j})$ and the utility for winning set $U$ is higher than the entry fee.  If we increase the value of $t_{j}$, two cases could happen: (1) $U$ remains the favorite set; (2) a different set $V$ becomes the new favorite set. In case (1), $(t_{j},a_{-j})\in \EE_U$. In case (2), item $j$ must lie in $V$ but not in $U$, otherwise how could $U$ be better than $V$ for type $(0,a_{-j})$ but worse for type $(t_{j},a_{-j})$.  If we keep increasing $t_{j}$, the bidder's favorite set remains to be $V$ and she keeps accepting the entry fee and purchasing $V$. Hence, $L_j(a_{-j})$ can intersect at most two different events.
	\end{itemize}
	
It is not hard to see that any event $\EE_R$ is an intersection of halfspaces, so the intersection of $L_j(a_{-j})$ with any event $\EE_R$ is an interval. {Also, notice that any type $t\in\mathbb{R}_{\geq 0}^{m}$ must lie in an event $\EE_R$ for some set $R\subseteq S$.} If $L_j(a_{-j})$ intersects with two different events $\EE_U$ and $\EE_V$, the two intersected intervals must lie back to back on $L_j(a_{-j})$. Otherwise, $L_j(a_{-j})$ intersects with at least three different events. Contradiction. Since $L_j(a_{-j})$ intersects with at most two different events, no matter which of these events are in $\{\EE_{R_\ell}\}_{\ell\in[K]}$, the intersection of $L_j(a_{-j})$ and $\bigcup_{\ell=1}^K \EE_{R_\ell}$ is either empty or an interval meaning $\bigcup_{\ell=1}^K \EE_{R_\ell}$ is single-intersecting. Now our claim simply follows from Lemma~\ref{lem:Kolmogorov stable for sc}.
\end{proof}

\begin{theorem}\label{thm:revenue stability under K-distance}
	Suppose all bidders' valuations are constrained additive over independent items. For any SPEM, let $\rev$ and $\widehat{\rev}$ be its expected revenue under $D$ and $\hat{D}$ respectively. If $D_{ij}$ and $\hat{D}_{ij}$ are both supported on $[0,H]$, and $||D_{ij}-\hat{D}_{ij}||_K\leq \xi$ for all $i\in[n]$ and $j\in[m]$, $$\left|\rev-\widehat{\rev}\right|\leq 2nm\xi\cdot \left(mH+\opt \right).$$	
	\end{theorem}
\begin{proof}
	We use a hybrid argument. Consider a sequence of distributions $\{D^{(i)}\}_{i\leq n}$, where $D^{(i)}=\hat{D}_1\times\cdots\times\hat{D}_i\times D_{i+1}\times\cdots\times D_{n},$ and $D^{(0)}=D$, $D^{(n)}=\hat{D}$.
	 We use $\rev^{(i)}$ to denote the expected revenue of the SPEM under $D^{(i)}$. To prove our claim, it suffices to argue that $\left|\rev^{(i-1)} -\rev^{(i)}\right|\leq 2\xi m\cdot \left(m\cdot H+\opt\right).$
	  We denote by $\SS_k$ and $\SS'_k$ the random set of items that remain available after visiting the first $k$ bidders under $D^{(i-1)}$ and $D^{(i)}$. Clearly, for $k\leq i-1$, $||\SS_k-\SS'_k||_{TV}=0$, so the expected revenue collected from the first $i-1$ bidders under $D^{(i-1)}$ and $D^{(i)}$ is the same. According to Lemma~\ref{lem:stable demand set}, $||\SS_i-\SS'_i||_{TV}\leq 2m\cdot \xi$. The total amount of money bidder $i$ spends can never be higher than her value for receiving all the items which is at most $m\cdot H$. So the difference in the expected revenue collected from bidder $i$ under  $D^{(i-1)}$ and $D^{(i)}$ is at most $2\xi\cdot m^2H$. Suppose $R$ is the set of remaining items after visiting the first $i$ bidders, then the expected revenue collected from the last $n-i$ bidders is the same under  $D^{(i-1)}$ and $D^{(i)}$, as these bidders have the same distributions. Moreover, this expected revenue is no more than $\opt$, since the optimal mechanism can simply just sell $R$ to the last $n-i$ bidders using the same prices and entry fee as in the SPEM we consider. Of course, for any fixed $R$, the probabilities that $\SS_i=R$ and $\SS'_i=R$ are different, but since for any $R$ the expected revenue from the last $n-i$ bidders is at most $\opt$, the difference in the expected revenue from the last $n-i$ bidders under  $D^{(i-1)}$ and $D^{(i)}$ is at most $||\SS_i-\SS'_i||_{TV} \cdot \opt \leq 2\xi\cdot m\opt$. Hence, the total difference between $\rev^{(i-1)}$ and  $\rev^{(i)}$ is at most $2\xi m\cdot \left(m H+\opt\right)$. 
\end{proof}

\begin{theorem}\label{thm:constrained additive Kolmogorov}(Max-min Learning for Constrained Additive Bidders)
	When all bidders' valuations are constrained additive over independent items and for any bidder $i$ and any item $j$, $D_{ij}$ and $\hat{D}_{ij}$ are supported on $[0,H]$ and $||D_{ij}-\hat{D}_{ij}||_K\leq \epsilon$  for some $\epsilon=O(\frac{1}{nm})$, then with only access to $\hat{D}=\times_{i,j} \hat{D}_{ij}$, our algorithm can learn an RSPM or ASPE whose revenue is at least $\frac{\opt}{c}-{ \epsilon\cdot O(m^2n H)}$, where $\opt$ is the optimal revenue by any BIC mechanism under $D=\times_{i,j} D_{ij}$. $c>1$ is an absolute constant.
	\end{theorem}
	
Clearly, Theorem~\ref{thm:constrained additive Kolmogorov} also implies a polynomial sample complexity bound for learning an approximately revenue-optimal mechanism. A better sample complexity bound can be obtained directly, i.e.~without invoking the uniform convergence of the revenue of SPEMs, and  is stated as Theorem~\ref{thm:XOS sample} for the broader class of XOS valuations. Similarly, when bidders have simpler valuations, i.e., additive or unit-demand valuations, we can sharpen our results and achieve polynomial-time learnability of the approximately optimal mechanism using more specialized techniques. See Sections~\ref{sec:unit-demand} and~\ref{sec:additive} for details.
\subsection{Unit-demand Valuations: Polynomial-Time Learning} \label{sec:unit-demand}

In this section, we consider bidders with unit-demand valuations, sharpening our results to show how to learn approximately revenue-optimal mechanisms in polynomial time. It is shown in a sequence of works~\cite{ChawlaHMS10, KleinbergW12, CaiDW16} that there exists a sequential posted price mechanism (\textbf{SPM} see Algorithm~\ref{alg:seq-mech} for details) that achieves at least $\frac{1}{24}$ of the optimal revenue when bidders are unit-demand. We show that under all three distribution access models of Section~\ref{sec:prelim} there exists a polynomial-time algorithm that learns a sequential posted price mechanism whose revenue approximates the optimal revenue. We only sketch the proof here and postpone the details to Appendix~\ref{sec:unit-demand appx}.

\begin{theorem}\label{thm:unit-demand}
	When all bidders have unit-demand valuations and \begin{itemize}
		\item $D_{ij}$ is supported on $[0,H]$ for all bidder $i$ and item $j$, there exists a polynomial time algorithm that learns an SPM whose revenue is at least $\frac{\opt}{144}-\epsilon H$ with probability $1-\delta$ given  $O\left(\left(\frac{1}{\epsilon}\right)^2 \left(m^2 n\log \frac{n}{\epsilon} + \log \frac{1}{\delta}\right)\right)$ samples from $D$; or
		\item $D_{ij}$ is a regular distribution for all bidder $i$ and item $j$, there exists a polynomial time algorithm that learns a randomized SPM whose revenue is at least $\frac{\opt}{33}$ with probability $1-\delta$ given $O(\max\{m,n\}^2m^2 n^2\cdot \log \frac{nm}{\delta})$ samples from $D$; or
		\item we are only given access to $\hat{D}_{ij}$ where $||\hat{D}_{ij}-D_{ij}||_K\leq \epsilon$ for all bidder $i$ and item $j$, there is a polynomial time algorithm that constructs a randomized SPM whose revenue under $D$ is at least $\left(\frac{1}{4}-(n+m)\cdot \epsilon\right)\cdot\left(\frac{\opt}{8}-2\epsilon\cdot mnH\right)$\footnote{If we set $\epsilon$ to be $O(\frac{1}{m+n})$, this is the max-min guarantee we want to achieve.}.
	\end{itemize}
\end{theorem}

\noindent\textbf{Sample Access to Bounded Distributions:} the result is due to Morgenstern and Roughgarden~\cite{MorgensternR16}. 

\vspace{.05in}
\noindent\textbf{Direct Access to Approximate Distributions:} we first consider a convex program based on $D$ (see Figure~\ref{fig:CP unit demand}) which is usually referred to as the ex-ante relaxation of the revenue maximization problem~\cite{Alaei11}, and use its optimum as a proxy for $\opt$. Next, we consider a similar convex program based on $\hat{D}$ (see Figure~\ref{fig:CP unit demand approximate dist}) and show that the optima of the two convex programs are close to each other. Finally, we use techniques developed by Chawla et al.~\cite{ChawlaHMS10} to convert the optimal solution of the second convex program into a randomized SPM. We can show that the constructed randomized SPM achieves a revenue that approximates the optimum of the second convex program under $D$, which implies that the mechanism's revenue  also approximates the $\opt$. As we are given $\hat{D}$, we can solve the second convex program and convert its optimal solution into a randomized SPM in polynomial time. See Theorem~\ref{thm:UD Kolmogorov} in Appendix~\ref{sec:unit-demand Kolmogorov} for further details.

\vspace{.05in}
\noindent\textbf{Sample Access to Regular Distributions:} we use a similar convex program relaxation based approach as in the previous case. The main difference is that regular distributions could be unbounded and thus ruin the approximation guarantee. We show how to use the Extreme Value theorem in~\cite{CaiD11b} to truncate the distributions without hurting the revenue by much. See Theorem~\ref{thm:UD regular} in Appendix~\ref{sec:unit-demand regular} for further details.
\subsection{Additive Valuations: Polynomial-Time Learning}\label{sec:additive}
In this section, we consider bidders with additive valuations, again sharpening our results to show polynomial-time learnability. It is known that the better of the following two mechanisms achieves at least $\frac{1}{8}$ of the optimal revenue when all bidders have additive valuations~\cite{Yao15,CaiDW16}:

\vspace{.05in}	
\noindent\textbf{Selling Separately}: the mechanism sells each item separately using Myerson's optimal auction.

\vspace{.05in}	
\noindent \textbf{VCG with Entry Fee}: the mechanism solicits bids $\bold{b}=(b_1,\cdots, b_n)$ from the bidders, then offers each bidder $i$ the option to participate for an entry fee $e_i(b_{-i},D_i)$, which is the median of the random variable $\sum_{j\in[m]}(t_{ij}-\max_{k\neq i} b_{kj})^+$, where $t_i\sim D_i$\footnote{The entry fee function defined in~\cite{Yao15,CaiDW16} is slightly different. They showed that there exists an entry fee $X_i$, such that bidder $i$ accepts the entry fee with probability at least $1/2$. Then they argued that extracting $X_i/2$ as the revenue in the VCG with entry fee mechanism is enough to obtain a factor $8$ approximation. It is not hard to observe that our entry fee is accepted with probability exactly $1/2$, thus our entry fee is at least as large as $X_i$. So our mechanism also suffices to provide a factor $8$ approximation.}. This random variable is exactly bidder $i$'s utility when her type is $t_i$ and the other bidders' are $b_{-i}$. If bidder $i$ chooses to participate, she pays the entry fee and can take any item $j$ at price $\max_{k\neq i} b_{kj}$. Notice that the mechanism never over allocate any item, as only the highest bidder for an item can afford it. 

Indeed, only counting the revenue from the entry fee in the second mechanism and the optimal revenue from selling the items separately already suffices to provide an $8$-approximation~\cite{Yao15, CaiDW16}. 

\begin{theorem}[\cite{CaiDW16}]\label{thm:UB additive}
	Let $\srev$ be the optimal revenue for selling the items separately and $\brev$ be the expected entry fee collected from the VCG with entry fee mechanism. Then $\opt\leq 6\cdot \srev+2\cdot\brev.$ 
\end{theorem}

Goldner and Karlin~\cite{GoldnerK16} showed that one sample suffices to learn a mechanism that achieves a constant fraction of the optimal revenue when $D_{ij}$ is regular for all $i\in[n]$ and $j\in[m]$. 
We show how to learn an approximately optimal mechanism in the other two models.
\begin{theorem}\label{thm:additive}
	When the bidders have additive valuations and\begin{itemize}
		\item $D_{ij}$ is supported on $[0,H]$ for all bidder $i$ and item $j$, we can learn in polynomial time a mechanism whose expected revenue is at least $\frac{\opt}{32}-{\epsilon}\cdot H$ with probability $1-\delta$ given $O\left(\left(\frac{m}{\epsilon}\right)^2 \cdot\left(n\log n\log \frac{1}{\epsilon}+\log\frac{1}{\delta} \right)\right)$ samples from $D$; or
		\item  we are only given access to distributions $\hat{D}_{ij}$ where $||\hat{D}_{ij}-D_{ij}||_K\leq \epsilon$ for all bidder $i$ and item $j$, there is a polynomial time algorithm that constructs a mechanism whose expected revenue under $D$ is at least $\frac{\opt}{266}-96\epsilon\cdot mnH$ when $\epsilon\leq \frac{1}{16\max\{m,n\}}$.
	\end{itemize}
\end{theorem}

\noindent\textbf{Sample Access to Bounded Distributions:} Goldner and Karlin's proof~\cite{GoldnerK16} can be directly applied to the bounded distributions to show a single sample suffices to learn a mechanism whose expected revenue approximates the $\brev$. Then as $\srev$ is the revenue of $m$ separate single-item auctions, we can use the result in~\cite{MorgensternR16} to approximate it. See Theorem~\ref{thm:additive bounded} in Appendix~\ref{sec:additive bounded} for further details.

\vspace{.05in}
\noindent\textbf{Direct Access to Approximate Distributions:} for each single item, we apply Theorem~\ref{thm:unit-demand} to learn an individual auction, then run these learned auctions in parallel. Clearly, the combined auction's revenue approximates $\srev$. For $\brev$, we show that for every bidder $i$ and every bid profile $b_{-i}$ of the other bidders, the event that corresponds to bidder $i$ accepting any entry fee is \emph{single-intersecting} (see Definition~\ref{def:single-intersecting}). This implies that the probability for a bidder to accept an entry fee under $\hat{D}$ and $D$ is close (Lemma~\ref{lem:Kolmogorov stable for sc}). So we can essentially use the median of $\sum_{j\in[m]}(t_{ij}-\max_{k\neq i} b_{kj})^+$ with $t_i\sim\hat{D}_i$ as the entry fee. See Theorem~\ref{thm:additive Kolmogorov} in Appendix~\ref{sec:additive Kolmogorov} for further details.
\section{XOS Valuations} \label{sec:constrained additive}
In this section we go beyond constained additive valuations to show learnability of approximately revenue-optimal auctions from polynomially many samples. The better of the following two mechanisms is known to achieve a constant fraction of the optimal revenue, when bidders have valuations that are XOS over independent items~\cite{CaiZ17}.

\vspace{.05in}
\noindent\textbf{Rationed Sequential Posted Price Mechanism (RSPM)}: the mechanism is almost the same as SPM in Algorithm~\ref{alg:seq-mech}, except there is an extra constraint that every bidder can purchase at most one item.

\vspace{.05in}
\noindent\textbf{Anonymous Sequential Posted Price with Entry Fee Mechanism (ASPE)}: every buyer faces the same collection of item prices $\{p_j\}_{j\in[m]}$. The seller visits the bidders sequentially. For every bidder, the seller shows her all the available items (i.e. items that have not yet been purchased) and the associated price for each item, then asks her to pay a personalized entry fee which depends on her type distribution and the set of available items. If the bidder accepts the entry fee, she can proceed to purchase any available item at the given price; if she rejects the entry fee, she neither receives nor pays anything. See Algorithm~\ref{alg:aspe-mech} for details.

\begin{theorem}\cite{CaiZ17}\label{thm:simple XOS}
	There exists a collection of prices $\{p^*_j\}_{j\in[m]}$, such that if we set the entry fee function $\delta^*_i(S)$ to be the median of bidder $i$'s utility for set $S$, either the ASPE$(p^*,\delta^*)$ or the best RSPM achieves at least a constant fraction of the optimal revenue when bidders' valuations are XOS over independent items. More formally, let $u^*_i(t_i, S)=\max_{S^*\subseteq S} v_i(t_i, S^*)-\sum_{j\in S^*} p^*_j$ be bidder $i$'s utility for the set of items $S$ when her type is $t_i$. We define $\delta^*_i(S)$ to be the median of the random variable $u^*_i(t_i,S)$ (with $t_i\sim D_i$) for any set $S\subseteq [m]$.  Moreover, the price $p^*_j$ for any item $j$ is no larger than $2G$, where $G = \max_{i,j} G_{ij}$ and $G_{ij}:= \sup_x \left\{ \Pr_{t_{ij}\sim D_{ij}}\left[V_i(t_{ij})\geq x\right]\geq \frac{1}{5\max\{m,n\}}\right\}$.
\end{theorem}

{
Our goal next is to bound the sample complexity for learning a near-optimal RSPM and the ASPE described in Theorem~\ref{thm:simple XOS} under XOS valuations.} 

We consider first the task of learning a near-optimal RSPM. In a RSPM, all bidders are restricted to be unit-demand, so the revenue of the best RSPM is upper bounded by the optimal revenue in the corresponding unit-demand setting. In Section~\ref{sec:unit-demand}, we have shown how to learn an approximately optimal mechanism for unit-demand bidders, and those algorithms can be used to approximate the best RSPM. 

So, for the rest of this section, it suffices to focus on learning an ASPE whose revenue approximates the revenue of the ASPE described in Theorem~\ref{thm:simple XOS}.  We will do this in Section~\ref{sec:XOS sample}. Before that,
we need a robust version of Theorem~\ref{thm:simple XOS}. In the next Lemma, we argue that if we use a collection of prices $\{p'_j\}_{j\in[m]}$ sufficiently close to $\{p^*_j\}_{j\in[m]}$ and entry fee $\delta'_i(S)$ sufficiently close to the median of the utility for every bidder $i$ and subset $S$, the better of the corresponding ASPE and the best RSPM still approximates the optimal revenue. We postpone the proof to Appendix~\ref{sec:appx XOS}.

\begin{lemma}\label{lem:approx ASPE}
	For any $\epsilon>0$ and $\mu\in[0,\frac{1}{4}]$, let $\{p'_j\}_{j\in[m]}$ be a collection of prices such that $|p'_j-p^*_j|\leq \epsilon$ for all $j\in[m]$, where $\{p^*_j\}_{j\in[m]}$ is the collection of prices in Theorem~\ref{thm:simple XOS}. Let $\delta'_i(S)$ be bidder $i$'s entry fee function such that $\Pr_{t_i\sim D_i}\left [u'_i(t_i,S)\geq \delta'_i(S)\right]\in [1/2-\mu,1/2+\mu]$ for any set $S\subseteq [m]$, where $u'_i(t_i,S) = \max_{S*\subseteq S} v_i(t_i,S^*)-\sum_{j\in S^*} p'_j$. Then, either the ASPE$(p',\delta')$ or the best RSPM achieves revenue at least $\frac{\opt}{\CC_1(\mu)}-\CC_2(\mu)\cdot (m+n)\cdot \epsilon$ when bidders' valuations are XOS over independent items. Both $\CC_1(\cdot)$ and $\CC_2(\cdot)$ are  monotonically increasing functions that only depend on $\mu$. 
\end{lemma}

\begin{definition}\label{def:eps mu ASPE}
	We say a collection of prices $\{p_j\}_{j\in[m]}$ is in the $B$-bounded $\epsilon$-net if $p_j$ is a multiple of $\epsilon$ and no larger than $B$ for any item $j$. For any collection of prices $\{p_j\}_{j\in[m]}$, we say the entry fee functions are $\mu$-balanced if for every bidder $i$ and every set $S\subseteq [m]$, her entry fee $\delta_i(S)$ satisfies  $\Pr_{t_i\sim D_i}[u_i(t_i,S)\geq \delta_i(S)]\in [1/2-\mu,1/2+\mu]$, where $u_i(t_i,S) = \max_{S*\subseteq S} v_i(t_i,S^*)-\sum_{j\in S^*} p_j$.
\end{definition}

\begin{corollary}\label{cor:discretization of prices}
	For  bidders with valuations that are XOS over independent items and any $\epsilon>0$, there exists a collection of prices $\{p_j\}_{j\in[m]}$ in the $2G$-bounded $\epsilon$-net such that for any $\mu$-balanced entry fee functions $\{\delta_i(\cdot)\}_{i\in[n]}$ with $\mu\in[0,\frac{1}{4}]$, either the ASPE$(p,\delta)$ or the best RSPM achieves revenue at least $\frac{\opt}{\CC_1(\mu)}-\CC_2(\mu)\cdot (m+n)\cdot \epsilon$. 
\end{corollary}


\notshow{
\subsection{Constrained Additive Valuations: direct access to approximate distributions}\label{sec:constrained additive kolomogorov}

In this section, we consider the model where we only have access to an approximate distribution $\hat{D}$ and bidders have constrained additive valuations. 
Our learning algorithm is a two-step procedure. In the first step, we argue that the approximate distribution $\hat{D}$ can be used to define a balanced entry fee function for each bidder under any prices. In particular, we propose to use the median of bidder $i$'s utility under $\hat{D}_i$ to compute the entry fees. The challenge is to show that, whenever $||D_{ij}-\hat{D}_{ij}||_{K}\leq \xi$ for every bidder $i$ and item $j$, the median of the utility of any bidder $i$ for any subset $S$ under any collection of prices $\{p_j\}_{j\in[m]}$ is not too different when the bidder's type is drawn from $D_i$ or $\hat{D}_i$; in particular, the entry fee functions thus defined from $\hat{D}$ are balanced. This is shown in Lemma~\ref{lem:Kolmogorov learn entry fee}. Given this lemma, we are able to construct an ASPE for every collection of prices $\{p_j\}_{j\in[m]}$ in the $\epsilon$-net. In the second step, we need to identify a mechanism  with high revenue among all ASPEs we constructed. If we have sample access to the  actual distribution, this is a simple task {as we can take a polynomial number of samples from $D$ and argue that with probability almost $1$, the empirical revenue of ASPE$(p,\delta)$ is close to its true expected revenue for every price vector $p$ in the $\epsilon$-net. So the ASPE with the highest empirical revenue also has high expected revenue.} 
 But we only have access to an approximate distribution $\hat{D}$; can we argue that the expected revenue of any ASPE does not change much under $D$ or $\hat{D}$? Notice that this is certainly not true for arbitrary mechanisms, as it is easy to construct a DSIC single item auction with two bidders where the expected revenue under two distributions that are close in Kolmogorov distance is far away
 . Surprisingly, we can argue that the expected revenue of any ASPE thus defined is within $\poly(n,m)\cdot H\cdot \xi$ when the bidders' types are drawn from $D$ or $\hat{D}$.

 To argue these results, we need to understand how the probability of events changes from distribution $D$ to $\hat{D}$. Clearly, if we consider arbitrary events, the difference in their probabilities under $D$ and $\hat{D}$ can be arbitrarily far. We need to identify structure in the events that are of interest to us, which allows us to argue that computations under $D$ and $\hat{D}$ are close. Since our distributions are supported on a subset of the Euclidean space, any event is simply a set in the Euclidean space. In Definition~\ref{def:single-intersecting}, we define a type of events called \emph{single-intersecting} events based on the geometric shape of the corresponding sets. For example, balls, rectangles and all convex sets are single-intersecting, but this definition includes some non-convex sets as well, for example, cross-shaped sets. It turns out that being able to include these non-convex sets is crucial for our result, as many events we care about are not convex but nonetheless are single-intersecting. Next we argue that for any single-intersecting event, the difference of the probability under $D$ and $\hat{D}$ only grows linearly in the dimension of the support space (Lemma~\ref{lem:Kolmogorov stable for sc}).
\notshow{\begin{definition}[Single-intersecting Events]\label{def:single-intersecting}
For any event $\EE$ in $\mathbb{R}^{\ell}$, $\EE$ is \textbf{single-intersecting} if the intersection of $\EE$ and any line that is parallel to one of the axes is an interval. More formally, for any $i\in[\ell]$ and any line $L_i=\left\{x \in \mathbb{R}^{\ell}\ |\ x_{-i}=a_{-i} \right \}$, where $a_{-i}$ is some $\ell-1$ dimensional vector in $\mathbb{R}^{\ell-1}$, the intersection of $L_i$ and $\EE$ is $\left\{x \in \mathbb{R}^{\ell}\ |\ x_{-i}=a_{-i}, x_i\in[\ubar{a},\bar{a}] \right\}$ for some real numbers $\ubar{a}$ and $\bar{a}$.
\end{definition}

\begin{lemma}\label{lem:Kolmogorov stable for sc}
For any integer $\ell$, let $\DD=\times_{i=1}^\ell \DD_i$ and $\hat{\DD}=\times_{i=1}^\ell \hat{\DD}_i$, where $\DD_i$ and $\hat{\DD}_i$ are both supported on $[0,H]$ for any $i\in[\ell]$. If $||\DD_i-\hat{\DD}_i||_K\leq \xi$, $\left|\Pr_\DD[\EE]-\Pr_{\hat{\DD}}[\EE]\right|\leq 2\xi\cdot \ell$ for any single-intersecting event $\EE$.
\end{lemma}
\begin{prevproof}{Lemma}{lem:Kolmogorov stable for sc} Let $\FF_i$ and $\hat{\FF}_i$ be the cdf for $\DD_i$ and $\hat{\DD}_i$ respectively. We will prove the statement using a hybrid argument. We create a sequence of product distributions $\{\DD^{(j)}\}_{j\leq \ell}$, where $\DD^{(j)}=\hat{\DD}_1\times\cdots\times\hat{\DD}_j\times \DD_{j+1}\times\cdots\times \DD_{\ell},$ and $\DD^{(0)}=\DD$, $\DD^{(\ell)}=\hat{\DD}$. To prove our claim, it suffices to show that for any integer $j\in[\ell]$, $\left|\Pr_{\DD^{(j-1)}}[\EE]-\Pr_{{\DD}^{(j)}}[\EE]\right|\leq 2\xi.$
	Let us first fix some notations. Let $\EE_{-j}=\left\{x_{-j}\ |\ \exists x_j\in \mathbb{R},\ (x_j,x_{-j})\in \EE\right\}$, $\EE_j(a_{-j})= \left\{x_{j} \in \mathbb{R} \ |\ (x_j,a_{-j})\in \EE\right\}$ for all $j\in[\ell]$. As $\EE$ is single-intersecting, $\EE_j(x_{-j})$ is an interval for all $j$ and $x_{-j}$. Moreover, since $||\DD_j-\hat{\DD}_j||_K\leq \xi$, we have $|\Pr_{\DD_j}[x_j\in \EE_j(x_{-j})]-\Pr_{\hat{\DD}_j}[x_j\in \EE_j(x_{-j})]|\leq 2\xi$ for all $j$ and $x_{-j}$. Next, we bound the difference for the probability of event $\EE$ under $\DD^{(j-1)}$ and $\DD^{(j)}$. 
	\begin{align*}
		&\left|\Pr_{\DD^{(j-1)}}[\EE]-\Pr_{{\DD}^{(j)}}[\EE]\right|\\
		=&\Big{|}\int_{\mathbb{R}^{\ell-1}}\ind{\left [x_{-j}\in \EE_{-j}\right]}\left(\int_{\mathbb{R}}\ind{\left[x_j\in \EE_j(x_{-j})\right]}d \FF_j(x_j)\right) d\hat{\FF}_1(x_1)\cdots d\hat{\FF}_{j-1}(x_{j-1}) d\FF_{j+1}(x_{j+1})\cdots d\FF_{\ell}(x_{\ell})\\
		 & - \int_{\mathbb{R}^{\ell-1}}\ind{\left [x_{-j}\in \EE_{-j}\right]}\left(\int_{\mathbb{R}}\ind{\left[x_j\in \EE_j(x_{-j})\right]}d \hat{\FF}_j(x_j)\right) d\hat{\FF}_1(x_1)\cdots d\hat{\FF}_{j-1}(x_{j-1}) d\FF_{j+1}(x_{j+1})\cdots d\FF_{\ell}(x_{\ell})\Big{|}\\
		 = & \left|\int_{\mathbb{R}^{\ell-1}}\ind{\left [x_{-j}\in \EE_{-j}\right]} \left(\Pr_{\DD_j}[x_j\in \EE_j(x_{-j})]-\Pr_{\hat{\DD}_j}[x_j\in \EE_j(x_{-j})]\right)
		  d\hat{\FF}_1(x_1)\cdots d\hat{\FF}_{j-1}(x_{j-1}) d\FF_{j+1}(x_{j+1})\cdots d\FF_{\ell}(x_{\ell})\right|\\
		  \leq & 2\xi\cdot\int_{\mathbb{R}^{\ell-1}}\ind{\left [x_{-j}\in \EE_{-j}\right]} d\hat{\FF}_1(x_1)\cdots d\hat{\FF}_{j-1}(x_{j-1}) d\FF_{j+1}(x_{j+1})\cdots d\FF_{\ell}(x_{\ell}) = 2\xi
	\end{align*}
	\end{prevproof}}
	It is quite easy to see that the single-intersecting condition can be relaxed to $k$-intersecting where the intersection of the event $\EE$ with any line parallel to an axis is the union of at most $k$ intervals. Indeed, using a proof almost identical to the proof of Lemma~\ref{lem:Kolmogorov stable for sc}, we can argue that the difference in the probability of an event $\EE$ under $\DD$ and $\hat{\DD}$ only grows linearly in the dimension of the support space and $k$. Next, we formally state the first step of our learning algorithm. The proof can be found in Appendix~\ref{sec:appx constrained additive}.
\begin{lemma}\label{lem:Kolmogorov learn entry fee}
Suppose $||D_{ij}-\hat{D}_{ij}||_K\leq \xi$ for any bidder $i$ and any item $j$. For any collection of prices $\{p_j\}_{j\in[m]}$, let $u^{(p)}_i(t_i,S)=\max_{S*\subseteq S} v_i(t_i,S^*)-\sum_{j\in S^*} p_j$. Define the entry fee $\delta_i^{(p)}(S)$
 of bidder $i$ for set $S$ under $\{p_j\}_{j\in[m]}$ to be the median of $u^{(p)}_i(t_i,S)$ when bidder $i$'s type $t_i$ is drawn from $\hat{D}_i$. Then $\left\{\delta_i^{(p)}(\cdot)\right\}_{i\in[n]}$ is a collection of $2m\xi$-balanced entry fee functions for any collection of prices $\{p_j\}_{j\in[m]}$.\end{lemma}

So far, for any collection of prices $\{p_j\}_{j\in[m]}$, we have defined (using $\hat{D}$) a collection of $2m\xi$-balanced entry fee functions $\{\delta_i^{(p)}(\cdot)\}_{i\in[n]}$. Next, we need to identify prices $\{p_j\}_{j\in[m]}$ such that the corresponding ASPE mechanism achieves high revenue under the actual distribution $D$. Our goal is to show that the expected revenue of simultaneously all ASPE mechanisms, defined for all prices using $\hat{D}$ as above, is not much different under $D$ and $\hat{D}$. We first establish a technical lemma, which states that, for any set of available items $S$ and entry fee $\delta_i(S)$, the distribution over the {\em set of items} purchased by bidder $i$ under $D_i$ and $\hat{D}_i$ has total variation distance at most $2m\xi$. This is quite surprising. Given that, for each set of items $S' \subseteq S$, the difference in the probability that the buyer will purchase this particular set $S'$ under $D_i$ and $\hat{D}_i$ could already be as large as {$\Theta(m\xi)$}, and the distribution has an exponentially large support size,  a trivial argument would give a bound of {$2^m\cdot \Theta(m\xi)$}. To overcome this analytical difficulty, we argue instead that for any collection of sets of items, the event that the buyer's favorite set lies in this collection is single-intersecting. Then our result follows from Lemma~\ref{lem:Kolmogorov stable for sc}. Notice that it is crucial that Lemma~\ref{lem:Kolmogorov stable for sc} holds for all events that are single-intersecting, as the event we consider here is clearly non-convex in general. 

\notshow{
\begin{lemma}\label{lem:stable favorite set}
	For any bidder $i$, any set $S\subseteq [m]$, any prices $\{p_j\}_{j\in[m]}$ and entry fee $\delta_i(S)$, let $\LL$ and $\hat{\LL}$ be the distributions over the set of items purchased from $S$ by bidder $i$ under prices $\{p_j\}_{j\in[m]}$ and entry fee $\delta_i(S)$ when her type is drawn from $D_i$ and $\hat{D}_i$ respectively. If $||D_{ij}-\hat{D}_{ij}||_K\leq \xi$ for all item $j$, $||\LL-\hat{\LL}||_{TV}\leq 2m \xi.$
\end{lemma}
\begin{proof}
	For any set $R\subseteq S$, let $\EE_R$ be the event that bidder $i$ purchases set $R$. Proving that the total variation distance between $\LL$ and $\hat{\LL}$ is no more than $2m\cdot \xi$ is the same as proving that for any $K\leq 2^{|S|}$, $\left|\ \Pr_{D_i}\left[t_i\in \bigcup_{\ell=1}^K \EE_{R_\ell}\right]-\Pr_{\hat{D}_i}\left[t_i\in \bigcup_{\ell=1}^K \EE_{R_\ell}\right]\right|\leq 2m\cdot \xi$ where $R_1,\cdots R_K$ are arbitrary distinct subsets of $S$. Since the dimension of the bidder's type space is $m$, if we can prove that $\bigcup_{\ell=1}^K \EE_{R_\ell}$ is always single-intersecting, our claim follows from Lemma~\ref{lem:Kolmogorov stable for sc}. 
	
	For any $j\in[m]$ and $a_{i,-j}\in [0,H]^{{m}-1}$, let $L_j(a_{i,-j})=\left\{ (t_{ij},a_{i,-j}) \ | t_{ij}\in [0,H] \right\}$. We claim that $L_j(a_{i,-j})$ intersects with at most two different $\EE_{U}$ and $\EE_{V}$ where $U$ and $V$ are subsets of $S$. 
	WLOG, we assume that  $(0,a_{i,-j})\in \EE_U$. 
	\begin{itemize}
		\item If $U=\emptyset$, that means the utility of the favorite set for type $(0,a_{i,-j})$ is smaller than the entry fee $\delta_i(S)$. If we increase the value of $t_{ij}$, two cases could happen: (1) the utility of the favorite set is still lower than the entry fee; (2) the utility of the favorite set is higher than the entry fee. In case (1), $(t_{ij},a_{i,-j})\in \EE_{\emptyset}$. In case (2), bidder $i$ pays the entry fee and purchases her favorite set $V$. Then item $j$ must be in $V$, because otherwise the utility for set $V$ does not change from type $(0,a_{i,-j})$ to type $(t_{ij},a_{i,-j})$. If we keep increasing $t_{ij}$, bidder $i$'s favorite set remains $V$ and she keeps paying the entry fee and purchasing $V$. Hence, $L_j(a_{i,-j})$ can intersect with at most one event $\EE_R$ where $R$ is non-empty.
		\item If $U\neq \emptyset$, that means $U$ is the favorite set of type $(0,a_{i,-j})$ and the utility for winning set $U$ is higher than the entry fee.  If we increase the value of $t_{ij}$, two cases could happen: (1) $U$ remains the favorite set; (2) a different set $V$ becomes the new favorite set. In case (1), $(t_{ij},a_{i,-j})\in \EE_U$. In case (2), item $j$ must be in $V$ and not in $U$, otherwise how could $U$ be better than $V$ for type $(0,a_{i,-j})$ but worse for type $(t_{ij},a_{i,-j})$.  If we keep increasing $t_{ij}$, bidder $i$'s favorite set remains to $V$ and she keeps paying the entry fee and purchasing $V$. Hence, $L_j(a_{i,-j})$ can intersect at most two different events.
	\end{itemize}
	
It is not hard to see that any event $\EE_R$ is an intersection of halfspaces, so the intersection of $L_j(a_{i,-j})$ with any event $\EE_R$ is an interval. {Also, notice that any type $t_i$ must lie in an event $\EE_R$ for some set $R\subseteq S$.} If $L_j(a_{i,-j})$ intersects with two different events $\EE_U$ and $\EE_V$, the two intersected intervals must lie back to back on $L_j(a_{i,-j})$. Otherwise, $L_j(a_{i,-j})$ intersects with at least three different events. Contradiction. Since $L_j(a_{i,-j})$ intersects with at most two different events, no matter which of these events are in $\{\EE_{R_\ell}\}_{\ell\in[K]}$, the intersection of $L_j(a_{i,-j})$ and $\bigcup_{\ell=1}^K \EE_{R_\ell}$ is either empty or an interval,
  which means $\bigcup_{\ell=1}^K \EE_{R_\ell}$ is single-intersecting. Now our claim simply follows from Lemma~\ref{lem:Kolmogorov stable for sc}.
\end{proof}
}

With Lemma~\ref{lem:stable favorite set}, we are ready to show that the revenue of any ASPE does not change much under $D$ and $\hat{D}$.\begin{lemma}\label{lem:difference in revenue Kolmogorov}
	For any ASPE$(p,\delta)$, let $\rev(p,\delta)$ and $\widehat{\rev}(p,\delta)$ be its expected revenue under $D$ and $\hat{D}$ respectively. If $||D_{ij}-\hat{D}_{ij}||_K\leq \xi$ for all $i\in[n]$ and $j\in[m]$, then $|\rev(p,\delta)-\widehat{\rev}(p,\delta)|\leq 2nm\xi\cdot \left(mH+\opt_{ASPE}\right)$,
	{where $\opt_{ASPE}$ is the optimal revenue obtainable by an ASPE mechanism.}
	\end{lemma}
\begin{prevproof}{Lemma}{lem:difference in revenue Kolmogorov}
	We use a hybrid argument. Consider a sequence of distributions $\{D^{(i)}\}_{i\leq n}$, where $D^{(i)}=\hat{D}_1\times\cdots\times\hat{D}_i\times D_{i+1}\times\cdots\times D_{n},$ and $D^{(0)}=D$, $D^{(n)}=\hat{D}$.
	 We use $\rev^{(i)}(p,\delta)$ to denote the expected revenue of ASPE$(p,\delta)$ under $D^{(i)}$. To prove our claim, it suffices to argue that $\left|\rev^{(i-1)}(p,\delta) -\rev^{(i)}(p,\delta)\right|\leq 2\xi m\cdot \left(m\cdot H+\opt\right).$
	  We denote by $\SS_k$ and $\SS'_k$ the random set of items that remain available after visiting the first $k$ bidders under $D^{(i-1)}$ and $D^{(i)}$. Clearly, for $k\leq i-1$, $||\SS_k-\SS'_k||_{TV}=0$, so the expected revenue collected from the first $i-1$ bidders under $D^{(i-1)}$ and $D^{(i)}$ is the same. According to Lemma~\ref{lem:stable favorite set}, $||\SS_i-\SS'_i||_{TV}\leq 2m\cdot \xi$. The total amount of money bidder $i$ spends can never be higher than her value for receiving all the items which is at most $m\cdot H$. So the difference in the expected revenue collected from bidder $i$ under  $D^{(i-1)}$ and $D^{(i)}$ is at most $2\xi\cdot m^2H$. Suppose $R$ is the set of remaining items after visiting the first $i$ bidders, then the expected revenue collected from the last $n-i$ bidders is the same under  $D^{(i-1)}$ and $D^{(i)}$, as these bidders have the same distributions. Moreover, this expected revenue is no more than $\opt_{ASPE}$, since the optimal ASPE can simply just sell $R$ to the last $n-i$ bidders using the same prices and entry fee as in ASPE$(p,\delta)$. Of course, for any fixed $R$, the probabilities that $\SS_i=R$ and $\SS'_i=R$ are different, but since for any $R$ the expected revenue from the last $n-i$ bidders is at most $\opt_{ASPE}$, the difference in the expected revenue from the last $n-i$ bidders under  $D^{(i-1)}$ and $D^{(i)}$ is at most $||\SS_i-\SS'_i||_{TV} \cdot \opt \leq 2\xi\cdot m\opt_{ASPE}$. Hence, the total difference between $\rev^{(i-1)}(p,\delta)$ and  $\rev^{(i)}(p,\delta)$ is at most $2\xi m\cdot \left(m H+\opt_{ASPE}\right)$. 
	  Furthermore, $\left|\rev(p,\delta)-\widehat{\rev}(p,\delta)\right|\leq \sum_{i=1}^{n}\left|\rev^{(i-1)}(p,\delta) -\rev^{(i)}(p,\delta)\right|\leq 2 nm\xi \left(mH+\opt_{ASPE}\right).$
\end{prevproof}

Using  Lemma~\ref{lem:difference in revenue Kolmogorov}, we can prove the main Theorem of this section. The proof is postponed to Appendix~\ref{sec:appx constrained additive}.

\begin{theorem}\label{thm:constrained additive Kolmogorov}
	When all bidders' valuations are drawn from distributions that are constrained additive, and for any bidder $i$ and any item $j$, $D_{ij}$ and $\hat{D}_{ij}$ are supported on $[0,H]$ and $||D_{ij}-\hat{D}_{ij}||_K\leq \xi$  for some $\xi=O(\frac{1}{nm})$, then with only access to $\hat{D}=\times_{i,j} \hat{D}_{ij}$, our algorithm can learn an RSPM and an ASPE such that the better of the two mechanisms has revenue at least $\frac{\opt}{c}-{\xi\cdot O(m^2n H)}$, where $\opt$ is the optimal revenue by any BIC mechanism under $D=\times_{i,j} D_{ij}$. $c>1$ is an absolute constant.
\end{theorem}

}

\subsection{XOS Valuations: sample access to bounded and regular distributions}\label{sec:XOS sample}
In this section, we consider how to learn an ASPE with high revenue given sample access to $D$. Our learning algorithm is a two-step procedure. In the first step, we take a few samples from $D$ and use these samples to set the entry fee for every collection of prices $\{p_j\}_{j\in[m]}$ in the $\epsilon$-net. More specifically, to decide $\delta_i(S)$ we compute the utility of bidder $i$ for set $S$ under $\{p_j\}_{j\in[m]}$ over all the samples and take the empirical median among all these utilities to be $\delta_i(S)$. With a polynomial number of samples, we can guarantee that for any $\{p_j\}_{j\in[m]}$ in the $\epsilon$-net the computed entry fee functions $\{\delta_i(\cdot)\}_{i\in[n]}$ are $\mu$-balanced. Now, we have created an ASPE for every $\{p_j\}_{j\in[m]}$ in the $\epsilon$-net. In the second step, we take some fresh samples from $D$ and use them to estimate the revenue for each of the ASPEs we created in the first step,  then pick the one that has the highest empirical revenue. It is not hard to argue that with a polynomial number of samples the mechanism we pick has high revenue with probability almost $1$. Combining our algorithm with Theorem~\ref{thm:unit-demand}, we obtain the following theorem.
\begin{theorem}\label{thm:XOS sample}
	When all bidders' valuations are XOS over independent items and
	\begin{itemize}
		\item the random variable $V_i(t_{ij})$ is supported on $[0,H]$ for each bidder $i$ and item $j$, we can learn an RSPM and an ASPE such that with probability at least $1-\delta$ the better of the two mechanisms has revenue at least $\frac{\opt}{c_1}-\xi\cdot H$ for some absolute constant $c_1>1$ given $O\left((\frac{mn}{\xi})^2 \cdot (m\cdot\log \frac{m+n}{\xi} + \log \frac{1}{\delta})\right)$ samples from $D$; 
		\item  the random variable $V_i(t_{ij})$ is regular for each bidder $i$ and item $j$, we can learn an RSPM and an ASPE such that with probability at least $1-\delta$ the better of the two mechanisms has revenue at least $\frac{\opt}{c_2}$ for some absolute constant $c_2>1$ given $O\left(\max\{m,n\}^2m^2n^2 \left(m\log ({m+n}) + \log \frac{1}{\delta}\right)\right)$  samples from $D$.
	\end{itemize}
\end{theorem}

 The bounded case is proved as Theorem~\ref{thm:XOS bounded} in Appendix~\ref{sec:XOS bounded}.  The regular case is proved as Theorem~\ref{thm:XOS regular} in Appendix~\ref{sec:XOS bounded}.

\section{Symmetric Bidders}\label{sec:symmetric bidders}
In this section, we consider symmetric bidders ($D_i=D_{i'}$ for all $i$ and $i'\in[n]$) with XOS and subadditive valuations. For XOS valuations, our goal is to improve our algorithms from Section~\ref{sec:constrained additive} to be computationally efficient under bidder symmetry. For subadditive valuations, our goal is to establish the learnability of approximately optimal mechanisms whose revenue improves as the number of bidders becomes comparable to the number of items. We only describe the results here and postpone the formal statements and proofs to Appendix~\ref{sec:symmetric appx}.
 \begin{itemize}
 	\item  \textbf{XOS valuations:} we can learn in polynomial time an approximately optimal mechanism with a polynomial number of samples when the valuations are XOS over independent items. {Our algorithm essentially estimates all the parameters needed to run the RSPM and ASPE used in~\cite{CaiZ17}. In general, it is not clear how to estimate these parameters efficiently. But when the bidders are symmetric, one only needs to consider ``symmetric parameters'' which greatly simplifies the search space and allows us to estimate all the parameters in polynomial time. }See Appendix~\ref{sec:symmetric XOS} for details.
 	\item \textbf{subadditive valuations:} when the valuations are subadditive over independent items,  the optimal revenue is at most $O\left(\frac{n}{\max\{m,n\}}\right)$ times larger than the highest revenue obtainable by an RSPM. In other words, if the number of items is within a constant times the number of bidders, an RSPM suffices to extract a constant fraction of the optimal revenue.  Applying our results for unit-demand bidders in Section~\ref{sec:unit-demand}, we can learn a nearly-optimal RSPM, which is also a good approximation to $\opt$. In fact, when the distribution for random variable $V_i(t_{ij})$ is regular for every bidder $i$ and item $j$, we can design a prior-independent mechanism that achieves a constant fraction of the optimal revenue. See Appendix~\ref{sec:symmetric subadditive} for details.
 \end{itemize}

\appendix
\section*{\huge{Appendix}}
\section{Our Mechanisms}

Here are the detailed description of the two major mechanisms we use: Sequential Posted Price Mechanism (SPM) and Anonymous Sequential Posted Price with Entry Fee Mechanism (ASPE). We also use the Rationed Sequential Posted Price Mechanism (RSPM) when bidders are not unit-demand. RSPM is almost identical to SPM except that there is an extra constraint saying that no bidder can purchase more than one item.
\begin{algorithm}[ht]
\begin{algorithmic}[1]
\REQUIRE $P_{ij}$ is the price for bidder $i$ to purchase item $j$.
\STATE $S\gets [m]$
\FOR{$i \in [n]$}
	\STATE Show bidder $i$ {the} set of available items $S$.
	 \STATE $i$ purchases her favorite bundle $S_i^{*}\in \max_{S'\subseteq S} v_i(t_i, S') - \sum_{j\in S'} P_{ij}$ and pays $\sum_{j\in S_i^{*}}P_{ij}$.
        \STATE $S\gets S\backslash S_i^{*}$.
\ENDFOR
\end{algorithmic}
\caption{{\sf Sequential Posted Price Mechanism (SPM)}}
\label{alg:seq-mech}
\end{algorithm} 

\begin{algorithm}[ht]
\begin{algorithmic}[1]
\REQUIRE A collection of prices $\{p_{j}\}_{j\in[m]}$ and a collection of entry fee functions $\{\delta_i(\cdot)\}_{i\in[n]}$ where $\delta_i: 2^{[m]}\mapsto \mathbb{R}$ is bidder $i$'s entry fee function.
\STATE $S\gets [m]$
\FOR{$i \in [n]$}
	\STATE Show bidder $i$ {the} set of available items $S$ and set the entry fee for bidder $i$ to be ${\delta_i}(S)$.
    \IF{Bidder $i$ pays the entry fee ${\delta_i}(S)$}
        \STATE $i$ receives her favorite bundle $S_i^{*}$ and pays $\sum_{j\in S_i^{*}}p_{j}$.
        \STATE $S\gets S\backslash S_i^{*}$.
    \ELSE
        \STATE $i$ gets nothing and pays $0$.
    \ENDIF
\ENDFOR
\end{algorithmic}
\caption{{\sf Anonymous Sequential Posted Price with Entry Fee Mechanism (ASPE)}}
\label{alg:aspe-mech}
\end{algorithm}

\section{Missing Details from Section~\ref{sec:unit-demand}}\label{sec:unit-demand appx}
\subsection{Unit-demand Valuations: direct access to approximate distributions}\label{sec:unit-demand Kolmogorov}
We first consider the model where we only have access to an approximate distribution $\hat{D}$. The following definition is crucial for proving our result. 

\begin{definition}
For any single dimensional distribution $\DD$ with cdf $F$, we define its revenue curve $R_{\DD}: [0,1]\mapsto \mathbb{R}_{\geq 0}$ as
	\begin{align*}
	{R}_{\DD} (q) = &\max  x\cdot \ubar{q}\cdot {F}^{-1}(1-\ubar{q}) +(1-x)\cdot \bar{q}\cdot {F}^{-1}(1-\bar{q})\\
	& \qquad\textbf{s.t. } x\cdot \ubar{q}+(1-x)\cdot \bar{q}=q\\
	&  \qquad\qquad x, \ubar{q}, \bar{q} \in [0,1]
\end{align*}
where $F^{-1}(1-p) = \sup\{x\in R: \Pr_{v\sim \DD}[v\geq x]\geq p\}$.
\end{definition}


\begin{lemma}[Folklore]
	Let ${\varphi}_{ij}(\cdot)$ and $\hat{\varphi}_{ij}(\cdot)$ be the ironed virtual value function for distribution $D_{ij}$ and $\hat{D}_{ij}$ respectively, then for any $q\in [0,1]$, ${R}_{D_{ij}}(q) = \int_{{F}^{-1}_{ij}(1-q)}^H {\varphi}(x) dF(x)$ and $R_{\hat{D}_{ij}}(q) = \int_{\hat{F}^{-1}_{ij}(1-q)}^H \hat{\varphi}(x) dF(x)$. Since the ironed virtual value function is monotonically non-decreasing, ${R}_{D_{ij}}(\cdot)$ and $R_{\hat{D}_{ij}}(\cdot)$ are concave functions. 
\end{lemma}

We provide an upper bound of the optimal revenue using $R_{D_{ij}}$ in the next Lemma. To do that, we first need the definition of the \emph{Single-Dimensional Copies Setting}.

\vspace{.1in}

\noindent\textbf{Single-Dimensional Copies Setting:} In the analysis for unit-demand bidders in~\cite{ChawlaHMS10, CaiDW16}, the optimal revenue is upper bounded by the optimal revenue in the single-dimensional copies setting defined in~\cite{ChawlaHMS10}. We use the same technique. We construct $nm$ agents, where agent $(i,j)$ has value $V_i(t_{ij})$ of being served with $t_{ij}\sim D_{ij}$, and we are only allow to use matchings, that is, for each $i$ at most one agent $(i,k)$ is served and for each $j$ at most one agent $(k,j)$ is served\footnote{This is exactly the copies setting used in~\cite{ChawlaHMS10}, if every bidder $i$ is unit-demand and has value $V_i(t_{ij})$ with type $t_i$. Notice that this unit-demand multi-dimensional setting is equivalent as adding an extra constraint, each buyer can purchase at most one item, to the original setting with subadditive bidders.}. 
Notice that this is a single-dimensional setting, as each agent's type is specified by a single number. Let $\copies$ be the optimal BIC revenue in this copies setting.

\begin{lemma}\label{lem:UB for UD rev}
For unit-demand bidders, there exists a collection of non-negative numbers $\{q_{ij}\}_{i\in[n], j\in[m]}$ satisfying $\sum_i q_{ij}\leq 1$ for all $j\in [m]$ and $\sum_j q_{ij}\leq 1$ for all $i\in [n]$, such that the optimal revenue $$\opt\leq 4\cdot \sum_{i,j} R_{D_{ij}}(q_{ij}).$$
\end{lemma}
\begin{proof}
As shown in~\cite{CaiDW16}, $\opt\leq 4 \copies$. Let $q_{ij}$ be the ex-ante probability that agent $(i,j)$ is served in the optimal mechanism for the copies setting. Chawla et al.~\cite{ChawlaHMS10} showed that $\copies\leq \sum_{i,j} R_{D_{ij}}(q_{ij})$. Our statement follows from the two inequalities above.\end{proof}

Next, we consider a convex program (Figure~\ref{fig:CP unit demand}) and argue that the value of the optimal solution of this program is at least $\frac{1}{8}$ of the optimal revenue.
\begin{figure}[ht]
\begin{minipage}{\textwidth} 
\begin{align*}\label{prog:convex ud}
&\max \sum_{i,j} R_{D_{ij}}(q_{ij})\\
\textbf{s.t. }& \sum_i q_{ij}\leq \frac{1}{2}\qquad \text{ for all $j\in[m]$}\\
& \sum_j q_{ij}\leq \frac{1}{2}\qquad \text{ for all $i\in[n]$}\\
& q_{ij}\geq 0\qquad \text{ for all $i\in[n]$ and $j\in[m]$}
\end{align*}
\end{minipage}
\caption{A Convex Program for Unit-demand Bidders with Exact Distributions.}
\label{fig:CP unit demand}
\end{figure}

\begin{lemma}\label{lem:compare exact CP with opt}
	The optimal solution of convex program in Figure~\ref{fig:CP unit demand} is at least $\frac{\opt}{8}$.
\end{lemma}
\begin{proof}
	Let $\{q'_{ij}\}$ be the collection of nonnegative numbers in Lemma~\ref{lem:UB for UD rev}. Clearly, $\left\{\frac{q_{ij}'}{2}\right\}$ is a set of feasible solution for the convex program. Since $R_{D_{ij}}(\cdot)$ is concave, $R_{D_{ij}}\left(\frac{q_{ij}'}{2}\right)\geq \frac{R_{D_{ij}}(q'_{ij})}{2} + \frac{R_{D_{ij}}(0)}{2}=\frac{R_{D_{ij}}(q'_{ij})}{2}$. Therefore, $$\sum_{i,j} R_{D_{ij}}\left(\frac{q_{ij}'}{2}\right)\geq 
	\frac{1}{2}\cdot \sum_{i,j} R_{D_{ij}}(q'_{ij})\geq \frac{\opt}{8}.$$
\end{proof}

If we know all $F_{ij}$ exactly, we can solve the convex program (Figure~\ref{fig:CP unit demand}) and use the optimal solution to construct an SPM via an approach provided in~\cite{ChawlaHMS10,CaiDW16}. The constructed sequential posted mechanism has revenue at least $\frac{1}{4}$ of the optimal value of the convex program, which is at least $\frac{\opt}{32}$. Next, we show that with only access to $\hat{F}_{ij}$, we can essentially carry out the same approach. Consider a different convex program (Figure~\ref{fig:CP unit demand approximate dist}).
\begin{figure}[ht]
\begin{minipage}{\textwidth} 
\begin{align*}
&\max \sum_{i,j} R_{\hat{D}_{ij}}(q_{ij})\\
\textbf{s.t. }& \sum_i q_{ij}\leq \frac{1}{2} + n\cdot\epsilon\qquad \text{ for all $j\in[m]$}\\
& \sum_j q_{ij}\leq \frac{1}{2}+m\cdot\epsilon \qquad \text{ for all $i\in[n]$}\\
& q_{ij}\geq 0\qquad \text{ for all $i\in[n]$ and $j\in[m]$}
\end{align*}
\end{minipage}
\caption{A Convex Program for Unit-demand Bidders with Approximate Distributions.}
\label{fig:CP unit demand approximate dist}
\end{figure}

 Not that if the support size for all $\hat{D}_{ij}$ is upper bounded by some finite number $s$, the convex program above can be rewritten as a linear program with size $\poly(n,m,s)$.
In the following Lemma, we prove that the optimal values of the two convex programs above are close. 
\begin{lemma}\label{lem:UD compare the two CP}
	Let $\{{q}^*_{ij}\}_{i\in[n],j\in[m]}$ and $\{\hat{q}_{ij}\}_{i\in[n],j\in[m]}$ be the optimal solution of the convex program in Figure~\ref{fig:CP unit demand} and~\ref{fig:CP unit demand approximate dist} respectively.	$$\sum_{i,j}R_{\hat{D}_{ij}}(\hat{q}_{ij})\geq \sum_{i,j}{R}_{D_{ij}}({q}^*_{ij})-\epsilon\cdot mn H.$$
	\end{lemma}

\begin{proof}
	We first fix some notations. For any bidder $i$ and item $j$, let $\ubar{q}^*_{ij}, \bar{q}^*_{ij}$ and $x_{ij}$ $\in[0,1]$ be the numbers satisfy that $x_{ij}\cdot \ubar{q}^*_{ij}\cdot {F_{ij}}^{-1}(1-\ubar{q}^*_{ij}) +(1-x_{ij})\cdot \bar{q}^*_{ij}\cdot {F_{ij}}^{-1}(1-\bar{q}^*_{ij})=R_{D_{ij}}(q^*_{ij})$ and $x_{ij}\cdot \ubar{q}^*_{ij} +(1-x_{ij})\cdot \bar{q}^*_{ij}=q^*_{ij}$. Let $\ubar{p}_{ij} = {F_{ij}}^{-1}(1-\ubar{q}^*_{ij})$, $\bar{p}_{ij} = {F_{ij}}^{-1}(1-\bar{q}^*_{ij})$, and  $q'_{ij} = x_{ij}\cdot \left(1-\hat{F}_{ij}(\ubar{p}_{ij})\right)+(1-x_{ij})\cdot \left(1-\hat{F}_{ij}(\bar{p}_{ij})\right)$. By the definition of $R_{\hat{D}_{ij}}(\cdot)$, \begin{equation}\label{eq:compare revenue curve}
		R_{\hat{D}_{ij}}(q_{ij}')\geq x_{ij}\cdot \left(1-\hat{F}_{ij}(\ubar{p}_{ij})\right)\cdot \ubar{p}_{ij}+(1-x_{ij})\cdot \left(1-\hat{F}_{ij}(\bar{p}_{ij})\right)\cdot \bar{p}_{ij}
	\end{equation} Since $||\hat{D}_{ij}-D_{ij}||_K\leq \epsilon$, $\hat{F}_{ij}(\ubar{p}_{ij})\in [1-\ubar{q}^*_{ij}-\epsilon, 1-\ubar{q}^*_{ij}+\epsilon]$ and $\hat{F}_{ij}(\bar{p}_{ij}) \in [1-\bar{q}^*_{ij}-\epsilon, 1-\bar{q}^*_{ij}+\epsilon]$. Hence, the RHS of inequality~(\ref{eq:compare revenue curve}) is greater than $R_{D_{ij}}(q_{ij}^*)-\epsilon\cdot H$. Therefore, $R_{\hat{D}_{ij}}(q_{ij}')\geq R_{D_{ij}}(q_{ij}^*)-\epsilon\cdot H$.
	
	Next, we argue that $\{q'_{ij}\}_{i\in[n],j\in[m]}$ is a feasible solution for the convex program in Figure~\ref{fig:CP unit demand approximate dist}. Since $1-\hat{F}_{ij}(\ubar{p}_{ij})\leq \ubar{q}^*_{ij}+\epsilon$ and $1-\hat{F}_{ij}(\bar{p}_{ij})\leq \bar{q}^*_{ij}+\epsilon$, $q'_{ij}\leq q^*_{ij}+\epsilon$. Thus, $\sum_i q'_{ij} \leq \sum_i q^*_{ij} + n\cdot \epsilon \leq \frac{1}{2}+n\cdot \epsilon$ for all $j\in[m]$. Similarly, we can prove $\sum_j q'_{ij}\leq \frac{1}{2}+m\cdot \epsilon$ for all $i\in[n]$. As $\{\hat{q}_{ij}\}_{i\in[n],j\in[m]}$ is the optimal solution for the second convex program, $\sum_{i,j}R_{\hat{D}_{ij}}(\hat{q}_{ij})\geq \sum_{i,j}R_{\hat{D}_{ij}}({q}'_{ij})\geq \sum_{i,j}{R}_{D_{ij}}({q}^*_{ij})-\epsilon\cdot mn H$.
 \end{proof}

Finally, we show how to use the optimal solution of the convex program in Figure~\ref{fig:CP unit demand approximate dist} to construct an SPM that approximates the optimal revenue well. We first provide a general transformation that turns any approximately feasible solution of convex program in Figure~\ref{fig:CP unit demand} to an SPM mechanism.

\begin{lemma}\label{lem:prices to SPM}
	 For any distribution $\DD=\times_{i\in[n], j\in[m]}\DD_{ij}$, given a collection of independent random variables $\{p_{ij}\}_{i\in[n],j\in[m]}$ such that $$\sum_{i\in[n]} \Pr_{p_{ij}, t_{ij}\sim \DD_{ij}}\left[t_{ij}\geq p_{ij}\right]\leq 1-\eta_1 \text{,\quad for all $j\in[m]$}$$ and $$\sum_{j\in[m]} \Pr_{p_{ij}, t_{ij}\sim \DD_{ij}}\left[t_{ij}\geq p_{ij}\right]\leq 1-\eta_2 \text{,\quad for all $i\in[n]$},$$ we can construct in polynomial time a randomized SPM such that the revenue under $\DD$ is at least $$\eta_1 \eta_2\cdot \sum_{i, j} \E_{p_{ij}}\left[p_{ij}\cdot \Pr_{t_{ij}\sim \DD_{ij}}\left[t_{ij}\geq p_{ij}\right]\right].$$
	 \end{lemma}

\begin{proof}
	Consider a randomized SPM that sells item $j$ to bidder $i$ at price $p_{ij}$. Notice that bidder $i$ purchases exactly item $j$ if all of the following three conditions hold: (i) for all bidders $\ell\neq i$, $t_{\ell j}$ is smaller than the corresponding price $p_{\ell j}$,  (ii) for all items $k\neq j$, $t_{ik}$ is smaller than the corresponding price $p_{ik}$, and (iii) $t_{ij}$ is greater than the corresponding price $p_{ij}$. These three conditions are independent from each other. The first condition holds with probability at least $1-\sum_{\ell\neq i} \Pr_{p_{\ell j}, t_{\ell j}\sim \DD_{\ell j}}\left[t_{\ell j}\geq p_{\ell j}\right]\geq \eta_1$. The second condition holds with probability at least $1-\sum_{k \neq j} \Pr_{p_{ik}, t_{ik}\sim \DD_{ik}}\left[t_{ik}\geq p_{ik}\right]\geq \eta_2$. When the first two conditions hold, bidder $i$ purchases item $j$ whenever she can afford it. Her expected payment is $\E_{p_{ij}}\left[p_{ij}\cdot \Pr_{t_{ij}\sim \DD_{ij}}\left[t_{ij}\geq p_{ij}\right]\right]$. Hence, the expected revenue for selling item $j$ to bidder $i$ is at least $\eta_1\eta_2\cdot \E_{p_{ij}}\left[p_{ij}\cdot \Pr_{t_{ij}\sim \DD_{ij}}\left[t_{ij}\geq p_{ij}\right]\right]$ and the total expected revenue is at least $\eta_1 \eta_2\cdot \sum_{i, j} \E_{p_{ij}}\left[p_{ij}\cdot \Pr_{t_{ij}\sim \DD_{ij}}\left[t_{ij}\geq p_{ij}\right]\right]$.
	
	\end{proof}

\begin{lemma}\label{lem:convert approx CP to SPM}
	Given any feasible solution $\{{q}_{ij}\}_{i\in[n],j\in[m]}$ of the convex program in Figure~\ref{fig:CP unit demand approximate dist}, we can construct a (randomized) SPM in polynomial time such that its revenue under $D$ is at least $\left(\frac{1}{4}-(n+m)\cdot \epsilon\right)\cdot \left( \sum_{i,j}R_{\hat{D}_{ij}}(q_{ij})-\epsilon\cdot nmH\right)$. 
\end{lemma}

\begin{proof}
	We first fix some notations. For any bidder $i$ and item $j$, let $\ubar{q}_{ij}, \bar{q}_{ij}$ and $x_{ij}$ $\in[0,1]$ be the numbers satisfying $x_{ij}\cdot \ubar{q}_{ij}\cdot \hat{F}_{ij}^{-1}(1-\ubar{q}_{ij}) +(1-x_{ij})\cdot \bar{q}_{ij}\cdot \hat{F}_{ij}^{-1}(1-\bar{q}_{ij})=R_{\hat{D}_{ij}}(q_{ij})$ and $x_{ij}\cdot \ubar{q}_{ij} +(1-x_{ij})\cdot \bar{q}_{ij}=q_{ij}$. We use $p_{ij}$ to denote a random variable that is $\ubar{p}_{ij} = \hat{F}_{ij}^{-1}(1-\ubar{q}_{ij})$ with probability $x_{ij}$ and $\bar{p}_{ij} = \hat{F}_{ij}^{-1}(1-\bar{q}_{ij})$ with probability $1-x_{ij}$.  
	
	Next, we construct a randomized SPM based on $\{p_{ij}\}_{i\in[n],j\in[m]}$ according to Lemma~\ref{lem:prices to SPM}. Note that $$\sum_{i\in[n]} \Pr_{p_{ij}, t_{ij}\sim D_{ij}}\left[t_{ij}\geq p_{ij}\right]\leq \sum_{i\in[n]} \left(\Pr_{p_{ij}, t_{ij}\sim \hat{D}_{ij}}\left[t_{ij}\geq p_{ij}\right]+\epsilon\right)=\sum_{i\in[n]} q_{ij}+n\epsilon\leq \frac{1}{2}+2n\epsilon$$ for all item $j$, and $$\sum_{j\in[m]} \Pr_{p_{ij}, t_{ij}\sim D_{ij}}\left[t_{ij}\geq p_{ij}\right]\leq \sum_{j\in[m]} \left(\Pr_{p_{ij}, t_{ij}\sim \hat{D}_{ij}}\left[t_{ij}\geq p_{ij}\right]+\epsilon\right)=\sum_{i\in[m]} q_{ij}+m\epsilon\leq \frac{1}{2}+2m\epsilon$$ for all bidder $i$. Hence, we can construct in polynomial time a randomized SPM with revenue at least \begin{align*} &\left(\frac{1}{2}-2n\epsilon\right)\left(\frac{1}{2}-2m\epsilon\right)\cdot \sum_{i, j} \E_{p_{ij}}\left[p_{ij}\cdot \Pr_{t_{ij}\sim D_{ij}}\left[t_{ij}\geq p_{ij}\right]\right]\\
 	\geq &  \left(\frac{1}{4}-(n+m)\epsilon\right)\sum_{i, j} \E_{p_{ij}}\left[p_{ij}\cdot \left(\Pr_{t_{ij}\sim \hat{D}_{ij}}\left[t_{ij}\geq p_{ij}\right]-\epsilon\right)\right]\\
 	\geq & \left(\frac{1}{4}-(n+m)\epsilon\right)\sum_{i, j} \left(R_{\hat{D}_{ij}}(q_{ij})-\epsilon\cdot nmH\right)
 \end{align*}
	The first inequality is because $\left|\left|D_{ij}-\hat{D}_{ij}\right|\right|_K\leq \epsilon$, and the second inequality is because $p_{ij}$ is upper bounded by $H$ and $\E_{p_{ij}}\left[p_{ij}\cdot \Pr_{t_{ij}\sim \hat{D}_{ij}}\left[t_{ij}\geq p_{ij}\right]\right]=R_{\hat{D}_{ij}}(q_{ij})$ by the definition of $p_{ij}$.

\notshow{	 \begin{align*}
	&x_{ij}\cdot \ubar{p}_{ij}\cdot \left(1-F_{ij}(\ubar{p}_{ij})\right)+(1-x_{ij})\cdot\bar{p}_{ij}\cdot \left(1-F_{ij}(\bar{p}_{ij})\right)\\
	\geq &x_{ij}\cdot \ubar{p}_{ij}\cdot (\ubar{q}_{ij}-\epsilon)+(1-x_{ij})\cdot\bar{p}_{ij}\cdot (\bar{q}_{ij}-\epsilon)\geq R_{\hat{D}_{ij}}(q_{ij})-\epsilon\cdot H.\end{align*} 
	Therefore, the expected revenue of the randomized SPM is at least $$\left(\frac{1}{2}-2n\cdot \epsilon\right)\cdot \left(\frac{1}{2}-2m\cdot \epsilon\right)\cdot \left( \sum_{i,j}R_{\hat{D}_{ij}}(q_{ij})-\epsilon\cdot nmH\right)\geq \left(\frac{1}{4}-(n+m)\cdot \epsilon\right)\cdot \left( \sum_{i,j}R_{\hat{D}_{ij}}(q_{ij})-\epsilon\cdot nmH\right).$$} 
	\end{proof}

\begin{theorem}\label{thm:UD Kolmogorov}
For unit-demand bidders, given distributions $\hat{D}_{ij}$ where $\left|\left|\hat{D}_{ij}-D_{ij}\right|\right|_K\leq \epsilon$ for all $i\in[n]$ and $j\in[m]$, there is a polynomial time algorithm that constructs a randomized SPM whose revenue under $D$ is at least $\left(\frac{1}{4}-(n+m)\cdot \epsilon\right)\cdot\left(\frac{\opt}{8}-2\epsilon\cdot mnH\right)$.
\end{theorem}

\begin{proof}
	Our algorithm first computes the optimal solution $\{\hat{q}_{ij}\}_{i\in[n],j\in[m]}$ for the convex program in Figure~\ref{fig:CP unit demand approximate dist}, then  constructs a randomized SPM based on $\{\hat{q}_{ij}\}_{i\in[n],j\in[m]}$ using Lemma~\ref{lem:convert approx CP to SPM}. It is not hard to see that our algorithm runs in polynomial time. By chaining the inequalities in Lemma~\ref{lem:compare exact CP with opt},~\ref{lem:UD compare the two CP} and~\ref{lem:convert approx CP to SPM}, we can argue that the revenue of our mechanism is at least $\left(\frac{1}{4}-(n+m)\cdot \epsilon\right)\cdot\left(\frac{\opt}{8}-2\epsilon\cdot mnH\right)$.
\end{proof}

\subsection{Unit-demand Valuations: sample access to bounded distributions} 

When the distributions $D_{ij}$ are all bounded, the following theorem provides the sample complexity.
\begin{theorem}\cite{MorgensternR16}\label{thm:UD bounded}
	When $D_{ij}$ is supported on $[0,H]$ for all bidder $i$ and item $j$, the sample complexity for $(\epsilon,\delta)$-uniformly learning the revenue of SPMs for unit-demand bidders is $O\left(\left(\frac{1}{\epsilon}\right)^2 \left(m^2 n\log n\log \frac{1}{\epsilon} + \log \frac{1}{\delta}\right)\right)$. That is, with probability $1-\delta$, the empirical revenue based on the samples for any SPM is within $\epsilon\cdot H$ of its true expected revenue. 
	{Moreover, with the same number of samples, there is a polynomial time algorithm that learns an SPM whose revenue is at least $\frac{\opt}{144}-\epsilon H$ with probability $1-\delta$. }\end{theorem}

\subsection{Unit-demand Valuations: sample access to regular distributions}\label{sec:unit-demand regular}
In this section, we show there exists a polynomial time algorithm that learns an SPM whose revenue is at least a constant fraction of the optimal revenue with polynomial in $n$ and $m$ samples. Note that unlike in the previous two models, the error of our learning algorithm is only multiplicative when the distributions are regular. First, we present a Lemma regarding the revenue curve function for regular distributions.

\begin{lemma}\cite{CaiD11b}\label{lem:regular revenue curve concave}
	For any regular distribution $F$, let $R_F(\cdot)$ be the corresponding revenue curve. For any $0<q'\leq q\leq  p < 1$, $$(1-p)\cdot R_F(q')\leq  R_F(q).$$
	\end{lemma}
 Throughout this section, we use $Z$ to denote $\max\{m,n\}$ and $C$ to be a constant that will be specified later. Using Lemma~\ref{lem:regular revenue curve concave}, we show in the next Lemma that restricting $q_{ij}$ to be at least $\frac{1}{ CZ}$ does not affect the objective value of the convex program in Figure~\ref{fig:CP unit demand} by too much.

\begin{lemma}\label{lem:lowering high prices}
	Suppose $\{q_{ij}^*\}_{i\in[n], j\in[m]}$ is the optimal solution of the convex program in Figure~\ref{fig:CP unit demand}. Let $q'_{ij}=\max\{\frac{1}{CZ}, q^*_{ij}\}$, then $\sum_{i,j} R_{D_{ij}}(q'_{ij})\geq \left(1-\frac{1}{ CZ}\right)\cdot \sum_{i,j} R_{D_{ij}}(q^*_{ij})\geq \left(1-\frac{1}{ CZ}\right)\cdot \frac{\opt}{8}$.
\end{lemma}
\begin{proof}
	According to Lemma~\ref{lem:compare exact CP with opt}, $\sum_{i,j} R_{D_{ij}}(q^*_{ij})\geq \frac{\opt}{8}$. So to prove the statement, it suffices to argue that for any $i$ and $j$, $R_{D_{ij}}(q'_{ij})\geq \left(1-\frac{1}{CZ}\right)\cdot R_{D_{ij}}(q^*_{ij})$. If $q^*_{ij} = q'_{ij}$, this inequality clearly holds. If $q^*_{ij} \neq q'_{ij}$, $q^*_{ij}\leq q'_{ij}=\frac{1}{CZ}$. Since $F_{ij}$ is regular, we can apply Lemma~\ref{lem:regular revenue curve concave} to $q'_{ij}$ and $q^*_{ij}$ and obtain inequality $R_{D_{ij}}(q'_{ij})\geq \left(1-\frac{1}{C Z}\right)\cdot R_{D_{ij}}(q^*_{ij})$.
\end{proof}

Using Lemma~\ref{lem:lowering high prices}, we argue how to compute in polynomial time an approximately optimal SPM. Suppose $D'_{ij}$ is the distribution that we obtain after truncating $D_{ij}$ at a threshold $H_{ij}$\footnote{Let $t_{ij}\sim D_{ij}$, then $\min\{t_{ij},H_{ij}\}$ is the corresponding truncated random variable drawn from $D'_{ij}$.}, and we have direct access to a discrete distribution $\hat{D}'_{ij}$ such that $\left|\left|\hat{D}'_{ij}-D'_{ij}\right|\right|_K\leq \epsilon$ for all $i$ and $j$. We show in the following Lemma that the optimal solution of a convex program similar to the one in Figure~\ref{fig:CP unit demand approximate dist} but for $\{\hat{D}'_{ij}\}_{i\in[n],j\in[m]}$ can guide us to design an approximately optimal SPM under $D$ in polynomial time. As we have sample access to $D$, we will argue later that a polynomial number of samples suffices to generate $\{\hat{D}'_{ij}\}_{i\in[n],j\in[m]}$.
\notshow{\begin{lemma}\label{lem:UD low prices suffice}
	There exists a collection of deterministic prices $\{p'_{ij}\}_{i\in[n],j\in[m]}$ satisfying $F_{ij}(p'_{ij})\leq 1-\frac{1}{C\cdot Z}$ for all $i\in[n]$ and $j\in[m]$, such that the corresponding SPM achieves revenue of $\left(\frac{1}{2}-\frac{1}{C}\right)^2\cdot \left(1-\frac{1}{C\cdot Z}\right)\cdot \frac{\opt}{8}$.
\end{lemma} 
\begin{proof}
Let $\{q_{ij}^*\}_{i\in[n], j\in[m]}$ be the optimal solution of the convex program in Figure~\ref{fig:CP unit demand} and $q'_{ij}=\max\{\frac{1}{C\cdot Z}, q^*_{ij}\}$. For every $i$ and $j$, let $p'_{ij}$ to be $F^{-1}_{ij}(1-q'_{ij})$. Clearly, $F_{ij}(p'_{ij})\leq 1-\frac{1}{C\cdot Z}$. 

As $D_{ij}$ is regular, $R_{D_{ij}}(q'_{ij}) = p'_{ij} q'_{ij}$. Therefore, $\sum_{i,j} p'_{ij}q'_{ij}=\sum_{i,j} R_{D_{ij}}(q'_{ij})\geq \left(1-\frac{1}{CZ} \right)\cdot \frac{\opt}{8}$ according to Lemma~\ref{lem:UD low prices suffice}. Also, $\sum_i q'_{ij}\leq \frac{1}{2}+\frac{n}{CZ}\leq \frac{1}{2}+\frac{1}{C}$ for all item $j$, and $\sum_j q'_{ij}\leq \frac{1}{2}+\frac{m}{CZ}\leq \frac{1}{2}+\frac{1}{C}$ for all bidder $i$. According to Lemma~\ref{lem:prices to SPM}, the SPM with prices $\{p'_{ij}\}_{i\in[n],j\in[m]}$ has revenue at least $\left(\frac{1}{2}-\frac{1}{C}\right)^2\cdot \sum_{i,j} p'_{ij}q'_{ij}\geq \left(\frac{1}{2}-\frac{1}{C}\right)^2\cdot  \left(1-\frac{1}{CZ} \right)\cdot \frac{\opt}{8}$.

\end{proof}

\begin{lemma}\label{lem:UD additive bound}
	Let $\{H_{ij}\}_{i\in[n],j\in[m]}$ be a collection of positive numbers satisfying $F_{ij}(H_{ij})\in [1-\frac{1}{C\cdot Z}, 1-\frac{1}{3C\cdot Z}]$ for all $i\in[n]$ and $j\in[m]$. With $O\left(\left(\frac{1}{\epsilon}\right)^2 \left(m^2 n\log n\log \frac{1}{\epsilon} + \log \frac{1}{\delta}\right)\right)$ samples, we can learn an SPM with probability $1-\delta$ that achieves revenue at least $\left(\frac{1}{2}-\frac{1}{C}\right)^2\cdot \left(1-\frac{1}{C\cdot Z}\right)\cdot \frac{\opt}{8}-\epsilon\cdot\max_{i,j} H_{ij}$. \yangnote{With the same number of samples, we can learn in polynomial time an SPM with probability $1-\delta$ that achieves revenue at least $\left(\frac{1}{2}-\frac{1}{C}\right)^2\cdot \left(1-\frac{1}{C\cdot Z}\right)\cdot \frac{\opt}{48}-\epsilon\cdot\max_{i,j} H_{ij}$.}
\end{lemma}
\begin{proof}
	First, we truncate each $D_{ij}$ at $H_{ij}$ to create a bounded distribution $D'_{ij}$. This is straightforward, as we only need to cap the value for any sample from $D_{ij}$ at $H_{ij}$. According to Lemma~\ref{lem:UD low prices suffice}, there exists an SPM that achieves revenue at least $\left(\frac{1}{2}-\frac{1}{C}\right)^2\cdot \left(1-\frac{1}{C\cdot Z}\right)\cdot \frac{\opt}{8}$ under $D'=\times_{i,j} D'_{ij}$, as the price $p_{ij}$ used in the SPM is less than $H_{ij}$ for all $i,j$. Combining this observation with Theorem~\ref{thm:UD bounded}, we know that if we take $O\left(\left(\frac{1}{\epsilon}\right)^2 \left(m^2 n\log n\log \frac{1}{\epsilon} + \log \frac{1}{\delta}\right)\right)$ samples we can learn with probability $1-\delta$ an SPM whose revenue is at least $\left(\frac{1}{2}-\frac{1}{C}\right)^2\cdot \left(1-\frac{1}{C\cdot Z}\right)\cdot \frac{\opt}{8}-\epsilon\cdot\max_{i,j} H_{ij}$. Similarly, we can prove the computational friendly version of this claim. \end{proof}}
	
\begin{lemma}\label{lem:UD additive bound}
	Let $\{H_{ij}\}_{i\in[n],j\in[m]}$ be a collection of positive numbers satisfying $F_{ij}(H_{ij})\in [1-\frac{1}{C\cdot Z}, 1-\frac{1}{3C\cdot Z}]$ for all $i\in[n]$ and $j\in[m]$. Let $D'_{ij}$ be the distribution of the random variable $\min\{t_{ij}, H_{ij}\}$ where $t_{ij}\sim D_{ij}$, and $\hat{D}_{ij}'$ be a discrete distribution such that $\left|\left|\hat{D}'_{ij}-D'_{ij}\right|\right|_K\leq \epsilon$ for all $i\in[n]$ and $j\in[m]$. Suppose $s$ is an upper bound of the support size for any distribution $\hat{D}'_{ij}$, then given direct access to $\hat{D}_{ij}'$, we can compute in time polynomial in $n$, $m$ and $s$ a randomized SPM that achieves revenue at least $\left(\frac{1}{2}-\frac{1}{C}-2n\epsilon\right)\cdot\left(\frac{1}{2}-\frac{1}{C}-2m\epsilon\right)\cdot \left(\left(1-\frac{1}{ CZ}\right)\cdot \frac{\opt}{8}-2\epsilon\cdot nmH\right)$ under $D$, where $H= \max_{i,j} H_{ij}$. 
\end{lemma}
\begin{proof}
Consider the following convex program:

\begin{minipage}{\textwidth} 
\begin{align*}
&\max \sum_{i,j} R_{\hat{D}_{ij}'}(q_{ij})\\
\textbf{s.t. }& \sum_i q_{ij}\leq \frac{1}{2} + \frac{1}{C} +n\cdot\epsilon\qquad \text{ for all $j\in[m]$}\\
& \sum_j q_{ij}\leq \frac{1}{2}+\frac{1}{C} +m\cdot\epsilon \qquad \text{ for all $i\in[n]$}\\
& q_{ij}\geq 0\qquad \text{ for all $i\in[n]$ and $j\in[m]$}
\end{align*}
\end{minipage}
\vspace{.1in}

Let $\{q_{ij}^*\}_{i\in[n], j\in[m]}$ be the optimal solution of the convex program in Figure~\ref{fig:CP unit demand} and $q'_{ij}=\max\{\frac{1}{C\cdot Z}, q^*_{ij}\}$. For every $i$ and $j$, let $p'_{ij}=F^{-1}_{ij}(1-q'_{ij})$ and $\tilde{q}_{ij} = \Pr_{t_{ij}\sim \hat{D}'_{ij}}\left[t_{ij}\geq p'_{ij}\right]$. By the definition of $H_{ij}$, $p'_{ij}\leq H_{ij}$, so \begin{align}\label{ineq:LB inequality}
p'_{ij}\tilde{q}_{ij}\geq p'_{ij}\left(\Pr_{t_{ij}\sim {D}'_{ij}}\left[t_{ij}\geq p'_{ij}\right]-\epsilon\right)\geq p'_{ij} q'_{ij}-\epsilon\cdot H_{ij}=R_{D_{ij}}(q'_{ij})-\epsilon\cdot H_{ij}.
 \end{align}
$p'_{ij} q'_{ij}$ equals to $R_{D_{ij}}(q'_{ij})$ because $D_{ij}$ is a regular distribution.

Next, we argue that $\{\tilde{q}_{ij}\}_{i\in[n],j\in[m]}$ is a feasible solution of the convex program above. Observe that $$\sum_{i} \tilde{q}_{ij} \leq \sum_i q'_{ij}+n\epsilon \leq \sum_i \left(q^*_{ij}+\frac{1}{CZ}\right)+n\epsilon\leq \frac{1}{2} + \frac{1}{C} +n\epsilon$$ for all item $j\in[m]$ and $$\sum_{j} \tilde{q}_{ij} \leq \sum_j q'_{ij}+n\epsilon \leq \sum_j \left(q^*_{ij}+\frac{1}{CZ}\right)+m\epsilon\leq \frac{1}{2}+\frac{1}{C} +m\epsilon$$ for all bidder $i\in[n]$. 

 Let $\widehat{\opt}$ be the optimal solution of the convex program above. As $\{\tilde{q}_{ij}\}_{i\in[n],j\in[m]}$ is a feasible solution,
$$\widehat{\opt}\geq \sum_{i,j} R_{\hat{D}_{ij}'}(\tilde{q}_{ij})\geq \sum_{i,j} p'_{ij}\tilde{q}_{ij}\geq \sum_{i,j} R_{D_{ij}}(q'_{ij})-\epsilon\cdot nmH\geq \left(1-\frac{1}{ CZ}\right)\cdot \frac{\opt}{8}-\epsilon\cdot nmH.$$
The second last inequality is due to inequality~(\ref{ineq:LB inequality}) and the last inequality is due to Lemma~\ref{lem:lowering high prices}.

So far, we have argued that the optimal solution of our convex program has value close to the $\opt$. We will show in the second part of the proof that using the optimal solution of our convex program, we can construct an SPM whose revenue under $D$ is close to $\widehat{\opt}$. Let $\hat{q}_{ij}$ be the optimal solution of the convex program above and $\hat{p}_{ij}$ be the corresponding random price, that is, $\Pr_{\hat{p}_{ij}, t_{ij}\sim \hat{D}_{ij}'}\left[t_{ij}\geq \hat{p}_{ij}\right]=\hat{q}_{ij}$ and $R_{\hat{D}_{ij}'}(\hat{q}_{ij}) = \E_{\hat{p}_{ij}}\left[\hat{p}_{ij}\cdot \Pr_{t_{ij}\sim \hat{D}_{ij}'}\left[t_{ij}\geq \hat{p}_{ij}\right]\right]$. As $\hat{p}_{ij}\leq H_{ij}$, $$\Pr_{\hat{p}_{ij}, t_{ij}\sim {D}_{ij}}\left[t_{ij}\geq \hat{p}_{ij}\right]=\Pr_{\hat{p}_{ij}, t_{ij}\sim {D}'_{ij}}\left[t_{ij}\geq \hat{p}_{ij}\right]\in  [\hat{q}_{ij}-\epsilon,\hat{q}_{ij}+\epsilon].$$ Therefore, for all item $j$ $$\sum_{i} \Pr_{\hat{p}_{ij}, t_{ij}\sim {D}_{ij}}\left[t_{ij}\geq \hat{p}_{ij}\right]\leq \sum_{i} \Pr_{\hat{p}_{ij}, t_{ij}\sim \hat{D}_{ij}'}\left[t_{ij}\geq \hat{p}_{ij}\right] + n\epsilon = \sum_i \hat{q}_{ij} +n\epsilon \leq \frac{1}{2} + \frac{1}{C} +2n\epsilon$$ and for all bidder $i$	$$\sum_{j} \Pr_{\hat{p}_{ij}, t_{ij}\sim {D}_{ij}}\left[t_{ij}\geq \hat{p}_{ij}\right]\leq \sum_{j} \Pr_{\hat{p}_{ij}, t_{ij}\sim \hat{D}_{ij}'}\left[t_{ij}\geq \hat{p}_{ij}\right] + m\epsilon = \sum_j \hat{q}_{ij} +m\epsilon \leq \frac{1}{2} + \frac{1}{C} +2m\epsilon.$$ According to Lemma~\ref{lem:prices to SPM}, we can construct a randomized SPM with $\{\hat{p}_{ij}\}_{i\in[n],j\in[m]}$ whose revenue is at least $\left(\frac{1}{2}-\frac{1}{C}-2n\epsilon\right)\cdot\left(\frac{1}{2}-\frac{1}{C}-2m\epsilon\right)\cdot \sum_{i,j} \E_{\hat{p}_{ij}}\left[\hat{p}_{ij}\cdot \Pr_{t_{ij}\sim {D}_{ij}}\left[t_{ij}\geq \hat{p}_{ij}\right]\right]$ under $D$. 
Clearly, $$\E_{\hat{p}_{ij}}\left[\hat{p}_{ij}\cdot \Pr_{t_{ij}\sim {D}_{ij}}\left[t_{ij}\geq \hat{p}_{ij}\right]\right]\geq \E_{\hat{p}_{ij}}\left[\hat{p}_{ij}\cdot \left(\Pr_{t_{ij}\sim \hat{D}'_{ij}}\left[t_{ij}\geq \hat{p}_{ij}\right]-\epsilon\right)\right] \geq R_{\hat{D}'_{ij}}(\hat{q}_{ij})-\epsilon\cdot H_{ij}.$$
 Therefore, the revenue of the constructed randomized SPM under $D$ is at least \begin{align*} &\left(\frac{1}{2}-\frac{1}{C}-2n\epsilon\right)\cdot\left(\frac{1}{2}-\frac{1}{C}-2m\epsilon\right)\cdot \left(\widehat{\opt}-\epsilon\cdot nmH\right)\\
 	\geq & \left(\frac{1}{2}-\frac{1}{C}-2n\epsilon\right)\cdot\left(\frac{1}{2}-\frac{1}{C}-2m\epsilon\right)\cdot \left(\left(1-\frac{1}{ CZ}\right)\cdot \frac{\opt}{8}-2\epsilon\cdot nmH\right).
 \end{align*}
It is not hard to see that both $\{\hat{q}_{ij}\}_{i\in[n],j\in[m]}$ and $\{\hat{p}_{ij}\}_{i\in[n],j\in[m]}$ can be computed in time polynomial in $n$, $m$ and $s$.
	\end{proof}

	When $\epsilon$ is small enough, the additive error in Lemma~\ref{lem:UD additive bound} can be converted into a multiplicative error. Next, we argue that with a polynomial number of samples, we can learn $\{H_{ij}\}_{i\in[n],j\in[m]}$ and $\{\hat{D}'_{ij}\}_{i\in[n],j\in[m]}$ with enough accuracy.

\begin{theorem}\label{thm:UD regular}
If for all bidder $i$ and item $j$, $D_{ij}$ is a regular distribution, we can learn in polynomial time with probability $1-\delta$ a randomized SPM whose revenue is at least $\frac{\opt}{33}$ with $O\left(Z^2m^2 n^2\cdot \log \frac{nm}{\delta}\right)$ ($Z=\max\{m,n\}$) samples. 
\end{theorem}
\begin{proof}
	First, if we take $O\left(C^2\cdot Z^2\cdot \log \frac{nm}{\delta}\right)$ samples from each $D_{ij}$, we can find an $H_{ij}$ such that $F_{ij}(H_{ij})$ lies in$[1-\frac{1}{CZ}, 1-\frac{1}{3CZ}]$ with probability $1-\frac{\delta}{2nm}$. By the union bound, the probability that all $H_{ij}$ satisfy the requirement is at least $1-\frac{\delta}{2}$. From now on, we assume $F_{ij}(H_{ij})\in[1-\frac{1}{C\cdot Z}, 1-\frac{1}{3C\cdot Z}]$ for all $i$ and $j$. 
	Observe that $\opt \geq \max_{i,j} H_{ij}\cdot \frac{1}{3C\cdot Z}$, as the expected revenue for selling item $j$ to bidder $i$ at price $H_{ij}$ is at least $ \frac{H_{ij}}{3C\cdot Z}$.
	 Therefore, there exists sufficiently large constant $d$ and $C$, if $\epsilon=\frac{1}{d\cdot Z nm}$ the randomized SPM learned in Lemma~\ref{lem:UD additive bound} has revenue at least $\frac{\opt}{33}$.
	  According to the Dvoretzky-Kiefer-Wolfowitz (DKW) inequality~\cite{DvoretzkyKW56}, if we take $O\left(d^2 Z^2n^2m^2\cdot \log \frac{nm}{\delta} \right)$ samples from $D'_{ij}$ (we can take samples from $D_{ij}$ then cap the samples at $H_{ij}$) and let $\hat{D}'_{ij}$ be the uniform distribution over the samples, $\left|\left|D'_{ij}-\hat{D}'_{ij}\right|\right|_K\leq \frac{1}{d\cdot Z nm}$ with probability $1-\frac{\delta}{2nm}$. 
	  By the union bound,  $\left|\left|D'_{ij}-\hat{D}'_{ij}\right|\right|_K\leq \frac{1}{d\cdot Z nm}$ for all $i\in[n]$ and $j\in[m]$ with probability at least $1-\delta/2$.
	   Finally, by another union bound,  the $H_{ij}$ and $\hat{D}'_{ij}$ we learned from $O\left( Z^2n^2m^2\cdot \log \frac{nm}{\delta} \right)$ samples satisfy $F_{ij}(H_{ij})\in[1-\frac{1}{C\cdot Z}, 1-\frac{1}{3C\cdot Z}]$ and $\left|\left|D'_{ij}-\hat{D}'_{ij}\right|\right|_K\leq \frac{1}{d\cdot Z nm}$ for all $i$ and $j$ with probability at least $1-\delta$.
	   In other words, we can learn a randomized SPM whose revenue is at least $\frac{\opt}{33}$ with probability at least $1-\delta$ using $O\left( Z^2n^2m^2\cdot \log \frac{nm}{\delta} \right)$ samples.
	    Furthermore, the support size of any $\hat{D}'_{ij}$ is at most $O\left( Z^2n^2m^2\cdot \log \frac{nm}{\delta} \right)$ samples, so our learning algorithm runs in time polynomial in $n$ and $m$.
	
	
\end{proof}

\section{Missing Details from Section~\ref{sec:additive}}
\subsection{Additive Valuations: sample access to bounded distributions}\label{sec:additive bounded}
As shown by Goldner and Karlin~\cite{GoldnerK16}, one sample suffices to design a mechanism that approximates $\brev$. The idea is to use the VCG with entry fee mechanism but replace the entry fee $e_i(b_{-i},D_i)$ for bidder $i$ with $e_i(b_{-i},s_i)=\sum_{j\in[m]}(s_{ij}-\max_{k\neq i} b_{kj})^+$, where $s_i$ is a sample drawn from $D_i$. It is easy to argue that for any $b_{-i}$, over the randomness of the sample $s_i$ and bidder $i$'s real type $t_i$, the event that $e_i(b_{-i},s_i)\geq e_i(b_{-i},D_i)$ and bidder $i$ accepts the entry fee $e_i(b_{-i},s_i)$ happens with probability at least $\frac{1}{8}$. As $\frac{1}{2}\cdot \sum_{i\in[n]} \E_{t}[e_i(t_{-i},D_i)]=\brev$, the expected revenue (over the randomness of the types and the samples) from their mechanism is at least $\frac{1}{8}\cdot \sum_{i\in[n]} \E_{t}[e_i(t_{-i},D_i)]=\frac{\brev}{4}$. Next, we show how to learn a mechanism that approximates $\srev$.

\begin{lemma}\label{lem:additive srev}
	When $D_{ij}$ is supported on $[0,H]$ for all bidder $i$ and item $j$, the sample complexity for $(\epsilon,\delta)$-uniformly learning the revenue of SPMs for additive bidders is $O\left(\left(\frac{1}{\epsilon}\right)^2 \left(m^2 n\log n\log \frac{1}{\epsilon} + \log \frac{1}{\delta}\right)\right)$. Moreover, we can learn in polynomial time an SPM whose revenue is at least $\frac{\srev}{4}- \frac{3\epsilon}{2}\cdot H$ with probability $1-\delta$ given the same number of samples. \end{lemma}
\begin{proof}

The first half of the Lemma was proved by Morgenstern and Roughgarden~\cite{MorgensternR16}.
 We show how to prove the second half of the claim.
Let $\opt_j$ be the optimal revenue for selling item $j$.
 By the prophet inequality~\cite{Samuel-cahn84}, there exists an SPM for selling item $j$ with a collection of prices $\{p_{ij}\}_{i\in[n]}$ that achieves revenue at least $\opt_j/2$.
  As the bidders are additive, if we run the SPMs for selling each item simultaneously, the expected revenue is exactly the sum of the revenue of the SPM mechanisms for auctioning a single item.
   Note that the simultaneous SPM is indeed a SPM for selling all items. 
   Hence, there exists an SPM that achieves revenue at least $\opt/2$.
    Since the sample complexity for $(\epsilon,\delta)$-uniformly learning the revenue of SPMs is $O\left(\left(\frac{m}{\epsilon}\right)^2 \left(n\log n\log \frac{1}{\epsilon} + \log \frac{1}{\delta}\right)\right)$, the empirical revenue induced by the samples is within $\epsilon\cdot H$ of the true expected revenue with probability $1-\delta$ for any SPM. 

We use $ER_{opt}$ to denote the optimal empirical revenue obtained by any SPM.
 If we apply the prophet inequality to the empirical distribution, we can construct an SPM whose empirical revenue $ER$ is at least $ER_{opt}/2$. Notice that $ER_{opt}$ is at most $\epsilon\cdot H$ less than the optimal true expected revenue obtained by any SPM, which is at least $\opt/2$. Combining the two inequalities above, we have $ER\geq \opt/4-\epsilon/2\cdot H$ with probability $1-\delta$. Also, the true expected revenue of our SPM is at least $ER-\epsilon\cdot H$, so our SPM achieves expected revenue at least $\frac{\opt}{4}-\frac{3\epsilon}{2}\cdot H$ with probability  $1-\delta$. \end{proof}

Now we are ready to prove our Theorem for additive bidders when their valuations are bounded.
\begin{theorem}\label{thm:additive bounded}
	When the bidders have additive valuations and $D_{ij}$ is supported on $[0,H]$ for all bidder $i$ and item $j$, we can learn in polynomial time a mechanism whose expected revenue is at least $\frac{\opt}{32}-{\epsilon}\cdot H$ with probability $1-\delta$ given $$O\left(\left(\frac{m}{\epsilon}\right)^2 \cdot\left(n\log n\log \frac{1}{\epsilon}+\log\frac{1}{\delta} \right)\right)$$ samples from $D$. 
\end{theorem}
\begin{proof}
	According to Lemma~\ref{lem:additive srev}, we can learn a mechanism whose revenue is at least $\frac{\srev}{4}-\frac{\epsilon}{24}\cdot H$ with probability $1-\delta$ given $O\left(\left(\frac{m}{\epsilon}\right)^2\cdot \left(n\log n\log \frac{1}{\epsilon}+\log\frac{1}{\delta} \right)\right)$ samples. As we explained in the beginning of this section, with one sample from the distribution we can construct a randomized mechanism whose expected revenue is at least $\frac{\brev}{4}$. Therefore, the better of our two mechanisms has expected revenue at least $\frac{\opt}{32}-{\epsilon}\cdot H$ with probability $1-\delta$.
	\end{proof}
	
\subsection{Additive Valuations: direct access to approximate distributions}\label{sec:additive Kolmogorov}
	
In this section, we discuss how to learn an approximately optimal mechanism for additive bidders when we are given direct access to approximate value distributions. Again, we first show how to learn a mechanism whose revenue approximates $\srev$ then we provide another mechanism whose revenue approximates $\brev$.

\begin{lemma}\label{lem:SREV kolmogorov}
	For additive bidders, given distributions $\hat{D}_{ij}$ where $\left|\left|\hat{D}_{ij}-D_{ij}\right|\right|_K\leq \epsilon$ for all $i\in[n]$ and $j\in[m]$, there is a polynomial time algorithm that constructs a randomized SPM whose revenue under $D$ is at least $\left(\frac{1}{4}- \epsilon\cdot n \right)\cdot\left(\frac{\srev}{8}-2\epsilon\cdot mnH\right)$.
\end{lemma}
\begin{proof}
	Let $\opt_j$ be the optimal revenue for selling item $j$.
	 As the bidders are additive, if we can construct a randomized SPM $M_j$ for every item $j$ such that its expected revenue under $D$ is at least $\left(\frac{1}{4}- \epsilon\cdot n \right)\cdot\left(\frac{\opt_j}{8}-2\epsilon\cdot nH\right)$, running these $m$ randomized SPMs in parallel generates expected revenue at least $$\sum_{j\in[m]} \left(\frac{1}{4}- \epsilon\cdot n \right)\cdot\left(\frac{\opt_j}{8}-2\epsilon\cdot nH\right) = \left(\frac{1}{4}- \epsilon\cdot n \right)\cdot\left(\frac{\srev}{8}-2\epsilon\cdot mnH\right)$$ under $D$.
	  Due to Theorem~\ref{thm:UD Kolmogorov}, we can construct in polynomial time such a randomized SPM $M_j$ for each item $j$ based on $\times_{i\in[n]} \hat{D}_{ij}$.
\end{proof}

Next, we show how to choose the entry fee based on $\hat{D} = \times_{i, j} \hat{D}_{ij}$, so that the VCG with entry fee mechanism has revenue that approximates $\brev$ under the true distribution $D$. More specifically, we use the median of $i$'s utility under $\hat{D}_i = \times_{j\in[m]} \hat{D}_{ij}$ as bidder $i$'s entry fee. We prove the result in two steps. We first show that if we can use an entry fee function such that every bidder $i$ accepts her entry fee with probability between $[1/2-\eta,1/2]$ for any possible bid profiles $b_{-i}$ of the other bidders, the expected revenue is at least $(1/2-\eta)\cdot \brev$. Second, we show how to compute in polynomial time such entry fee functions with $\eta=O(m\epsilon)$ based on $\hat{D}$. 

\begin{lemma}\label{lem:approximate entry fee for BREV}
	Suppose for every bidder $i$, $d_i(\cdot): T_{-i}\mapsto R$ is a randomized entry fee function such that for any bid profile $b_{-i}\in T_{-i}$ of the other bidders $$\Pr_{t_i\sim D_i}\left[\sum_{j\in[m]} \left(t_{ij}-\max_{k\neq i} b_{kj}\right)^+\geq d_i(b_{-i})\right]\in\left[\frac{1}{2}-\eta,\frac{1}{2}\right]$$ with probability at least $1-\delta$. 
	Then if we use $d_i(\cdot)$ as the entry fee function in the VCG with entry fee mechanism, the expected revenue is at least 
	$\left(1-\delta-2\eta\right)\cdot\brev$.
\end{lemma}
\begin{proof}
	When $\Pr_{t_i\sim D_i}\left[\sum_{j\in[m]} \left(t_{ij}-\max_{k\neq i} b_{kj}\right)^+\geq d_i(b_{-i})\right]\in\left[\frac{1}{2}-\eta,\frac{1}{2}\right]$, $d_i(b_{-i})$ is no less than the original entry fee $e_i(b_{-i}, D_i)$ for any bid profile $b_{-i}$ of the other bidders. The expected revenue under the new entry fee functions is at least $(\left(1-\delta\right)\cdot\left(\frac{1}{2}-\eta\right)\cdot \sum_{i\in[n]} \E_{b_{-i}\sim D_{-i}}\left[e_i(b_{-i},D_i)\right]
 	\geq \left(1-\delta-2\eta\right)\cdot\brev$.
\end{proof}

\begin{lemma}\label{lem:learn approx entry fee from approx dist}
	For any bidder $i$ and any bid profile $b_{-i}$ from the other bidders, let $\FF_{i,b_{-i}}$ and $\hat{\FF}_{i,b_{-i}}$ be the distributions for the random variable $\sum_{j\in[m]} \left(t_{ij}-\max_{k\neq i} b_{kj}\right)^+$ when $t_i$ is drawn from $D_i$ and $\hat{D}_i$ respectively. If $\left|\left|D_{ij}-\hat{D}_{ij}\right|\right|_K\leq \epsilon$ for all bidder $i$ and item $j$, $\left|\left|\FF_{i,b_{-i}}-\hat{\FF}_{i,b_{-i}}\right|\right|_K\leq 2m\epsilon$ for all $i$ and $b_{-i}$. Moreover, when $m\epsilon\leq 1/16$, we can compute a randomized mechanism whose expected revenue is at least $\frac{\brev}{5}$.
\end{lemma}
\begin{proof}
	For any real number $x$, consider event $\EE_{i,b_{-i},x}=\left\{t_i\ \Big{|}\ \sum_{j\in[m]} \left(t_{ij}-\max_{k\neq i} b_{kj}\right)^+\geq x\right\}$. 
	It is easy to see that $\EE_{i,b_{-i},x}$ is single-intersecting for any any $i$, $b_{-i}$ and $x$.
	 According to Lemma~\ref{lem:Kolmogorov stable for sc}, $$\left|\Pr_{t_i\sim D_i}\left[\EE_{i,b_{-i},x}\right]-\Pr_{t_i\sim \hat{D}_i}\left[\EE_{i,b_{-i},x}\right]\right|\leq 2m\epsilon$$ for any $i$, $b_{-i}$ and $x$. 
	 Hence, $\left|\left|\FF_{i,b_{-i}}-\hat{\FF}_{i,b_{-i}}\right|\right|_K\leq 2m\epsilon$. 
	
	Next, we argue how to construct a randomized entry fee $d_i(b_{-i})$ in polynomial time with only sample access of $\hat{\FF}_{i,b_{-i}}$. 
	Suppose we take $k$ samples from $\hat{\FF}_{i,b_{-i}}$ and sort them in descending order $s_1\geq s_2\geq \cdots\geq s_k$.
	 Let the entry fee $d_i(b_{-i})$ to be $s_{\left\lceil\frac{5k}{16}\right\rceil}$.
	 By the Chernoff bound, with probability at least $1-\exp(-k/128)$ (over the randomness of the samples) $\Pr_{t_i\sim \hat{D_i}}\left[\sum_{j\in[m]} \left(t_{ij}-\max_{k\neq i} b_{kj}\right)^+\geq d_i(b_{-i}) \right]=\Pr_{t_i\sim \hat{D_i}}\left[\EE_{i,b_{-i},d_i(b_{-i})}\right]$ lies in $\left[\frac{1}{4},\frac{3}{8}\right]$.
	  Since $\Pr_{t_i\sim D_i}\left[\EE_{i,b_{-i},d_i(b_{-i})}\right]=\Pr_{t_i\sim \hat{D_i}}\left[\EE_{i,b_{-i},d_i(b_{-i})}\right]\pm 2m\epsilon$, $$\Pr_{t_i\sim D_i}\left[\EE_{i,b_{-i},d_i(b_{-i})}\right]\in [\frac{1}{8},\frac{1}{2}],$$ if $m\epsilon\leq 1/16$.
	   According to Lemma~\ref{lem:approximate entry fee for BREV}, the expected revenue under our entry fee $d_i(b_{-i})$ is at least $\left(\frac{1}{4}-\exp(-k/128)\right)\cdot\brev\geq \frac{\brev}{5}$ if we choose $k$ to be larger than some absolute constant. Clearly, the procedure above can be completed in polynomial time with access to $\hat{D}$.
\end{proof}

Combining Lemma~\ref{lem:SREV kolmogorov} and~\ref{lem:learn approx entry fee from approx dist}, we are ready to prove our main result of this section.
\begin{theorem}\label{thm:additive Kolmogorov}
If all bidders have additive valuations, given distributions $\hat{D}_{ij}$ where $\left|\left|\hat{D}_{ij}-D_{ij}\right|\right|_K\leq \epsilon$ for all $i\in[n]$ and $j\in[m]$, there is a polynomial time algorithm that constructs a mechanism whose expected revenue under $D$ is at least $\frac{\opt}{266}-96\epsilon\cdot mnH$ when $\epsilon\leq \frac{1}{16\max\{m,n\}}$.\end{theorem}
\begin{proof}
	Since $\epsilon\leq \frac{1}{16\max\{m,n\}}$, we can learn in polynomial time a randomized SPM whose revenue is at least $\frac{3}{16}\cdot\left(\frac{\srev}{8}-2\epsilon\cdot mnH\right)$ and a VCG with entry fee mechanism whose revenue is at least $\brev/5$. As $\opt\leq 6\cdot \srev+2\brev$ (Theorem~\ref{thm:UB additive}), the better of the two mechanisms we can learn in polynomial time has revenue at least $\frac{\opt}{266}-96\epsilon\cdot mnH$.
\end{proof}

\section{Missing Details from Section~\ref{sec:constrained additive}}\label{sec:appx XOS}

\begin{prevproof}{Lemma}{lem:approx ASPE}
We only sketch the proof here. Let $\prev$ denote the highest revenue obtainable by any RSPM. In~\cite{CaiZ17}, Cai and Zhao constructed an upper bound of the optimal revenue using duality and separated the upper bound into three components: $\single$, $\tail$ and $\core$. Both $\single$ and $\tail$ are within constant times the $\prev$, and the ASPE$(p^*,\delta^*)$ is used to bound the $\core$. It turns out one can use essentially the same proof as in~\cite{CaiZ17} to prove that the mechanism ASPE$(p',\delta')$ has revenue at least $a_1(\mu)\cdot \core-a_2(\mu)\cdot\prev-a_3(\mu)\cdot (n+m)\cdot \epsilon$ where  $a_1(\mu)$, $a_2(\mu)$ and $a_3(\mu)$ are functions that map $\mu$ to positive numbers. In other words, we can replace ASPE$(p^*,\delta^*)$ with ASPE$(p',\delta')$ and still obtain a constant factor approximation.
\end{prevproof}

\notshow{
\subsection{Missing Proofs from Section~\ref{sec:constrained additive kolomogorov}}\label{sec:appx constrained additive}
\begin{prevproof}{Lemma}{lem:Kolmogorov learn entry fee}
	To prove our claim, it suffices to prove that for any collection of prices $\{p_j\}_{j\in[m]}$, any bidder $i$ and any set of items $S$ that $\left|\Pr_{D_i}\left[u^{(p)}_i(t_i,S)\geq \delta_i^{(p)}(S)\right]-\Pr_{\hat{D}_i}\left[u^{(p)}_i(t_i,S)\geq \delta_i^{(p)}(S)\right]\right|\leq 2m\xi$. Let $\EE$ be the event that $u^{(p)}_i(t_i,S)\geq \delta_i^{(p)}(S)$. Since this is an event over bidder $i$'s type set, the dimension is $m$. If we can argue that $\EE$ is single-intersecting, our claim follows from Lemma~\ref{lem:Kolmogorov stable for sc}. 
	
	For any $j\in[m]$ and $a_{i,-j}\in [0,H]^{{m}-1}$, let $L_j(a_{i,-j})=\left\{ (t_{ij},a_{i,-j}) \ | t_{ij}\in [0,H] \right\}$, then $L_j(a_{i,-j})\cap \EE = \left\{ (t_{ij},a_{i,-j}) \ | t_{ij}\in [0,H] \text{ and }  u^{(p)}_i((t_{ij},a_{i,-j}),S)\geq \delta_i^{(p)} (S)\right\}$. If $(x_{ij}, a_{i,-j})\in L_j(a_{i,-j})\cap \EE$, clearly for any $x'_{ij}\geq x_{ij}$,  $u^{(p)}_i((x'_{ij},a_{i,-j}),S)\geq u^{(p)}_i((x_{ij},a_{i,-j}),S)$. Hence, $(x'_{ij}, a_{i,-j})\in L_j(a_{i,-j})\cap \EE$. Therefore, $\EE$ is single-intersecting.
\end{prevproof}

\notshow{
\begin{prevproof}{Lemma}{lem:difference in revenue Kolmogorov}
	We use a hybrid argument. Consider a sequence of distributions $\{D^{(i)}\}_{i\leq n}$, where $D^{(i)}=\hat{D}_1\times\cdots\times\hat{D}_i\times D_{i+1}\times\cdots\times D_{n},$ and $D^{(0)}=D$, $D^{(n)}=\hat{D}$.
	 We use $\rev^{(i)}(p,\delta)$ to denote the expected revenue of ASPE$(p,\delta)$ under $D^{(i)}$. To prove our claim, it suffices to argue that $\left|\rev^{(i-1)}(p,\delta) -\rev^{(i)}(p,\delta)\right|\leq 2\xi m\cdot \left(m\cdot H+\opt\right).$
	  We denote by $\SS_k$ and $\SS'_k$ the random set of items that remain available after visiting the first $k$ bidders under $D^{(i-1)}$ and $D^{(i)}$. Clearly, for $k\leq i-1$, $||\SS_k-\SS'_k||_{TV}=0$, so the expected revenue collected from the first $i-1$ bidders under $D^{(i-1)}$ and $D^{(i)}$ is the same. According to Lemma~\ref{lem:stable favorite set}, $||\SS_i-\SS'_i||_{TV}\leq 2m\cdot \xi$. The total amount of money bidder $i$ spends can never be higher than her value for receiving all the items which is at most $m\cdot H$. So the difference in the expected revenue collected from bidder $i$ under  $D^{(i-1)}$ and $D^{(i)}$ is at most $2\xi\cdot m^2H$. Suppose $R$ is the set of remaining items after visiting the first $i$ bidders, then the expected revenue collected from the last $n-i$ bidders is the same under  $D^{(i-1)}$ and $D^{(i)}$, as these bidders have the same distributions. Moreover, this expected revenue is no more than $\opt_{ASPE}$, since the optimal ASPE can simply just sell $R$ to the last $n-i$ bidders using the same prices and entry fee as in ASPE$(p,\delta)$. Of course, for any fixed $R$, the probabilities that $\SS_i=R$ and $\SS'_i=R$ are different, but since for any $R$ the expected revenue from the last $n-i$ bidders is at most $\opt_{ASPE}$, the difference in the expected revenue from the last $n-i$ bidders under  $D^{(i-1)}$ and $D^{(i)}$ is at most $||\SS_i-\SS'_i||_{TV} \cdot \opt \leq 2\xi\cdot m\opt_{ASPE}$. Hence, the total difference between $\rev^{(i-1)}(p,\delta)$ and  $\rev^{(i)}(p,\delta)$ is at most $2\xi m\cdot \left(m H+\opt_{ASPE}\right)$. 
	  Furthermore, $\left|\rev(p,\delta)-\widehat{\rev}(p,\delta)\right|\leq \sum_{i=1}^{n}\left|\rev^{(i-1)}(p,\delta) -\rev^{(i)}(p,\delta)\right|\leq 2 nm\xi \left(mH+\opt_{ASPE}\right).$
\end{prevproof}}

\begin{prevproof}{Theorem}{thm:constrained additive Kolmogorov}
 Let $\prev$ be the revenue of the optimal RSPM. According to Theorem~\ref{thm:UD Kolmogorov}, given $\hat{D}$, we can learn an RSPM that has revenue $\left(\frac{1}{4}-(n+m)\cdot \xi \right)\cdot \left(\frac{\prev}{8}-2\xi\cdot mnH\right)$. On the other hand, with access to $\hat{D}$, we can learn an ASPE such that either this ASPE or the best RSPM achieves a constant fraction of $\opt$ minus $\xi\cdot O(m^2nH)$. This can be proved by setting $\epsilon$ to be $O(\frac{\opt}{m+n})$ in Corollary~\ref{cor:discretization of prices}, then combine it with Lemma~\ref{lem:Kolmogorov learn entry fee} and~\ref{lem:difference in revenue Kolmogorov}. To sum up, we can learn an SPM and an ASPE with only access to $\hat{D}$ such that the better among the two achieves revenue at least $\frac{\opt}{c}-\xi\cdot O(m^2nH)$ under the actual distribution $D$ for some constant $c>1$.
\end{prevproof}
}

\subsection{Missing Proofs from Section~\ref{sec:XOS sample}}\label{sec:XOS bounded}

 We formalize the first step of our algorithm in the following lemma.
\begin{lemma}\label{lem:learn median}
For any $B>0$, $\epsilon>0$, $\eta\in [0,1]$ and $\mu\in[0,\frac{1}{4}]$, suppose we take 
$K=O\left(\frac{\log \frac{1}{\eta}+ \log n +m\log \frac{B}{\epsilon}}{\mu^2}\right)$ samples $t^{(1)},\cdots, t^{(K)}$ from $D$. For any collection of prices $\{p_j\}_{j\in[m]}$ in the $B$-bounded $\epsilon$-net, define the entry fee $\delta_i^{(p)}(S)$
 of bidder $i$ for set $S$ under $\{p_j\}_{j\in[m]}$ to be the median of $u_i(t^{(1)}_i,S),\cdots, u_i(t^{(K)}_i,S)$, where $u_i(t_i,S)=\max_{S*\subseteq S} v_i(t_i,S^*)-\sum_{j\in S^*} p_j$. Then with probability $1-\eta$, for any collection of prices $\{p_j\}_{j\in[m]}$ in the $B$-bounded $\epsilon$-net, $\left\{\delta_i^{(p)}(\cdot)\right\}_{i\in[n]}$ is a collection of $\mu$-balanced entry fee functions.\end{lemma}
 \begin{proof}
 	For any fixed $\{p_j\}_{j\in[m]}$, fixed bidder $i$ and fixed set $S$, it is easy to argue that the probability for $\Pr_{t_i\sim D_i}[u_i(t_i,S)\geq \delta_i^{(p)}(S)]$ to be larger than  $\frac{1}{2}+\mu$ or smaller than $\frac{1}{2}-\mu$ is at most 2$\exp(-2K\mu^2)$ due to the Chernoff bound. 
 	Next, we take a union bound over all $\{p_j\}_{j\in[m]}$ in the $\epsilon$-net, all bidders and all possible subsets of $[m]$, so the probability that for any collection of prices $\{p_j\}_{j\in[m]}$ in the $\epsilon$-net $\{\delta_i^{(p)}(\cdot)\}_{i\in[n]}$ is a collection of $\mu$-balanced entry fee functions is at least $1- 2\exp(-2K\mu^2)\cdot \left(\frac{B}{\epsilon}\right)^m\cdot 2^m\cdot n$. If we take $K$ to be at least $\frac{\log \frac{1}{\eta}+ \log n +m\log \frac{B}{\epsilon}}{\mu^2}$, the success probability is at least $1-\eta$.
 \end{proof}

Next, we formalize the second step of our learning algorithm.
\begin{lemma}\label{lem:learn best ASPE}
	For any $B\geq 2G$, $\epsilon, \epsilon'>0$,  $\eta\in [0,1]$ and $\mu\in[0,\frac{1}{4}]$, suppose for every collection of prices $\{p_j\}_{j\in[m]}$ in the $B$-bounded $\epsilon$-net, $\{\delta^{(p)}_i(\cdot)\}_{i\in[n]}$ is a collection of $\mu$-balanced entry fee functions. We use $\mathcal{S}$ to denote the set that contains ASPE$(p,\delta^{(p)})$ for every $p$ in the $B$-bounded $\epsilon$-net. If we take 
	$K=O\left(\frac{\log \frac{1}{\eta}+m\log \frac{B}{\epsilon}}{\epsilon'^2}\right)$ samples $t^{(1)},\cdots, t^{(K)}$ from $D$ and let ASPE$(p',\delta^{(p')})$ be the mechanism that has the highest revenue in $\mathcal{S}$. Then with probability at least $1-\eta$, the better of ASPE$(p',\delta^{(p')})$ and the best RSPM achieves revenue at least $\frac{\opt}{\CC_1(\mu)}-\CC_2(\mu)\cdot (m+n)\cdot \epsilon-2mnB\cdot \epsilon'$. 
\end{lemma}
\begin{proof}
	For any $\{p_j\}_{j\in[m]}$ in the $\epsilon$-net, define $\rev(p)$ to be the expected revenue of ASPE$(p,\delta^{(p)})$ and  $\widehat{\rev}(p)$ be the average revenue of ASPE$(p,\delta^{(p)})$ among the $K$ samples. First, we argue that $\widehat{\rev}(p)$ is a random variable that lies between $[0,mnB]$. The revenue from selling the items can be at most $mB$ as there are only $m$ items and $p_j\leq B$ for all $j\in[m]$. How about the entry fee? 
	For any bidder $i$, $$\Pr_{t_i\sim D_i}\left[v_i(t_i,[m])\geq mG\right]\leq \sum_{j\in[m]} \Pr_{t_{ij}\sim D_{ij}}\left[V_i(t_{ij})\geq G\right]\leq \frac{m}{5\max\{m,n\}}\leq\frac{1}{5}.$$ The first inequality is because $v_i(t_i,\cdot)$ is a subadditive function for every type $t_i\in T_i$, so for $v_i(t_i,[m])$ to be greater than $mG$, there must exist a item $j$ such that $V_i(t_{ij})\geq G$. The second inequality follows from the definition of $G$ in Theorem~\ref{thm:simple XOS}.
	
	 If there exists a set $S\subseteq[m]$ such that $\delta_i^{(p)}(S)> mG$, we have $$\Pr_{t_i\sim D_i}\left[v_i(t_i,[m])\geq mG\right]\geq \Pr_{t_i\sim D_i}\left[v_i(t_i,[m])\geq \delta_i^{(p)}(S)\right]\geq 
	  \frac{1}{2}-\mu\geq \frac{1}{4}.$$ Contradiction. Note that the second inequality is because $\delta_i^{(p)}(\cdot)$ is $\mu$-balanced. Hence, the entry fee is always upper bounded by $mG$ and $\widehat{\rev}(p)$ is at most $mnG+mB\leq mnB$. Also, notice that the expectation of $\widehat{\rev}(p)$ is exactly $\rev(p)$. 
	  By the Chernoff bound, $$\Pr\left[\left|\rev(p)-\widehat{\rev}(p)\right|\leq mnB\cdot \epsilon'\right]\geq 1-2\exp(-2K\cdot \epsilon'^2)$$ for any fixed $\{p_j\}_{j\in[m]}$. By the union bound, the probability that for all $\{p_j\}_{j\in[m]}$ in the $\epsilon$-net $$\left|\rev(p)-\widehat{\rev}(p)\right| \leq mnB\cdot \epsilon'$$ is at least $1-2\exp(-2K\cdot \epsilon'^2)\cdot \left(\frac{B}{\epsilon}\right)^m$, which is lower bounded by $1-\eta$ due to our choice of $K$. 
When this happens, the expected revenue of ASPE$(p',\delta^{(p')})$ is at most $2mnB\cdot \epsilon'$ less than the highest expected revenue achievable by any of these mechanisms, because 
$$\rev(p')\geq \widehat{\rev}(p')-mnB\cdot \epsilon' \geq \widehat{\rev}(p)-mnB\cdot \epsilon'\geq \rev(p)-2mnB\cdot \epsilon'$$ for any $p$ in the $\epsilon$-net. Combining this inequality with Corollary~\ref{cor:discretization of prices} completes our proof.
\end{proof}

Note that Lemma~\ref{lem:learn median} and \ref{lem:learn best ASPE} hold for all distributions $D$. The reason we require $D$ to be bounded or regular is because without these restrictions, we do not know how to approximate the best RSPM. In the following Theorem, we combine Lemma~\ref{lem:learn median}, ~\ref{lem:learn best ASPE} and Theorem~\ref{thm:UD bounded} to obtain the sample complexity of our learning algorithm for bounded distributions. 
\begin{theorem}\label{thm:XOS bounded}
	When all bidders' valuations are XOS over independent items and the random variable $V_i(t_{ij})$ is supported on $[0,H]$ for any bidder $i$ and any item $j$, with $O\left(\left(\frac{mn}{\xi}\right)^2 \cdot \left(m\cdot\log \frac{m+n}{\xi} + \log \frac{1}{\delta}\right)\right)$ samples from $D$, we can learn an RSPM and an ASPE such that with probability at least $1-\delta$ the better of the two mechanisms has revenue at least $\frac{\opt}{c}-\xi\cdot H$ for some absolute constant $c>1$. 
\end{theorem}
\begin{proof}
With $O\left(\left(\frac{1}{\xi}\right)^2 \left(m^2 n\log n\log \frac{1}{\xi} + \log \frac{1}{\delta}\right)\right)$ samples, we can obtain an RSPM whose revenue is at least $\frac{1}{24}$ of the revenue of the best RSPM minus $\frac{\xi}{2}\cdot H$ with probability $1-\delta/2$ according to Theorem~\ref{thm:UD bounded}. 
Let $\mu$ be some fixed constant in $[0,\frac{1}{4}]$, $B=2H$, $\epsilon = \frac{\xi\cdot H}{6\CC_2(\mu) (m+n)}$ and $\epsilon' = \frac{\xi}{12mn}$. 
According to Lemma~\ref{lem:learn median}, given $O\left(\log \frac{1}{\delta}+ \log n +m\log \frac{m+n}{\xi}\right)$ samples, we can construct an entry fee function for each price vector in the $B$-bounded $\epsilon$-net, such that all these entry fee functions are $\mu$-balanced with probability at least $1-\delta/4$.
 According to Lemma~\ref{lem:learn best ASPE}, we can learn an ASPE with $O\left(\left(\frac{mn}{\xi}\right)^2 \cdot \left(m\cdot\log \frac{m+n}{\xi} + \log \frac{1}{\delta}\right)\right)$ fresh samples from $D$, such that the better of the ASPE we learned and the best RSPM has revenue of at least $\frac{\opt}{\CC_1(\mu)}-\frac{\xi}{2}\cdot H$ with probability $1-\delta/4$. Combining the statements above, we can learn with probability $1-\delta$ a mechanism whose revenue is at least $\frac{\opt}{c}-\xi\cdot H$ with $O\left(\left(\frac{mn}{\xi}\right)^2 \cdot \left(m\cdot\log \frac{m+n}{\xi} + \log \frac{1}{\delta}\right)\right)$ samples.
\end{proof}

In the next Theorem, we combine Lemma~\ref{lem:learn median}, ~\ref{lem:learn best ASPE} and Theorem~\ref{thm:UD regular} to obtain the sample complexity of our learning algorithm for regular distributions.

\begin{theorem}\label{thm:XOS regular}
	When all bidders' valuations are XOS over independent items and the random variable $V_i(t_{ij})$ is regular for each item $j\in[m]$ and bidder $i\in[n]$, with $O\left(Z^2m^2n^2 \cdot \left(m\cdot\log ({m+n}) + \log \frac{1}{\delta}\right)\right)$ ($Z=\max\{m,n\}$) samples from $D$, we can learn an RSPM and an ASPE such that with probability at least $1-\delta$ the better of the two mechanisms has revenue at least $\frac{\opt}{c}$ for some absolute constant $c>1$. 
\end{theorem}

\begin{proof}
According to Theorem~\ref{thm:UD regular}, we can learn with probability $1-\delta/2$ a randomized RSPM whose revenue is at least $\frac{1}{33}$ of the optimal RSPM with $O\left(Z^2m^2 n^2\cdot  \log \frac{nm}{\delta}\right)$ samples. Next, we learn an ASPE with high revenue.
	With $O\left( Z^2\cdot \log \frac{nm}{\delta} \right)$ samples from each $D_{ij}$, we can estimate $W_{ij}$ such that $$\Pr_{t_{ij}\sim D_{ij}}\left[V_i(t_{ij}) \geq W_{ij}\right]\in\left[\frac{1}{6Z}, \frac{1}{5Z}\right]$$ with probability $1-\frac{\delta}{4nm}$. By the union bound, the probability that all $W_{ij}$ satisfy the requirement is at least $1-\frac{\delta}{4}$. So with probability at least $1-\frac{\delta}{4}$, $W_{ij}\geq G_{ij}$ for all $i\in[n]$ and $j\in[m]$. 
	
	Let $B=2\cdot \max_{i,j} W_{ij}$, $\mu$ be some fixed constant in $[0,\frac{1}{4}]$, $\epsilon = \frac{\xi\cdot B}{\CC_2(\mu) Z(m+n)}$ and $\epsilon' = \frac{\xi}{2mnZ}$ for some small constant $\xi$, which will be specified later. We know that given $O\left(\log \frac{1}{\delta}+ \log n +m\log ({m+n})\right)$ samples, we can construct $\mu$-balanced entry fee functions for all price vectors in the $B$-bounded $\epsilon$-net with probability $1-\delta/8$ due to Lemma~\ref{lem:learn median}. According to Lemma~\ref{lem:learn best ASPE}, we can learn an ASPE with $$O\left(Z^2m^2n^2 \cdot \left(m\cdot\log (m+n) + \log \frac{1}{\delta}\right)\right)$$ fresh samples from $D$, such that the better of the ASPE we learned and the best RSPM has revenue of at least $\frac{\opt}{\CC_1(\mu)}-\frac{2\xi\cdot B}{Z}$ with probability $1-\delta/8$. Note that there exists a bidder $i$ and an item $j$ such that $W_{ij}=B/2$, so $\opt\geq \frac{B}{2}\cdot \frac{1}{6Z}$ and for sufficiently small $\xi$, $\frac{\opt}{\CC_1(\mu)}-\frac{2\xi\cdot B}{Z}\geq \frac{\opt}{2\CC_1(\mu)}$. Combining the statements above, we can learn with probability $1-\delta$ a mechanism whose revenue is at least $\frac{\opt}{c}$ for some absolute constant $c$ with $O\left(Z^2m^2n^2 \cdot \left(m\cdot\log (m+n) + \log \frac{1}{\delta}\right)\right)$ samples.

\end{proof}

\section{Learning Algorithms for Symmetric Bidders}\label{sec:symmetric appx}
\subsection{An Upper Bound of the Optimal Revenue for Symmetric Bidders}\label{sec:UB symmetric}
In this section, we introduce an upper bound to $\opt$ based on duality~\cite{CaiZ17}, which is crucial for us to prove the approximation ratios of our learning algorithms. 
We first fix some notation. Let $\prev$ be the highest revenue obtainable by any RSPM. As the bidders are symmetric, we drop the subscript $i$ when there is no confusion. In particular, we use $V(t_{ij})$ to denote bidder $i$'s value for winning item $j$ if her private information for item $j$ is $t_{ij}$, and $v(t_i,S)$ to denote bidder $i$'s value for set $S$ when her type is $t_i$. We use $D_j$ to denote the distribution of the private information about item $j$. Let $\sigma_{iS}(t)$ be the interim probability for bidder $i$ to receive exactly set $S\subseteq [m]$ when her type is $t$. 

In~\cite{CaiZ17}, an upper bound of the optimal revenue is derived using duality theory. Their upper bound applies to asymmetric bidders with valuations that are subadditive over independent items. When the bidders are symmetric, we can simplify their upper bound. First, we need the definition of $b$-balanced thresholds.

\begin{definition}[$b$-balanced Thresholds]\label{def:b-balanced}
	For any constant $b\in (0,1)$, a collection of positive real numbers $\{\beta_j\}_{j\in[m]}$ is $b$-balanced if for all $i\in[n]$ and $j\in[m]$, $\Pr_{t_{ij}\sim D_j}\left [V(t_{ij})\geq \beta_j\right]\in[\frac{b}{n},\frac{b}{n-1}]$.
\end{definition}

Note that when bidders are asymmetric, $b$-balanced thresholds are not guaranteed to exist, as there may not exist any $\beta_j$ that satisfies $\Pr_{t_{ij}\sim D_j}\left [V(t_{ij})\geq \beta_j\right]\in[\frac{b}{n},\frac{b}{n-1}]$ for all bidder $i$ simultaneously. Next, we define the $\core_\eta(\bbeta)$ which will be crucial for upper bounding the optimal revenue.\footnote{For readers that are familiar with the definition of the $\core$ in~\cite{CaiZ17}, $\core_\eta(\bbeta)$ is essentially the same term but adapted for symmetric bidders.}

\begin{definition}[\core]\label{def:core and tail}
Given any collection of 
 thresholds $\{\beta_j\}_{j\in[n]}$ and a nonnegative constant $\eta\leq \frac{1}{4}$, \begin{itemize}
	\item if $\sum_{j\in[m]} \Pr_{t_{j}\sim D_{j}}\left[V(t_{j})\geq \beta_{j}\right]\leq \frac{1}{2}-\eta$, let $c_\eta(\boldsymbol{\beta})$ be $0$;
	\item otherwise, let $c_\eta(\boldsymbol{\beta})$ be a nonnegative number such that $\sum_{j\in[m]} \Pr_{t_{j}\sim D_{j}}\left[V(t_{j})\geq \beta_{j}+c_\eta(\boldsymbol{\beta})\right]\in  \left[\frac{1}{2}-\eta, \frac{1}{2}\right]$.
\end{itemize}
For every type $t$, let $\mathcal{C}_{\eta}(t)=\{j\ |\ V(t_{ j})< \beta_{j}+c_\eta(\bbeta)\}$. Then, 
$$\core_\eta(\bbeta)= \max_{\sigma \in P(D)} \sum_{i\in[n]}\sum_{t_i\in T_i}f(t_i)\cdot \sum_{S\subseteq[m]}\sigma_{iS}(t_i)\cdot v\left(t_i,S\cap \mathcal{C}_{\eta}(t_i)\right),$$ where $P(D)$ is the set of all feasible interim allocation rules. That is, {$\core_\eta(\bbeta)$ is the maximum welfare a mechanism can extract out of the allocation of items whose individual value for the bidder they are allocated to is lower than the adjusted thresholds.}
\end{definition}


It was shown in~\cite{CaiZ17} that every collection of thresholds induces an upper bound to the optimal revenue. In particular, for any choice of thresholds $\{\beta_j\}_{j\in[m]}$ and $\eta$\footnote{In~\cite{CaiZ17}, the thresholds are allowed to depend on the identity of the bidder. More specifically, for any $i\in[n]$ and $j\in[m]$, there is an associated threshold $\beta_{ij}$. Their upper bound applies to asymmetric thresholds as well. Indeed, when the bidders are asymmetric, their upper bound is induced by a set of asymmetric thresholds. As we only discuss symmetric bidders in this section, we focus on symmetric thresholds for simplicity. Regarding $\eta$, Cai and Zhao only considered the case when $\eta=0$, but their analysis can be easily modified to accommodate any $\eta\leq 1/4$. See Theorem~\ref{thm:UB subadditive} for the modified upper bound.}, the revenue $\rev(M)$ of any BIC mechanism $M$ is upper bounded by  $$2\cdot \single(M,\bbeta)+4\cdot \tail_\eta(M,\bbeta)+4\cdot\core_\eta(M,\bbeta)\quad\text{(Adapted from Theorem 2 in~\cite{CaiZ17})}.$$ These terms depend on the choice of $\{\beta_j\}_{j\in[m]}$, $\eta$ as well as the mechanism $M$. We refer interested readers to~\cite{CaiZ17} for the definitions of these terms. To obtain a benchmark/upper bound of the optimal revenue, one can simply replace the above expression with $$2\cdot \max_M \single(M,\bbeta)+4\cdot \max_M \tail_\eta(M,\bbeta)+4\cdot\max_M \core_\eta(M,\bbeta).$$ {It is not hard to see that this benchmark may be impossible to approximate for certain choices of the thresholds. Just imagine the case when the thresholds are extremely high, then $\max_M \core_\eta(M,\bbeta)$ becomes the optimal social welfare which can be arbitrarily large comparing to the optimal revenue. What Cai and Zhao~\cite{CaiZ17} showed was that when the thresholds are $b$-balanced, this upper bound can indeed be approximated by the revenue of an RSPM and an ASPE. From now on, we only consider $b$-balanced thresholds.}


Using results in~\cite{CaiZ17}, we can further simplify the benchmark. In particular, $\max_{M}\single(M,\bbeta)$ is less than $6\cdot\prev$ for all choices of $\{\beta_j\}_{j\in[n]}$ and $\max_{M}\tail_\eta(M,\bbeta)$ is less than $\frac{2}{1-b}\cdot\prev$ for any choice of $\eta$ and $b$-balanced thresholds $\{\beta_j\}_{j\in[n]}$. Moreover, $\max_{M}\core_\eta(M,\bbeta)\leq \core_\eta(\bbeta)$. Combining the inequalities above, we obtain the following Theorem.
\begin{theorem}[Adapted from~\cite{CaiZ17}]\label{thm:UB subadditive}
	When the bidders are symmetric and have valuations that are subadditive over independent items, for any constant $b\in (0,1)$, $\eta\leq \frac{1}{4}$ and a collection of $b$-balanced thresholds $\{\beta_j\}_{j\in[m]}$, $$\opt \leq \left(12+\frac{8}{1-b}\right)\cdot \prev + 4\cdot\core_\eta(\bbeta).$$
\end{theorem}

\subsection{Symmetric Bidders with XOS Valuations}\label{sec:symmetric XOS}
In this section, we show how to learn in polynomial time an approximately optimal mechanism for symmetric bidders with XOS valuations given sample access to the distributions. According to Theorem~\ref{thm:UB subadditive}, we only need to learn a mechanism that approximates $\prev$ and $\core_\eta(\bbeta)$. From Section~\ref{sec:unit-demand}, we know how to approximated $\prev$ in polynomial time, so we focus on learning a mechanism whose revenue  approximates $\core_\eta(\bbeta)$. 

First, we need a crucial property about XOS valuations.

\begin{lemma}[Supporting Prices~\cite{DobzinskiNS05}]\label{lem:supporting price}
If $v(t,\cdot)$ is an XOS function, for any subset $S\subseteq[m]$ there exists a collection of supporting prices $\left\{\theta^{S}_j (t)\right\}_{j\in S}$ for $v(t,S)$ such that\begin{enumerate}
	\item $v(t,S') \geq \sum_{j\in S'} \theta^{S}_j (t)$ for all $S'\subseteq S$ and
	\item $\sum_{j\in S}\theta^{S}_j (t) = {v(t,S)}$.
\end{enumerate}
\end{lemma}

Let $v'(t_i,S)=v\left(t_i,S\cap \CC_{\eta}(t_i)\right)$ and $\FF_i$ be the distribution of the valuation $v'(t_i,S)$. As the bidders are symmetric, $\FF_i=\FF_{i'}$ for any $i$ and $i'$. The $\core_\eta(\bbeta)$ is exactly the maximum expected social welfare if every bidder $i$'s valuation is drawn independently from $\FF_i$. Cai and Zhao~\cite{CaiZ17} showed how to use an ASPE to approximate this term. In the next Lemma, we construct the prices used in their ASPE and show its relation to $\core_\eta(\bbeta)$.

\begin{lemma}\label{lem:Q_j}(Adapted from~\cite{CaiZ17})
	Let every bidder $i$'s valuation be  $v'(t_i,S)=v\left(t_i,S\cap \CC_{\eta}(t_i)\right)$ when her type is $t_{i}$ and $\sigma^*$ be a symmetric allocation that achieves $\alpha$-fraction of the optimal social welfare with respect to $v'(\cdot,\cdot)$. 
	For every item $j\in[m]$, let $$Q_{\eta,j}=\frac{1}{2}\cdot \sum_{i\in[n]} \sum_{t_i\in T_i}f(t_i)\cdot\sum_{S:j\in S}\sigma_{iS}^{*}(t_i)\cdot \theta_j^{S\cap  \CC_{\eta}(t_i)}(t_i),$$ where $\left\{\theta_j^{S\cap  \CC_{\eta}(t_i)}(t_i)\right\}_{j\in S\cap  \CC_{\eta}(t_i)}$ is the supporting prices for $v\left(t_i,S\cap  \CC_{\eta}(t_i)\right)$.
	 Let $$u^*(t, S)=\max_{S^*\subseteq S} v(t, S^*)-\sum_{j\in S^*} Q_{\eta,j}$$ be a bidder's utility for the set of items $S$ when her type is $t$. 
	 We define $\delta^*(S)$ to be the median of the random variable $u^{*}(t,S)$ (with $t\sim \times_{j\in[m]}D_j$) for any set $S\subseteq [m]$.
	  The revenue of ASPE$\left(\{Q_{\eta,j}\}_{j\in[m]}, \delta^*\right)$ is at least $$\frac{\alpha\cdot \core_{\eta}(\bbeta)}{2}-\CC(b,\eta)\cdot \prev,$$ where $\CC(b,\eta)$ is a function that only depends on $b$ and $\eta$. 
	\end{lemma}
\begin{proof}
We can essentially use the same proof in~\cite{CaiZ17} to prove that the expected revenue of the ASPE is at least $$\sum_{j\in[m]} Q_{\eta,j}-\CC(b,\eta)\cdot \prev.$$ For readers that are familiar with that proof, the only thing we need to make sure is that our choice of $\sigma^*$ and $\{\beta_j\}_{j\in[m]}$ satisfy Lemma 5 in~\cite{CaiZ17}. Since $\sigma^*$ is symmetric and $\{\beta_j\}_{j\in[m]}$ is $b$-balanced, for all bidder $i$ and item $j$ $$\sum_{k\neq i}\Pr_{t_{kj}\sim D_{j}}\left[V(t_{kj})\geq \beta_{j}\right]\leq \frac{b}{n-1}\cdot (n-1)=b,$$ and $$\Pr_{t_{ij}\sim D_j}\left[V(t_{ij}) \geq \beta_j\right]/b\geq 1/n\geq  \cdot \sum_{t_i\in T_i} f_i(t_i)\cdot \sum_{S: j\in S} \sigma^*_{iS}(t_i).$$

Next, we argue $\sum_{j\in[m]} Q_{\eta,j}\geq \frac{\alpha\cdot \core_{\eta}(\bbeta)}{2}$. Observe that $$\sum_{j\in[m]} Q_{\eta,j}=\frac{1}{2}\cdot \sum_{i\in[n]} \sum_{t_i\in T_i}f(t_i)\cdot\sum_{S}\sigma_{iS}^{*}(t_i)\cdot v'(t_i,S)\geq \alpha\cdot \core_{\eta}(\bbeta).$$ The last inequality is because $\core_{\eta}(\bbeta)$ is the maximum social welfare under $v'(\cdot,\cdot)$ and $\sigma^*$ achieves $\alpha$ fraction of that.
\end{proof}	
	
\begin{lemma}\label{lem:symmetric approx ASPE}
	For any $\epsilon>0$ and $\mu\in[0,\frac{1}{4}]$, let $\{Q_j\}_{j\in[m]}$ be a collection of prices such that $\left|Q_j-Q_{\eta,j}\right|\leq \epsilon$ for all $j\in[m]$. Let $\delta(S)$ be the entry fee function such that $\Pr_{t\sim \times_{j\in[m]} D_j}\left [u(t,S)\geq \delta(S)\right]\in [1/2-\mu,1/2+\mu]$ for any set $S\subseteq [m]$, where $u(t,S) = \max_{S*\subseteq S} v(t,S^*)-\sum_{j\in S^*} Q_j$. Then, the ASPE$(Q,\delta)$ achieves at least $\frac{\alpha\cdot\core_{\eta}(\bbeta)}{\BB_{1}(\mu)}-\BB_{2}(b,\eta,\mu)\cdot \prev-\BB_3(\mu)\cdot (m+n)\cdot \epsilon$ revenue when bidders' valuations are XOS over independent item. Both $\BB_1(\mu)$ and $\BB_{3}(\mu)$ are functions that only depend on $\mu$ and $\BB_{2}(b,\eta,\mu)$ is a function that only depends on $\mu$, $b$ and $\eta$. 
\end{lemma}
\begin{proof}
It turns out the proof in~\cite{CaiZ17} is robust enough to accommodate the error $\epsilon$ and $\mu$. We can prove the claim by following essentially the same analysis as in~\cite{CaiZ17}. We do not include the details here.\end{proof}
	
\subsubsection{Leaning the ASPE in Polynomial Time}

We first show how to learn a collection of $b$-balanced thresholds and the corresponding $c_\eta(\bbeta)$.
\begin{lemma}\label{lem:symmetric XOS learn beta and c}
	For any positive constant $b<1$ and $\eta\leq \frac{1}{4}$, there is a polynomial time algorithm that computes a collection of $b$-balanced thresholds $\{\beta_j\}_{j\in[m]}$ and $c_\eta(\bbeta)$ with probability $1-\delta$ using $O\left(m^2 n^4 \log \frac{m}{\delta}\right)$ samples from distribution $\times_{j\in[m]} D_j$.
\end{lemma}
 	\begin{proof}
 		Given $K=O\left(m^2n^4\left(\log m+\log \frac{1}{\delta}\right)\right)$ samples $t_j^{(1)},\ldots, t_j^{(K)}$ from distribution $D_j$, we construct $\FF_j$ as the uniform distribution over $V\left(t_j^{(1)}\right),\ldots,V\left(t_j^{(K)}\right)$. According to the DKW Theorem~\cite{DvoretzkyKW56}, with probability at least $1-\delta/m$, 
 		\begin{equation}\label{eq:symmetric XOS empirical}
 			\quad \left|\Pr_{t_{j}\sim D_j}\left[V(t_j)\geq x\right]-\Pr_{v_j\sim \FF_j}\left[v_j\geq x\right]\right|\leq \frac{1}{c\cdot mn^2}~~\text{for all $x$} 	
 		\end{equation}
 where $c$ is a constant that will be specified later. From now on, we assume that Inequality~(\ref{eq:symmetric XOS empirical}) holds for every $j$, which happens with probability $1-\delta$. 
 
As 
$1/K\leq \frac{b}{3n^2}\leq \frac{b}{n-1}-\frac{b}{3n^2}- \frac{b}{n}-\frac{b}{3n^2}$, there must exist a sample $t_j^{(\ell)}$ such that $\Pr_{v_j\sim \FF_j}\left[v_j\geq V\left(t_j^{(\ell)}\right)\right]\in \left[\frac{b}{n}+\frac{b}{3n^2}, \frac{b}{n-1}-\frac{b}{3n^2}\right]$. Let $\beta_j=V\left(t_j^{(\ell)}\right)$. 
Note that $$\Pr_{t_{j}\sim D_j}\left[V(t_j)\geq \beta_j\right]\in \left[\Pr_{v_j\sim \FF_j}\left[v_j\geq \beta_j \right]-\frac{1}{c\cdot mn^2}, \Pr_{v_j\sim \FF_j}\left[v_j\geq \beta_j\right]+\frac{1}{c\cdot mn^2}\right].$$ 
If $c$ is less than $\frac{b}{3}$, $\Pr_{t_{j}\sim D_j}\left[V(t_j)\geq \beta_j\right]\in\left[\frac{b}{n}, \frac{b}{n-1}\right]$. Thus, $\beta_j$ is $b$-balanced  for all item $j$.
 		
 		Next, we argue how to learn $c_\eta(\bbeta)$. If $\sum_{j\in[m]} \Pr_{v_j\sim\FF_j}\left[v_j\geq \beta_j \right] \leq \frac{1}{2}-\frac{\eta}{2}$, let $c_{\eta}(\bbeta)=0$. 
		This is a valid choice, as $\sum_{j} \Pr_{t_{j}\sim D_j}\left[V(t_j)\geq \beta_j\right]$ is at most $\frac{1}{2}-\frac{\eta}{2}+\frac{1}{cn^2}\leq \frac{1}{2}$ 
		 according to inequality (\ref{eq:symmetric XOS empirical}). 
		 Suppose $\sum_{j\in[m]} \Pr_{v_j\sim\FF_j}\left[v_j\geq \beta_j \right] > \frac{1}{2}-\frac{\eta}{2}$, as $m/K< \frac{\eta}{4}$, there must exist some item $k\in[m]$ and a sample $V\left(t_k^{(\ell)}\right)\geq \beta_k$ such that $\sum_{j\in[m]} \Pr_{v_j\sim\FF_j}\left[v_j\geq \beta_j+V\left(t_k^{(\ell)}\right)- \beta_k \right]\in \left[\frac{1}{2}-\frac{\eta}{4},\frac{1}{2}-\frac{\eta}{2}\right]$.
		  Let $c_\eta(\bbeta)= V\left(t_k^{(\ell)}\right)- \beta_k$.
		   According to inequality~(\ref{eq:symmetric XOS empirical}),  $$\sum_{j\in[m]} \Pr_{t_{j}\sim D_j}\left[V(t_j)\geq \beta_j+c_\eta(\bbeta) \right]\in  \left[\frac{1}{2}-\frac{\eta}{4}-\frac{1}{cn^2},\frac{1}{2}-\frac{\eta}{2}+\frac{1}{cn^2}\right].$$
		    For sufficiently large $c$, $\sum_{j} \Pr_{t_{j}\sim D_j}\left[V(t_j)\geq \beta_j+c_\eta(\bbeta) \right]\in  \left[\frac{1}{2}-\eta,\frac{1}{2}\right]$. 
 		
 		Finding each $\beta_j$ takes $O(K\log K)$ time and finding the $c_\eta(\bbeta)$ takes $O(mK)$ time. So we can learn in polynomial time a collection of $b$-balanced thresholds $\{\beta_j\}_{j\in[m]}$ and $c_\eta(\bbeta)$ with probability $1-\delta$ using $O\left(m^2n^4\log \frac{m}{\delta}\right)$ samples.
 			\end{proof}
 			
 		Next, we show how to learn the prices of the ASPE. As showed by Feige~\cite{Feige09}, there exists a polynomial time algorithm that achieves $1-\frac{1}{e}$ fraction of the optimal social welfare when bidders have XOS valuations. We let $\sigma^*$ be the interim allocation rule induced by Feige's algorithm and estimate the prices by running Feige's algorithm on sampled valuation profiles.	To run Feige's algorithm, we need a demand oracle for bidder's valuations. In the following Lemma, we argue that $v'(t,\cdot)$ is an XOS function for any type $t$, and given a value (or demand, XOS) oracle for $v(t,\cdot)$, we can construct in polynomial time the corresponding oracle for $v'(t,\cdot)$. First, we define these oracles formally.
 
 \begin{definition}
 We consider the following three oracles for a bidder's valuation function $v(t,\cdot)$:
 \begin{itemize}
 	\item \textbf{Value oracle}: takes a set $S\subseteq [m]$ as the input and returns $v(t,S)$.
 	\item \textbf{Demand oracle}: takes a collection of prices $\{p_j\}_{j\in[m]}$ as an input and returns the favorite set under these prices, that is, $S^*\in \argmax_{S\in[m]} v(t,S)-\sum_{j\in S} p_j$.
 	\item \textbf{XOS oracle} (only when $v(t,\cdot)$ is XOS):  takes a set $S\subseteq [m]$ as the input and returns the supporting prices $\{\theta_j^{S}(t)\}_{j\in S}$ for $v(t,S)$.
 \end{itemize}
 \end{definition}

 \begin{lemma}\label{lem:truncated still XOS and oracle}
 	Given a collection of thresholds $\{\beta_j\}_{j\in[m]}$ and $c_{\eta}(\bbeta)$. 
 	For any set $S\subseteq[m]$, let $v'(t,S) = v(t,S\cap \CC_{\eta}(t))$. If $v(t,\cdot)$ is an XOS function, $v'(t,\cdot)$ is also an XOS function. Given a value (or demand, XOS) oracle for $v(t,\cdot)$, we can construct in polynomial time a value (or demand, XOS) oracle for $v'(t,\cdot)$.
 \end{lemma}			 
 	\begin{proof}
 		If $v(t,\cdot)$ is an XOS function, $v(t,\cdot)$ can be represented as the max of a collection of additive functions. Observe that if we change the values for items in $\CC_{\eta}(t)$ to $0$ in each of these additive functions, $v'(t,\cdot)$ equals to the max of this new collection of additive functions. Hence, $v'(t,\cdot)$ is also an XOS function.
 		
 		If we are given a value oracle for $v(t,\cdot)$, it is straightforward to construct a value oracle for $v'(t,\cdot)$.
 		 If we are given a demand oracle for $v(t,\cdot)$, here is how to construct a demand oracle for $v'(t,\cdot)$. 
 		 For every queried price vector $\{p_j\}_{j\in[m]}$, we change the price for each item outside $\CC_{\eta}(t)$ to $2v(t,[m])$ and keep the prices for the items in $\CC_{\eta}(t)$. 
 		 Let this new price vector be $p'$. We query the demand oracle of $v(t,\cdot)$ on $p'$. 
 		 The output set should also be the demand set for $v'(t,\cdot)$ under prices $p$, as the bidder can only afford items in $\CC_{\eta}(t)$ and $v'(t,S)=v(t,S)$ for any set $S\subseteq \CC_{\eta}(t)$. 
 		 Finally, we consider the XOS oracle. 
 		 For any set $S$, let $\left\{\theta^{S\cap \CC_\eta(t)}_j(t)\right\}_{j\in {S\cap \CC_\eta(t)}}$ be the supporting prices for 
 		   $v(t,{S\cap \CC_\eta(t)})$. 
 		   Let $\gamma^{S}_j(t)=\theta^{S\cap \CC_\eta(t)}_j$ for all item $j$ in $\CC_\eta(t)\cap S$ and  $\gamma^{S}_j(t)=0$ for all item $j$ in $S-\CC_\eta(t)$. According to the definition of $v'(t,\cdot)$, $\{\gamma_j^S(t)\}_{j\in S}$ is the supporting price for $v'(t,S)$. So given an XOS oracle for $v(t,\cdot)$, we can compute the supporting price of any set $S$ for $v'(t,\cdot)$ in polynomial time.
 	\end{proof}	
 	
 	Lemma~\ref{lem:truncated still XOS and oracle} shows that $v'(t,\cdot)$ is also an XOS function for any type $t$ and with access to a demand oracle for $v(t,\cdot)$  we can construct a demand oracle for $v'(t,\cdot)$ in polynomial time.
 	 So we can indeed run Feige's algorithm on $v'$. In the next Lemma, we show how to learn a collection of prices $\{Q_j\}_{j\in[m]}$ and entry fee function $\delta(\cdot,\cdot)$ such that the corresponding ASPE has high revenue.
 	
\begin{lemma}\label{lem:symmetric XOS learning prices}
Given a collection of $b$-balanced thresholds $\{\beta_j\}_{j\in[m]}$ and $c_\eta(\bbeta)$, and access to value, demand and XOS oracles for valuation $v(t,\cdot)$ for every type $t$, there is a polynomial time algorithm that learns an ASPE$(\{Q_j\}_{j\in[m]},\delta)$ whose revenue is at least  $\frac{\core_{\eta}(\bbeta)}{\KK_1}-g(b,\eta)\cdot \prev-\KK_2\cdot \xi\cdot \opt$ with probability at least $1-\zeta$ using $O\left(n^3(m+n)^2\log \frac{m}{\zeta}\right)$ samples from $\times_{j\in[m]} D_j$, where $\KK_1$ and $\KK_2$ are positive absolute constants, and $g(b,\eta)$ is a function that only depends on $b$ and $\eta$.
\end{lemma}
\begin{proof}
	According to Lemma~\ref{lem:truncated still XOS and oracle}, we can construct value, demand and XOS oracles for valuation $v'(t,\cdot)$ given access to the corresponding oracles for $v(t,\cdot)$. We use $\{\gamma_j^S(t)\}_{j\in S}$ to denote the output of the XOS oracle for $v'(t,\cdot)$ on set $S$. In particular, $\gamma_j^S(t)=0$ for all $j\in S-\CC_\eta(t)$ and $\gamma_j^S(t)=\theta_j^{S\cap\CC_\eta(t)}(t)$ for all $j\in S\cap\CC_\eta(t)$, where $\{\theta_j^{S\cap\CC_\eta(t)}(t)\}_{j\in {S\cap\CC_\eta(t)}}$ is the supporting prices for $v(t,{S\cap\CC_\eta(t)})$. 
	Let $\AA(\boldsymbol{t})$ be the allocation computed by Feige's algorithm on the valuation profile $\left(v'(t_1,\cdot),\ldots,v'(t_n,\cdot)\right)$, where $\AA_i(\boldsymbol{t})$ denotes the set of items that bidder $i$ receives. 
	Let $\sigma^*$ be the interim allocation rule induced by $\AA(\cdot)$ when bidders types are all drawn from $\times_{j\in[m]} D_j$ independently. 
	That is, $\sigma^*_{iS}(t_i)=\Pr_{t_{-i}}\left[\AA_i(\boldsymbol{t})=S\right]$. 
	We use the same definition for $Q_{\eta,j}$ as in Lemma~\ref{lem:Q_j}. 
	In other words, $Q_{\eta,j}$ is the contribution of item $j$ to the social welfare under allocation rule $\sigma^*$, so we can rewrite it as $$\frac{1}{2}\cdot\E_{\boldsymbol{t}}\left[\sum_{i\in[n]} \ind\left[j\in \AA_i(\boldsymbol{t})\right]\cdot \gamma_j^{\AA_i(\boldsymbol{t})}(t_i)\right].$$ 
	
	Let $\boldsymbol{t^{(1)}},\ldots, \boldsymbol{t^{(K)}}$ be $K$ sampled type profiles, and $q^{(\ell)}=\frac{1}{2}\sum_{i\in[n]} \ind\left[j\in \AA_i(\boldsymbol{t^{(\ell)}})\right]\cdot \gamma_j^{\AA_i(\boldsymbol{t^{(\ell)}})}(t^{(\ell)}_i)$. 
	We set $Q_j$ to be $\frac{1}{K}\cdot\sum_{\ell\in[K]} q^{(\ell)}$. Since $\gamma^S_j(t)\leq \beta_j+c_\eta(\bbeta)$ for any $j$, $S$ and $t$, $Q_{j}\leq \beta_j+c_\eta(\bbeta)$. By the Chernoff bound, 
	
	$$\Pr\left[\left|Q_j-Q_{\eta,j}\right|\leq \epsilon\cdot \left(\beta_j+c_\eta(\bbeta)\right) \right]\geq 1-2\exp(-2K\cdot \epsilon^2).$$
	
	As $\{\beta_j\}_{j\in[m]}$ is a collection of $b$-balanced thresholds, we can obtain revenue $\beta_j\cdot \frac{b}{n}$ by only selling item $j$ to one bidder at price $\beta_j$. 
	Hence,  $\beta_j\leq \frac{n\cdot\opt}{b}$. 
	Now, consider a posted price mechanism that sells item $j$ at price $\beta_j+c_\eta(\bbeta)$. 
	A single bidder will purchase at least one item with probability at least $\sum_j \Pr_{t_{j}\sim D_{j}}\left[V(t_{j})\geq \beta_{j}+c_\eta(\boldsymbol{\beta})\right]$ which is no less than  $\frac{1}{2}-\eta$ if $c_\eta(\bbeta)>0$. 
	Hence, the revenue of this mechanism is at least $c_\eta(\bbeta)\cdot \left(\frac{1}{2}-\eta\right)$. 
	As $\eta\leq \frac{1}{4}$, $c_\eta(\bbeta)\leq 4\opt$. 
	If we let $\epsilon = \frac{\xi}{(m+n)\cdot\left(n/b+4\right)}$ for some small constant $\xi$ which will be specified later and $K =\frac{\log \frac{4m}{\zeta}}{2\epsilon^2}$, we have $\Pr\left[\left|Q_j-Q_{\eta,j}\right|\leq \frac{\xi}{m+n}\cdot\opt \right]\geq 1-\frac{\zeta}{2m}$. 
	In other words, with $O\left(n^3(m+n)^2\log \frac{m}{\zeta}\right)$ samples from $\times_{j\in[m]} D_j$ (as each $\boldsymbol{t}^{\ell}$ costs $n$ samples), we can learn in polynomial time a collection of prices $\{Q_j\}_{j\in [m]}$ such that $\left|Q_j-Q_{\eta,j}\right|\leq \frac{\xi}{m+n}\cdot\opt $ for all item $j$ with probability $1-\zeta/2$.
	
	Next, we consider the entry fee function. We use essentially the same argument as in Lemma~\ref{lem:learn median}. Suppose we take
$L$ samples $t^{(1)},\cdots, t^{(L)}$ from $\times_{j\in[m]} D_j$. Define the entry fee $\delta(S)$
 for set $S$ under $\{Q_j\}_{j\in[m]}$ to be the median of $u(t^{(1)},S),\cdots, u(t^{(L)},S)$, where $u(t,S)=\max_{S*\subseteq S} v(t,S^*)-\sum_{j\in S^*} p_j$. Given any constant $\mu\in[0,1/4]$, for any fixed set $S$, it is easy to argue that the probability for $\Pr_{t\sim \times_{j\in[m]} D_j}[u(t,S)\geq \delta(S)]$ to be larger than  $\frac{1}{2}+\mu$ or less than  $\frac{1}{2}-\mu$ is at most $2\exp(-2L\mu^2)$ due to the Chernoff bound. If we let $L$ to be $a\cdot \frac{m+\log 1/\zeta}{\mu^2}$ for a sufficiently large constant $a$, the probability that $\delta(\cdot)$ is a $\mu$-balanced entry fee function is at least $1-\zeta/2$ by the union bound.
 
 Hence, with $O\left(n^3(m+n)^2\log \frac{m}{\zeta}\right)$ samples from  $\times_{j\in[m]} D_j$, we can compute in polynomial time a collection of prices $\{Q_j\}_{j\in[m]}$ and a entry fee function $\delta(\cdot)$ such that the revenue of the ASPE$\left(\{Q_j\}_{j\in[m]},\delta(\cdot)\right)$ is at least $\frac{(1-1/e)\cdot\core_{\eta}(\bbeta)}{\BB_{1}(\mu)}-\BB_{2}(b,\eta,\mu)\cdot \prev-\xi\cdot\BB_3(\mu)\cdot \opt$ with probability $1-\zeta$ due to Lemma~\ref{lem:symmetric approx ASPE}. Our claim follows by fixing the value of $\mu$ to be some constant.
	\end{proof}

\begin{theorem}\label{thm:symmetric XOS}
	For symmetric bidders with valuations that are XOS over independent items, \begin{enumerate}
		\item when $V(t_j)$ is upper bounded by $H$ for any $j\in[m]$ and any $t_j$, with $$O\left(\left(n^5+m^2n^4\right)\cdot\log\frac{m}{\delta}+\left(\frac{1}{\epsilon}\right)^2 \left(m^2 n\log n\log \frac{1}{\epsilon} + \log \frac{1}{\delta}\right)\right)$$ samples from $\times_{j\in[m]} D_j$, we can learn in polynomial time with probability $1-\delta$ a mechanism whose revenue is at least $c_1\cdot \opt-\epsilon\cdot H$ for some absolute constant $c_1$;
		\item when the distribution of random variable $V(t_j)$ with $t_j\sim D_j$ is regular for all item $j\in[m]$, with 
		$$O\left(n^5\cdot\log\frac{m}{\delta}+\max\{m,n\}^2m^2 n^2\cdot  \log \frac{nm}{\delta}\right)$$ samples from $\times_{j\in[m]} D_j$, we can learn in polynomial time with probability $1-\delta$ a mechanism whose revenue is at least $c_2\cdot \opt$ for some absolute constant $c_2$.
	\end{enumerate}
\end{theorem}
\begin{prevproof}{Theorem}{thm:symmetric XOS}
Combining Lemma~\ref{lem:symmetric XOS learn beta and c}, Lemma~\ref{lem:symmetric XOS learning prices} and Theorem~\ref{thm:UB subadditive}, we know how to compute in polynomial time an ASPE whose revenue is at least $a_1\cdot \opt-a_2\cdot\prev$ with probability $1-\delta/2$ for some absolute constant $a_1$, $a_2$, and we only need $O\left(\left(n^5+m^2n^4\right)\cdot\log\frac{m}{\delta}\right)$ samples from $\times_{j\in[m]} D_j$. When the distributions are bounded, we can learn in polynomial time  an RSPM whose revenue is at least $\frac{\prev}{144}-\xi H$ with probability $1-\delta/2$ using $O\left(\left(\frac{1}{\xi}\right)^2 \left(m^2 n\log n\log \frac{1}{\epsilon} + \log \frac{1}{\delta}\right)\right)$ samples (Theorem~\ref{thm:UD bounded}). By choosing the ratio between $\xi$ and $\epsilon$ to be the right constant, we can show the first part of our claim. When $V(t_j)$ is a regular random variable for every item $j$, we can learn in polynomial time  an RSPM whose revenue is at least $\frac{\prev}{33}$ with probability $1-\delta/2$ using  $O\left(\max\{m,n\}^2m^2 n^2\cdot \log \frac{nm}{\delta}\right)$ samples (Theorem~\ref{thm:UD regular}). Therefore, we can learn a mechanism in polynomial time such that with probability $1-\delta$ whose revenue is at least a constant fraction of the $\opt$. This proves the second part of our claim.\end{prevproof}

\subsection{Symmetric Bidders with Subadditive Valuations}\label{sec:symmetric subadditive}
In this section, we argue that if the bidders are symmetric and $m=O(n)$, there exists a collection of $b$-balanced thresholds $\{\beta_j\}_{j\in[m]}$ for a fixed constant $b$, such that $\prev$ is within a constant fraction of the benchmark. Note that this argument only applies to symmetric bidders, as $b$-balanced thresholds may not even exist for asymmetric bidders.  

We set $\eta=0$ for this section and drop the subscript $\eta$ when there is no confusion. 
 We show how to upper bound $\core(\bbeta)$ with $\prev$ by choosing a particular collection of $b$-balanced thresholds. Let $Z=\max\{m,n\}$ and $b = \frac{n}{3Z}$. It is not hard to see that $\frac{n}{3Z}$-balanced thresholds exist, as we can choose $\beta_j$ such that $\Pr_{t_j\sim D_j}[V(t_j)\geq \beta_j]=\frac{1}{3Z}$. 

\begin{lemma}\label{lem:bounding core with beta}
	Let $\{\beta_j\}_{j\in[m]}$ be a collection of $\frac{n}{3Z}$-balanced thresholds, then $\core(\bbeta)\leq \sum_{j\in[m]} \beta_j$.
	\end{lemma}
\begin{proof}

As $\{\beta_j\}_{j\in[m]}$ are $\frac{n}{3Z}$-balanced, $\Pr_{t_{ij}\sim D_j}[V(t_{ij})\geq \beta_j]\leq \frac{n}{(n-1)\cdot 3Z}\leq \frac{1}{2Z}$. Therefore, $$\sum_{j\in[m]} \Pr_{t_{j}\sim D_{j}}\left[V(t_{j})\geq \beta_{j}\right]\leq \frac{1}{2},$$ so $c(\bbeta) = 0$. Next, we upper bound $\core(\bbeta)$ by $\sum_{j\in[m]} \beta_j$.

	\begin{align*}
		\core(\bbeta) = &\max_{\sigma \in P(D)} \sum_{i\in[n]}\sum_{t_i\in T_i}f(t_i)\cdot \sum_{S\subseteq[m]}\sigma_{iS}(t_i)\cdot v(t_i,S\cap \CC(t_i))\\
		\leq &\max_{\sigma \in P(D)} \sum_{i\in[n]}\sum_{t_i\in T_i}f(t_i)\cdot \left(\sum_{S\subseteq[m]}\sigma_{iS}(t_i)\cdot \sum_{j\in S} \beta_j\right)\\
		= &\max_{\sigma \in P(D)} \sum_{i\in[n]} \sum_{j\in[m]} \beta_j\cdot \left(\sum_{t_i\in T_i}f(t_i)\cdot \sum_{S: j\in S}\sigma_{iS}(t_i)\right)\\
		= &\max_{\sigma \in P(D)} \sum_j \beta_j\cdot  \left(\sum_{i\in[n]} \sum_{t_i\in T_i}f_i(t_i)\cdot \sum_{S: j\in S}\sigma_{iS}(t_i)\right)\\
		\leq & \sum_j \beta_j
	\end{align*}
	
	The first inequality is because $v(t_i,\cdot)$ is a subadditive function, so $$v(t_i,S\cap \mathcal{C}_i(t_i))\leq \sum_{j\in S\cap \mathcal{C}_i(t_i)} V(t_{ij})\leq \sum_{j\in S\cap \mathcal{C}_i(t_i)} \beta_j \leq \sum_{j\in S} \beta_j.$$ The last inequality is because $\sum_{i\in[n]}\sum_{t_i\in T_i}f_i(t_i)\cdot \sum_{S: j\in S}\sigma_{iS}(t_i)\leq 1$ is the ex-ante probability for bidder $i$ to receive item $j$, and for any feasible interim allocation $\sigma$, the sum of all bidders' ex-ante probabilities for receiving item $j$ should not exceed $1$.
\end{proof}

In the following Lemma, we demonstrate that $\sum_{j\in [m]} \beta_j$ is upper bounded by $\frac{9Z}{n}\cdot \prev$.

\begin{lemma}\label{lem:bounding beta with prev}
		Let $\{\beta_j\}_{j\in[m]}$ be a collection of $\frac{n}{3Z}$-balanced thresholds, $\prev\geq \frac{n}{9Z}\cdot \sum_{j\in [m]} \beta_j$.
\end{lemma}

\begin{proof}
	Let us consider an RSPM where the price for selling item $j$ to bidder $i$ is $\beta_j$. Bidder $i$ purchases item $j$ if that is the only item she can afford and no one else can afford item $j$. As $\{\beta_j\}_{j\in[m]}$ are $\frac{n}{3Z}$-balanced, the probability that no one else can afford item $j$ is at least $$\left(1-\sum_{k\neq i}\Pr_{t_{kj}\sim D_j}[V(t_{kj})\geq \beta_j]\right)\geq (1-\frac{n}{3Z})\geq \frac{2}{3}.$$ Also, the probability that $i$ cannot afford any item other than $j$ is at least $$\left(1-\sum_{\ell\neq j}\Pr_{t_{i\ell}\sim D_\ell}[V(t_{i\ell})\geq \beta_\ell]\right)\geq 1-\frac{n(m-1)}{3Z(n-1)}\geq\frac{1}{2}.$$ Therefore, bidder $i$ purchases item $j$ with probability at least $\frac{1}{3}\Pr_{t_{ij}\sim D_j}[V(t_{ij}\geq \beta_j)]\geq \frac{1}{9Z}$. Whenever this event happens, it contributes $\beta_j$ to the revenue. So the total revenue is at least $\sum_j\sum_i\frac{\beta_j}{9Z}= \frac{n}{9Z}\cdot \sum_j \beta_j$.
\end{proof}

Combining Theorem~\ref{thm:UB subadditive}, Lemma~\ref{lem:bounding core with beta} and~\ref{lem:bounding beta with prev}, we obtain the following Theorem.

\begin{theorem}\label{thm:subbadditive UB RSPM}
	For symmetric bidders with valuations that are subadditive over independent items, $$\opt\leq \left(24+\frac{36 \max\{n,m\}}{n}\right)\cdot \prev.$$ \end{theorem}
\begin{proof}
	Combining Lemma~\ref{lem:bounding core with beta} and~\ref{lem:bounding beta with prev}, we have $\prev\geq \frac{n}{9\max\{n,m\}}\cdot \core(\bbeta)$ if $\{\beta_j\}_{j\in[m]}$ is a collection of $\frac{n}{3\max\{n,m\}}$-balanced thresholds. By setting $b$ to be $\frac{n}{3\max\{n,m\}}$ and replacing $\core(\bbeta)$ with $\frac{9\max\{n,m\}}{n}\cdot \prev$ in Theorem~\ref{thm:UB subadditive}, we have $$\opt\leq \left(12+\frac{8}{1-\frac{n}{3\max\{n,m\}}}+\frac{36 \max\{n,m\}}{n}\right)\cdot \prev.$$ As $\frac{n}{3\max\{n,m\}}\leq 1/3$, $$\opt\leq \left(24+\frac{36 \max\{n,m\}}{n}\right)\cdot \prev.$$
\end{proof}


\subsubsection{Learning an Approximately Optimal Mechanism for Symmetric Subadditive Bidders}

With Theorem~\ref{thm:subbadditive UB RSPM}, we only need to learn a mechanism that approximates the optimal revenue obtainable by any RSPM. The next Lemma connects RSPMs with SPMs in an induced unit-demand setting.
\begin{lemma}\label{lem: RSPM to UD}
	Consider $n$ symmetric bidders whose types are drawn independently from $\times_{j=1}^m D_j$. Let $\FF_j$ be the distribution for random variable $V(t_j)$ where $t_j\sim D_j$. We define an \emph{induced unit-demand} setting with  $n$ symmetric unit-demand bidders whose values for item $j$ are drawn independently from $\FF_j$. For any collection of prices $\{p_{ij}\}_{i\in[n],j\in[m]}$, the revenue of the RSPM with these prices in the original setting is exactly the same as the revenue of the SPM with these prices in the induced unit-demand setting.
\end{lemma}
\begin{proof}
As in an RSPM bidders can purchase at most one item, bidders behave exactly the same as in the induced unit-demand setting. Since the prices in the SPM and RSPM are the same, bidders purchase exactly the same items. Hence, the revenue is the same.\end{proof}

\begin{corollary}\label{cor:subadditive to UD}
For symmetric bidders with valuations that are subadditive over independent items, $$\opt^{UD}\geq \Omega\left(\frac{n}{Z}\right)\cdot\opt, $$ where $Z=\max\{m,n\}$ and $\opt^{UD}$ is the optimal revenue for the induced unit-demand setting.\end{corollary}
\begin{proof}
	Combine Theorem~\ref{thm:subbadditive UB RSPM} and Lemma~\ref{lem: RSPM to UD}.
\end{proof}

Lemma~\ref{lem: RSPM to UD} implies that learning an approximately optimal RSPM is equivalent as learning an approximately optimal SPM in the induced unit-demand setting. Next, we apply our results in Section~\ref{sec:unit-demand} to the induced unit-demand setting to learn an RSMP that approximates the optimal revenue in the original setting.


In the next Theorem, we show that even though the bidders' valuations could be complex set functions, e.g., submodular, XOS and subadditive, as long as $m=O(n)$, the approximate distributions for the bidders' values for winning any \emph{single item} provides sufficient information to learn an approximately optimal mechanism. 

\begin{theorem} \label{thm:subadditive Kolmogorov}
	For symmetric bidders with valuations that are subadditive over independent items, let $\FF_j$ be the distribution of $V(t_j)$ where $t_j\sim D_j$. If $\FF_j$ is supported on $[0,H]$ for all $j\in[m]$, given distributions $\hat{\FF}_{j}$ where $\left|\left|\hat{\FF}_{j}-\FF_{j}\right|\right|_K\leq \epsilon$ for all $j\in[m]$, there is a polynomial time algorithm that constructs a randomized RSPM whose revenue under the true distribution $D$ is at least  $$\left(\frac{1}{4}-(n+m)\cdot \epsilon\right)\cdot\left(\Omega\left(\frac{n}{\max\{m,n\}}\right)\cdot \opt-2\epsilon\cdot mnH\right).$$
\end{theorem}
\begin{proof}
 Let $Z=\max\{m,n\}$. According to Corollary~\ref{cor:subadditive to UD}, $\opt^{UD}=\Omega\left(\frac{n}{Z}\right)\cdot \opt$. Since $||\hat{\FF}_{j}-\FF_{j}||_K\leq \epsilon$ for all $j\in[m]$, we can learn a randomized SPM in the induced unit-demand setting whose revenue under the true distribution is at least $\left(\frac{1}{4}-(n+m)\cdot \epsilon\right)\cdot\left(\frac{\opt^{UD}}{8}-2\epsilon\cdot mnH\right)$ based on Theorem~\ref{thm:UD Kolmogorov}.  By Lemma~\ref{lem: RSPM to UD},  we can construct an RSPM with the same collection of (randomized) prices and achieve revenue $$\left(\frac{1}{4}-(n+m)\cdot \epsilon\right)\cdot\left(\Omega\left(\frac{n}{Z}\right)\cdot \opt-2\epsilon\cdot mnH\right)$$ in the original setting.\end{proof}

If we are given sample access to bounded distributions, we show in the following Theorem that a polynomial number of samples suffices to learn an approximately optimal mechanism, when $m=O(n)$.

\begin{theorem}\label{thm:subadditive bounded}
	For symmetric bidders with valuations that are subadditive over independent items, let $\FF_j$ be the distribution of ~$V(t_j)$ where $t_j\sim D_j$. If $\FF_j$ is supported on $[0,H]$ for all $j\in[m]$, there is a polynomial time algorithm that learns  an RSPM whose revenue is $\Omega\left(\frac{n}{{\max\{m,n\}}}\right)\cdot \opt -\epsilon H$ with probability $1-\delta$ using $$O\left(\left(\frac{1}{\epsilon}\right)^2\cdot \left(m^2 n\log n\log \frac{1}{\epsilon} + \log \frac{1}{\delta}\right)\right)$$ samples.
\end{theorem}
\begin{proof}
	According to Corollary~\ref{cor:subadditive to UD}, $\opt^{UD}=\Omega\left(\frac{n}{{\max\{m,n\}}}\right)\cdot \opt$. Due to Theorem~\ref{thm:UD bounded}, $$O\left(\left(\frac{1}{\epsilon}\right)^2\cdot \left(m^2 n\log n\log \frac{1}{\epsilon} + \log \frac{1}{\delta}\right)\right)$$ samples suffices to learn in polynomial time with probability $1-\delta$ an SPM with revenue at least $\Omega(\opt^{UD})-\epsilon\cdot H$ for the induced unit-demand setting. By Lemma~\ref{lem: RSPM to UD},  we can construct an RSPM with the same collection of prices and achieve revenue $\Omega\left(\frac{n}{{\max\{m,n\}}}\right)\cdot \opt -\epsilon H$ in the original setting. \end{proof}
	
Finally, if the distribution of the random variable $V(t_j)$ with $t_j\sim D_j$ is regular for all item $j\in[m]$, we prove in the next theorem that there exists a prior-independent mechanism that achieves a constant fraction of the optimal revenue if $m=O(n)$. Note that approximately optimal prior-independent mechanisms for symmetric unit-demand bidders are known due to the work by Devanur et al.~\cite{DevanurHKN11} and Roughgarden et al.~\cite{RoughgardenTY12}. Our result is obtained by combining Theorem~\ref{thm:subbadditive UB RSPM} and the afore-mentioned prior independent mechanisms.

\begin{theorem}\label{thm:subadditive regular}
	For symmetric bidders with valuations that are subadditive over independent items, let $\FF_j$ be the distribution of ~$V(t_j)$ where $t_j\sim D_j$. If $\FF_j$ is regular for all $j\in[m]$, there is a prior-independent mechanism with revenue at least $\Omega\left(\frac{n}{{\max\{m,n\}}}\right)\cdot \opt$. Moreover, the mechanism can be implemented efficiently.
\end{theorem}
\begin{proof}
	The mechanism in~\cite{DevanurHKN11} or~\cite{RoughgardenTY12} provides an approximately optimal prior-independent mechanism in the induced unit-demand setting. Let us use $M$ to denote this mechanism. Suppose we restrict every bidder to purchase at most one item in the original setting and then run mechanism $M$. The expected revenue is the same as $M$'s expected revenue in the induced setting. Since $M$'s expected revenue is  $\Omega(\opt^{UD})$ and $\opt^{UD}=\Omega\left(\frac{n}{{\max\{m,n\}}}\right)\cdot \opt$, the mechanism we constructed has revenue $\Omega\left(\frac{n}{{\max\{m,n\}}}\right)\cdot \opt$. Since $M$ can be implemented efficiently for unit-demand bidders, our mechanism can also be implemented efficiently.\end{proof}

\newpage
\bibliographystyle{plain}	
\bibliography{Yang}

\begin{thebibliography}{10}

\bibitem{Alaei11}
Saeed Alaei.
\newblock {Bayesian Combinatorial Auctions: Expanding Single Buyer Mechanisms
  to Many Buyers}.
\newblock In {\em the 52nd Annual IEEE Symposium on Foundations of Computer
  Science (FOCS)}, 2011.

\bibitem{AlaeiFHHM13}
Saeed Alaei, Hu~Fu, Nima Haghpanah, and Jason Hartline.
\newblock {The Simple Economics of Approximately Optimal Auctions}.
\newblock In {\em the 54th Annual IEEE Symposium on Foundations of Computer
  Science (FOCS)}, 2013.

\bibitem{AlaeiFHHM12}
Saeed Alaei, Hu~Fu, Nima Haghpanah, Jason Hartline, and Azarakhsh Malekian.
\newblock {Bayesian Optimal Auctions via Multi- to Single-agent Reduction}.
\newblock In {\em the 13th ACM Conference on Electronic Commerce (EC)}, 2012.

\bibitem{AtheyH2007nonparametric}
Susan Athey and Philip Haile.
\newblock Nonparametric approaches to auctions.
\newblock {\em Handbook of econometrics}, 6(Part A):3847--3965, 2007.

\bibitem{BabaioffILW14}
Moshe Babaioff, Nicole Immorlica, Brendan Lucier, and S.~Matthew Weinberg.
\newblock {A Simple and Approximately Optimal Mechanism for an Additive Buyer}.
\newblock In {\em the 55th Annual IEEE Symposium on Foundations of Computer
  Science (FOCS)}, 2014.

\bibitem{BhalgatGM13}
Anand Bhalgat, Sreenivas Gollapudi, and Kamesh Munagala.
\newblock Optimal auctions via the multiplicative weight method.
\newblock In {\em {ACM} Conference on Electronic Commerce, {EC} '13,
  Philadelphia, PA, USA, June 16-20, 2013}, pages 73--90, 2013.

\bibitem{CaiD11b}
Yang Cai and Constantinos Daskalakis.
\newblock {Extreme-Value Theorems for Optimal Multidimensional Pricing}.
\newblock In {\em the 52nd Annual IEEE Symposium on Foundations of Computer
  Science (FOCS)}, 2011.

\bibitem{CaiDW12a}
Yang Cai, Constantinos Daskalakis, and S.~Matthew Weinberg.
\newblock {An Algorithmic Characterization of Multi-Dimensional Mechanisms}.
\newblock In {\em the 44th Annual ACM Symposium on Theory of Computing (STOC)},
  2012.

\bibitem{CaiDW12b}
Yang Cai, Constantinos Daskalakis, and S.~Matthew Weinberg.
\newblock {Optimal Multi-Dimensional Mechanism Design: Reducing Revenue to
  Welfare Maximization}.
\newblock In {\em the 53rd Annual IEEE Symposium on Foundations of Computer
  Science (FOCS)}, 2012.

\bibitem{CaiDW13b}
Yang Cai, Constantinos Daskalakis, and S.~Matthew Weinberg.
\newblock {Understanding Incentives: Mechanism Design becomes Algorithm
  Design}.
\newblock In {\em the 54th Annual IEEE Symposium on Foundations of Computer
  Science (FOCS)}, 2013.

\bibitem{CaiDW16}
Yang Cai, Nikhil~R. Devanur, and S.~Matthew Weinberg.
\newblock A duality based unified approach to bayesian mechanism design.
\newblock In {\em the 48th Annual ACM Symposium on Theory of Computing (STOC)},
  2016.

\bibitem{CaiH13}
Yang Cai and Zhiyi Huang.
\newblock {Simple and Nearly Optimal Multi-Item Auctions}.
\newblock In {\em the 24th Annual ACM-SIAM Symposium on Discrete Algorithms
  (SODA)}, 2013.

\bibitem{CaiZ17}
Yang Cai and Mingfei Zhao.
\newblock Simple mechanisms for subadditive buyers via duality.
\newblock In {\em the 49th Annual ACM Symposium on Theory of Computing (STOC)},
  2017.

\bibitem{ChawlaHK07}
Shuchi Chawla, Jason~D. Hartline, and Robert~D. Kleinberg.
\newblock {Algorithmic Pricing via Virtual Valuations}.
\newblock In {\em the 8th ACM Conference on Electronic Commerce (EC)}, 2007.

\bibitem{ChawlaHMS10}
Shuchi Chawla, Jason~D. Hartline, David~L. Malec, and Balasubramanian Sivan.
\newblock {Multi-Parameter Mechanism Design and Sequential Posted Pricing}.
\newblock In {\em the 42nd ACM Symposium on Theory of Computing (STOC)}, 2010.

\bibitem{ChawlaM16}
Shuchi Chawla and J.~Benjamin Miller.
\newblock Mechanism design for subadditive agents via an ex-ante relaxation.
\newblock In {\em Proceedings of the 2016 {ACM} Conference on Economics and
  Computation, {EC} '16, Maastricht, The Netherlands, July 24-28, 2016}, pages
  579--596, 2016.

\bibitem{ColeR14}
Richard Cole and Tim Roughgarden.
\newblock The sample complexity of revenue maximization.
\newblock In {\em Proceedings of the 46th Annual ACM Symposium on Theory of
  Computing}, 2014.

\bibitem{Daskalakis15}
Constantinos Daskalakis.
\newblock Multi-item auctions defying intuition?
\newblock {\em ACM SIGecom Exchanges}, 14(1):41--75, 2015.

\bibitem{DaskalakisDT17}
Constantinos Daskalakis, Alan Deckelbaum, and Christos Tzamos.
\newblock Strong duality for a multiple-good monopolist.
\newblock {\em Econometrica}, 2017.

\bibitem{DaskalakisDW15}
Constantinos Daskalakis, Nikhil~R. Devanur, and S.~Matthew Weinberg.
\newblock Revenue maximization and ex-post budget constraints.
\newblock In {\em Proceedings of the Sixteenth {ACM} Conference on Economics
  and Computation, {EC} '15, Portland, OR, USA, June 15-19, 2015}, pages
  433--447, 2015.

\bibitem{DevanurHKN11}
Nikhil~R. Devanur, Jason~D. Hartline, Anna~R. Karlin, and C.~Thach Nguyen.
\newblock Prior-independent multi-parameter mechanism design.
\newblock In {\em Internet and Network Economics - 7th International Workshop,
  {WINE} 2011, Singapore, December 11-14, 2011. Proceedings}, pages 122--133,
  2011.

\bibitem{DevanurHP16}
NR~Devanur, Z~Huang, and CA~Psomas.
\newblock The sample complexity of auctions with side information.
\newblock In {\em 48th Annual ACM SIGACT Symposium on Theory of Computing},
  2016.

\bibitem{DobzinskiNS05}
Shahar Dobzinski, Noam Nisan, and Michael Schapira.
\newblock Approximation algorithms for combinatorial auctions with
  complement-free bidders.
\newblock In {\em STOC}, pages 610--618, 2005.

\bibitem{DughmiHN14}
Shaddin Dughmi, Li~Han, and Noam Nisan.
\newblock Sampling and representation complexity of revenue maximization.
\newblock In {\em International Conference on Web and Internet Economics
  (WINE)}, 2014.

\bibitem{DvoretzkyKW56}
Aryeh Dvoretzky, Jack Kiefer, and Jacob Wolfowitz.
\newblock Asymptotic minimax character of the sample distribution function and
  of the classical multinomial estimator.
\newblock {\em The Annals of Mathematical Statistics}, pages 642--669, 1956.

\bibitem{Elkind07}
Edith Elkind.
\newblock Designing and learning optimal finite support auctions.
\newblock In {\em Proceedings of the eighteenth annual ACM-SIAM symposium on
  Discrete algorithms}, pages 736--745. Society for Industrial and Applied
  Mathematics, 2007.

\bibitem{Feige09}
Uriel Feige.
\newblock On maximizing welfare when utility functions are subadditive.
\newblock {\em SIAM Journal on Computing}, 39(1):122--142, 2009.

\bibitem{GoldnerK16}
Kira Goldner and Anna~R. Karlin.
\newblock A prior-independent revenue-maximizing auction for multiple additive
  bidders.
\newblock In {\em Web and Internet Economics - 12th International Conference,
  {WINE} 2016, Montreal, Canada, December 11-14, 2016, Proceedings}, pages
  160--173, 2016.

\bibitem{GonczarowskiN16}
Yannai~A. Gonczarowski and Noam Nisan.
\newblock Efficient empirical revenue maximization in single-parameter auction
  environments.
\newblock {\em CoRR}, abs/1610.09976, 2016.

\bibitem{GuerrePV00}
Emmanuel Guerre, Isabelle Perrigne, and Quang Vuong.
\newblock Optimal nonparametric estimation of first-price auctions.
\newblock {\em Econometrica}, 68(3):525--574, 2000.

\bibitem{HuangMR15}
Zhiyi Huang, Yishay Mansour, and Tim Roughgarden.
\newblock Making the most of your samples.
\newblock In {\em Proceedings of the Sixteenth ACM Conference on Economics and
  Computation}, 2015.

\bibitem{KleinbergW12}
Robert Kleinberg and S.~Matthew Weinberg.
\newblock {Matroid Prophet Inequalities}.
\newblock In {\em the 44th Annual ACM Symposium on Theory of Computing (STOC)},
  2012.

\bibitem{MohriM14}
Mehryar Mohri and Andres~Munoz Medina.
\newblock Learning theory and algorithms for revenue optimization in second
  price auctions with reserve.
\newblock In {\em ICML}, pages 262--270, 2014.

\bibitem{MorgensternR15}
Jamie Morgenstern and Tim Roughgarden.
\newblock The pseudo-dimension of near-optimal auctions.
\newblock {\em arXiv preprint arXiv:1506.03684}, 2015.

\bibitem{MorgensternR16}
Jamie Morgenstern and Tim Roughgarden.
\newblock Learning simple auctions.
\newblock In {\em Proceedings of the 29th Conference on Learning Theory, {COLT}
  2016, New York, USA, June 23-26, 2016}, pages 1298--1318, 2016.

\bibitem{Myerson81}
Roger~B. Myerson.
\newblock {Optimal Auction Design}.
\newblock {\em Mathematics of Operations Research}, 6(1):58--73, 1981.

\bibitem{PaarschH06}
Harry~J Paarsch and Han Hong.
\newblock An introduction to the structural econometrics of auction data.
\newblock {\em MIT Press Books}, 1, 2006.

\bibitem{RoughgardenS16}
Tim Roughgarden and Okke Schrijvers.
\newblock Ironing in the dark.
\newblock In {\em Proceedings of the 2016 ACM Conference on Economics and
  Computation}, pages 1--18. ACM, 2016.

\bibitem{RoughgardenTY12}
Tim Roughgarden, Inbal Talgam-Cohen, and Qiqi Yan.
\newblock Supply-limiting mechanisms.
\newblock In {\em 13th ACM Conference on Electronic Commerce (EC)}, 2012.

\bibitem{RubinsteinW15}
Aviad Rubinstein and S.~Matthew Weinberg.
\newblock Simple mechanisms for a subadditive buyer and applications to revenue
  monotonicity.
\newblock In {\em Proceedings of the Sixteenth {ACM} Conference on Economics
  and Computation, {EC} '15, Portland, OR, USA, June 15-19, 2015}, pages
  377--394, 2015.

\bibitem{Samuel-cahn84}
Ester Samuel-Cahn.
\newblock Comparison of threshold stop rules and maximum for independent
  nonnegative random variables.
\newblock {\em Ann. Probab.}, 12(4):1213--1216, 11 1984.

\bibitem{Yao15}
Andrew~Chi{-}Chih Yao.
\newblock An n-to-1 bidder reduction for multi-item auctions and its
  applications.
\newblock In {\em SODA}, 2015.

\end{thebibliography}
\end{document}